\documentclass[sigconf]{acmart}
\usepackage[a-2b]{pdfx} 
\usepackage[capitalize,noabbrev]{cleveref}
\usepackage[colorinlistoftodos]{todonotes}
\usepackage[font=small]{caption}
\usepackage[font=small]{subcaption}
\usepackage[framemethod=TikZ]{mdframed}
\usepackage[linesnumbered,ruled,vlined]{algorithm2e}
\usepackage{adjustbox}
\usepackage{soul}
\usepackage[title]{appendix}
\usepackage{amsfonts}
\usepackage{amsmath}
\usepackage{array}
\usepackage{appendix}
\usepackage{booktabs}
\usepackage{colortbl} 
\usepackage{comment}
\usepackage{diagbox}
\usepackage{dsfont}
\usepackage{graphicx}
\usepackage{lscape}
\usepackage{listings}
\usepackage{mathtools}
\usepackage{multirow}
\usepackage{paralist}
\usepackage{pgf}
\usepackage{pgfplots}
\usepackage{pgfplotstable}
\pgfplotsset{compat=1.17} % adjust version as needed
\usepgfplotslibrary{groupplots}
\usepackage{pifont}
\usepackage{relsize}
\usepackage{rotating}
\usepackage{scalerel}
\usepackage{setspace}
\usepackage{standalone}
\usepackage{tcolorbox}
\usepackage{tikz}
\usepackage{wrapfig}
\usepackage{xcolor}
\usepackage{xspace}
\usetikzlibrary{patterns}
\usepackage{enumitem}

\usepackage[ruled,vlined]{algorithm2e}

\SetKwInOut{Input}{input}
\SetKwInOut{Output}{output}

\AtBeginDocument{%
  }

\setcopyright{acmlicensed}
\copyrightyear{2026}
\acmYear{2026}
\acmDOI{XXXXXXX.XXXXXXX}

\acmConference[SIGMOD '26]{International Conference on Management of Data}{May 31--June 5, 2026}{Bengaluru, India}
\acmISBN{978-1-4503-XXXX-X/2018/06}

\Crefname{algocf}{Algorithm}{Algorithms}

\let\oldnl\nl
\newcommand{\nonl}{\renewcommand{\nl}{\let\nl\oldnl}}

\SetCommentSty{mycommfont}
\def\HiLi{\leavevmode\rlap{\hbox to \hsize{\color{red!20}\leaders\hrule height .8\baselineskip depth .5ex\hfill}}}

\DeclareMathOperator*{\argmax}{arg\,max}
\newtheorem{problem}{Problem}
\newtheorem{definition}{Definition}[section]

\newtheorem{proposition}{Proposition}[section]
\newtheorem{example}{Example}[section]

\newcommand{\brit}[1]{\textcolor{teal}{\bf (Brit)~[#1]}{\typeout{#1}}}

\newcommand{\newtext}[1]{\textcolor{black}{#1}{\typeout{#1}}}

\newcommand{\ignore}[1]{}

\newcommand{\cut}[1]{}

\newcommand{\drawbar}[5]{
#1#3 
\resizebox{.08\textwidth}{!}{
\raisebox{-0.2\height}{
\begin{tikzpicture}
\draw[fill=#4!70] (0,0) rectangle (#1, -20); 
\draw[fill=#4!10] (#1, 0) rectangle (100,-20);
\end{tikzpicture}}}
& #2#3 #5
}

\newcommand{\probName}{\textsc{Causal Repair Explanation}\xspace}

\newcommand{\sysName}{\textsc{ExDis}\xspace}
\newcommand{\attrset}{\ensuremath{\mathbb{A}}\xspace}

\newcommand{\dom}{{\tt dom}\xspace}

\newcommand{\db}{\ensuremath{D}\xspace}

\newcommand{\pattern}{\ensuremath{\mathcal{\psi}}\xspace}
\newcommand{\model}{\ensuremath{\mathcal{G}_\db}\xspace}

\newcommand{\mutable}{{\ensuremath{\mathbf M}}\xspace}
\newcommand{\immutable}{{\ensuremath{\mathbf I}\xspace}}
\newcommand{\fact}{disparity explanation\xspace}
\newcommand{\facts}{disparity explanations\xspace}
\newcommand{\disp}{{\ensuremath{\Delta}\xspace}}

\newcommand{\expSim}{\textsc{sim}\xspace}
\newcommand{\score}{disparity score\xspace}

\definecolor{darkgreen}{rgb}{0.0, 0.5, 0.0}
\newcommand{\increase}[1]{
  \ensuremath{
    \tikz[baseline=(n.base)]{
    \node[inner sep=0pt] (n) {\text{\small #1}};
    \draw[->, thick, ForestGreen] 
      ([yshift=-0.6ex]n.south west) -- ([yshift=-0.6ex]n.south east);
  }
  }
}

\newcommand{\decrease}[1]{
  \ensuremath{
    \tikz[baseline=(n.base)]{
    \node[inner sep=0pt] (n) {\text{\small #1}};
    \draw[<-, thick, Red] 
      ([yshift=-0.6ex]n.south west) -- ([yshift=-0.6ex]n.south east);
  }
  }
}
\newcommand{\uparrowbold}{\textcolor{ForestGreen}{$\boldsymbol{\uparrow}$}}
\newcommand{\downarrowbold}{\textcolor{red}{$\boldsymbol{\downarrow}$}}

\tcbuselibrary{breakable}

\definecolor{darkgreen}{RGB}{0, 200, 0}
\definecolor{experimentblue}{RGB}{127, 159, 186}
\definecolor{experimentgreen}{RGB}{127, 186, 130}
\definecolor{experimentred}{RGB}{186, 138, 127}
\definecolor{experimentgray}{RGB}{211,211,211}

\newtcolorbox[auto counter,number within=subsection]{ruleset}[1]{
  breakable, boxrule=0pt, boxsep=3pt,left=0pt,right=0pt,top=0pt,bottom=0pt, colbacktitle=experimentblue, title=#1:,
}

\definecolor{moonstoneblue}{rgb}{0.45, 0.66, 0.76}
\definecolor{oldlace}{rgb}{0.99, 0.96, 0.9}
\definecolor{mintcream}{rgb}{0.96, 1.0, 0.98}
\definecolor{mintgreen}{rgb}{0.6, 1.0, 0.6}
\definecolor{mistyrose}{rgb}{1.0, 0.89, 0.88}
\definecolor{palegold}{rgb}{0.9, 0.75, 0.54}
\definecolor{palechestnut}{rgb}{0.87, 0.68, 0.69}

\newcommand{\reva}[1]{{\leavevmode\color{blue!90!black}{#1}}}
\newcommand{\revb}[1]{{\leavevmode\color{green!70!black}{#1}}}
\newcommand{\revc}[1]{{\leavevmode\color{red!80!black}{#1}}}
\newcommand{\common}[1]{{\leavevmode\color{magenta!70!black}{#1}}}

\renewcommand{\reva}[1]{{\leavevmode\color{black}{#1}}}
\renewcommand{\revb}[1]{{\leavevmode\color{black}{#1}}}
\renewcommand{\revc}[1]{{\leavevmode\color{black}{#1}}}
\renewcommand{\common}[1]{{\leavevmode\color{black}{#1}}}

\def\HiLiG{\leavevmode\rlap{\hbox to \hsize{\color{green!30}\leaders\hrule height .8\baselineskip depth .5ex\hfill}}}
\def\HiLiY{\leavevmode\rlap{\hbox to \hsize{\color{yellow!50}\leaders\hrule height .8\baselineskip depth .5ex\hfill}}}
\usepackage{wrapfig}
\definecolor{light-gray}{gray}{0.95}
\usepackage{tcolorbox}

\begin{document}

\setcounter{page}{1}

\title{Causal Explanations for Disparate Trends: Where and Why?}
% \author{Anonymous Authors}

\author{Tal Blau}
\affiliation{%
  \institution{Ben-Gurion University of the Negev}
  \city{}
  \state{}
  \country{}}
\email{tbl@post.bgu.ac.il}

\author{Brit Youngmann}
\affiliation{%
  \institution{Technion}
  \city{}
  \state{}
  \country{}}
\email{brity@technion.ac.il}

\author{Anna Fariha}
\affiliation{%
  \institution{University of Utah}
  \city{}
  \state{}
  \country{}}
\email{afariha@cs.utah.edu}

\author{Yuval Moskovitch}
\affiliation{%
  \institution{Ben-Gurion University of the Negev}
  \city{}
  \state{}
  \country{}}
\email{yuvalmos@bgu.ac.il}

\begin{abstract}

\looseness-1 During data analysis, we are often perplexed by certain
\textit{disparities} observed between two groups of interest within a dataset.
To better understand an observed disparity, we need \textit{explanations} that
can pinpoint the data regions where the disparity is most pronounced, along with
its causes, i.e., factors that alleviate or exacerbate the disparity. This task
is complex and tedious, particularly for large and
high-dimensional datasets, demanding an automatic system for discovering
\textit{explanations} (data regions and causes) of an observed disparity. It is critical that explanations for disparities
are not only interpretable but also actionable---enabling users to make
informed, data-driven decisions. This requires explanations to go beyond
surface-level correlations and instead capture \emph{causal} relationships. We
introduce \sysName, a framework for discovering causal
\underline{Ex}plana\-tions for \underline{Dis}parities between two groups of
interest. \sysName identifies data regions (subpopulations) where disparities
are most pronounced (or reversed), and associates specific factors that causally
contribute to the disparity within each identified data region. We formally
define the \sysName framework and the associated optimization problem, analyze
its complexity, and develop an efficient algorithm to solve the problem. Through
extensive experiments over three real-world datasets, we demonstrate that
\sysName generates meaningful causal explanations, outperforms prior methods,
and scales effectively to handle large, high-dimensional datasets.

\end{abstract}

\maketitle

%!TEX root=main.tex
\vspace{-2mm}
\section{Introduction}
\label{sec:introduction}

Data is the main building block of modern, data-driven decision-making. People rely on the trends observed in the data to gain insights, and in turn, use those insights to draw conclusions or even make important decisions. For large, high-dimensional datasets, certain observed trends often require further drilling down or \emph{explanations}. For example, after observing a \textit{disparate} trend that females are more likely to experience nervous breakdowns and anxiety attacks than non-females, one might wonder: ``For which subpopulation is this disparity more pronounced?'', ``Which factors contribute to this disparity?'', ``Are there particular countries/races where this trend is reversed?'', and so on. Manually searching for these answers is like finding a needle in a haystack, which demands automated ways to pinpoint the subpopulation or data region where a trend is amplified or substantially reversed from the global trend, and identify factors that are connected to these disparate trends. 

% Understanding the {\em causal reasons} behind disparities in average outcomes between two groups is crucial for informed, data-driven decision-making aimed at addressing inequities. For example, if a policymaker recognizes the factors causally contributing to a lower average salary for a certain subpopulation compared to the rest of the population, they can implement targeted corrective measures. 

Understanding {\em causal reasons} behind disparities in outcomes between two
groups is essential for making informed, data-driven decisions to address
inequities. For instance, if a policymaker identifies the factors that causally
lead to lower average salaries in a specific subpopulation compared to the rest
of the population, they can design targeted interventions to mitigate the gap.

\looseness-1 In this paper, we propose \sysName for automatically
\underline{Ex}plaining an observed \underline{Dis}parate trend. We proceed to
provide three examples, highlighting unique use cases, to motivate the need to
discover \textit{explanations} (data regions and associated causes) of an
observed disparate trend between two groups of interest within a dataset.

\begin{table}[t]
\footnotesize
    \centering
	\vspace{3mm}

    \caption{\small A sample (toy) dataset over a partial schema of the Stack
    Overflow Annual Developer Survey dataset~\cite{stackoverflowreport}.
}
	
	\vspace{-4mm}
    \resizebox{\columnwidth}{!}{
    \begin{tabular}[b]{@{}llll@{}rr@{}}
        \toprule
        %textbf{ID}& 
        \textbf{Gender} &\textbf{Ethnicity}& \textbf{Education} & \textbf{Role} & \textbf{YrsProfCoding} & \textbf{TC}
        \\ 
        \midrule
        Non-binary  &White          &BS & Business analyst   & 6-8  &   83K     \\
        Male        &South Asian    &PhD & Data analyst       & 4-6  &   124K    \\
        Female      &South Asian    &MS & Back-end developer & 2-4  &   75K     \\
        Male        &East Asian     &BS & Back-end developer & 6-8   &   59K     \\
        \bottomrule
    \end{tabular}}
	\vspace{-6mm}
    \label{tab:data}
\end{table}

\newcommand{\groupOneAvg}{\$106K\xspace}
\newcommand{\groupTwoAvg}{\$96K\xspace}
\newcommand{\howMuchHigher}{10\%\xspace}
\newcommand{\groupOne}{White individuals aged between 25–34\xspace}
\newcommand{\treatment}{having 6--8 years of professional coding experience\xspace}
\newcommand{\support}{35\%\xspace}
\newcommand{\groupOneAvgSR}{\$115K\xspace}
\newcommand{\groupTwoAvgSR}{\$105K\xspace}
\newcommand{\howMuchHigherSR}{9.5\%\xspace}
\newcommand{\treatmentChangeOne}{\$44K\xspace}
\newcommand{\treatmentChangeTwo}{\$10K\xspace}

\begin{example}[Investigating a disparate trend]\label{ex:one}

An analyst Miro is examining tech workers' total-compensation data (Table~\ref{tab:data}), which contains
information about individuals' demographics (gender, ethnicity, age, etc.),
role, professional coding experience, their own and parents' education,
total compensation (TC), etc. Contrary to the common knowledge, Miro observes
that the average TC of data or business analysts (\groupOneAvg) is about
\howMuchHigher higher than the average TC of back-end developers
(\groupTwoAvg)---a surprising trend!
Miro wants to identify large subpopulations that contribute
significantly to this trend, and uncover causes behind it. He discovers that one of such subpopulation is \groupOne, which
constitutes \support of the population, within which, on
average, analysts (\groupOneAvgSR) earn \howMuchHigherSR more than developers
(\groupTwoAvgSR). Miro further discovers that having worked as a professional
coder is a major contributing factor to this disparity, particularly for this
subpopulation. Specifically, among \colorbox{Orange!30}{\strut{\groupOne}},
\colorbox{Cyan!30}{\strut{\treatment}} causes a TC increase of
\increase{\treatmentChangeOne} for {\colorbox{Yellow!70}{\strut
\texttt{\small analysts}}}, and only \increase{\treatmentChangeTwo} for $\;$
\colorbox{Lavender!70}{\strut  \texttt{\small back-end developers}}, further
exacerbating the TC-gap.

% \af{Update these
% based on new SO results.}

% Miro wants to identify the subpopulations for which the trend is more pronounced
% and uncover the causes behind it. After several hours of digging, he discovers
% that one of the pronounced subpopulations is \groupOne, which constitutes
% \support of the survey population. On average, analysts (\groupOneAvgSR) earn
% \howMuchHigherSR more than developers (\groupTwoAvgSR) within this
% subpopulation. Miro further discovers that exercising is a major contributing
% factor to this disparity, particularly for this subpopulation. Specifically,
% among \textcolor{orange}{\groupOne}, \textcolor{blue}{\treatment} causes a TC
% increase of \textcolor{PineGreen}{\treatmentChangeOne} for
% {\colorbox{Yellow!70}{\strut \texttt{analysts}}}, and only
% \textcolor{PineGreen}{\treatmentChangeTwo} for \colorbox{Lavender!70}{\strut
% \texttt{back-end developers}}, further exacerbating the TC-gap. \af{Update these
% based on new SO results.}

Miro wonders if there are other significant data regions
with similar properties. What are the major causes of the disparity in those
regions? Unfortunately, the dataset contains a number of multi-valued
attributes, making manual exploration of all possible data regions
infeasible. Furthermore, exploring all possible factors that can cause a
significant disparity in the target attribute (TC) requires a more involved
search, making it impossible for Miro to do it manually.
\end{example}

\renewcommand{\groupOneAvg}{78\%\xspace}
\renewcommand{\groupTwoAvg}{91\%\xspace}
\newcommand{\howMuchLower}{13\%\xspace}
\renewcommand{\support}{16\%\xspace}
\renewcommand{\groupOneAvgSR}{65\%\xspace}
\renewcommand{\groupTwoAvgSR}{84\%\xspace}
\newcommand{\howMuchLowerSR}{19\%\xspace}
\renewcommand{\treatmentChangeOne}{2\%\xspace} 
\renewcommand{\treatmentChangeTwo}{1\%\xspace}

\begin{example}[Debugging bias]\label{ex:bias}

\looseness-1 While analyzing a health insurance coverage
dataset~\cite{ACS_Data}, Soha observes a disparate trend that individuals with
an occupation that involves manual labor have a \howMuchLower lower chance
(\groupOneAvg) of being covered by health insurance than the overall population
across all occupations (\groupTwoAvg). This indicates a ``blue-collar
bias''~\cite{godbolt2011black} and Soha wishes to discover data regions where
this bias is significant and uncover why. Turns out that among non-natives,
which makes up \support of the population, the disparity is even more severe.
Within this subpopulation, manual-labor workers have a \groupOneAvgSR chance of
being insured, which is \howMuchLowerSR lower than the \groupTwoAvgSR coverage
rate for all occupations. Furthermore, among \colorbox{Orange!30}{\strut{not-natives}},
\colorbox{Cyan!30}{\strut{earning between 25K--55K}} \textit{boosts} the chance of having
health insurance for \colorbox{Yellow!70}{\strut \texttt{\small manual-labor workers}}
by \increase{\treatmentChangeOne}, where it \textit{hurts} the
chance for \colorbox{Lavender!70}{\strut \texttt{\small people with any
occupation}} by \decrease{\treatmentChangeTwo}.
% \tal{the example needs to update to - Furthermore, among White individuals from Southern region who speak Spanish,
% have no personal earnings the chance of having health insurance for manual-labor workers by 7.74\%, where it increase the chance for the people with any occupation by 4.87\%.}

\end{example}

\renewcommand{\groupOneAvg}{37\%\xspace}
\renewcommand{\groupTwoAvg}{45\%\xspace}
\renewcommand{\support}{2\%\xspace}
\renewcommand{\groupOneAvgSR}{47\%\xspace}
\renewcommand{\groupTwoAvgSR}{43\%\xspace}
\renewcommand{\treatmentChangeOne}{21\%\xspace} 
\renewcommand{\treatmentChangeTwo}{14\%\xspace}

\begin{example}[Discovering reverse trends]\label{ex:three}

Generally, males have a lower likelihood (\groupOneAvg) of feeling nervous
frequently than non-males (\groupTwoAvg). Madison is investigating a Medical
Expenditure Panel Survey dataset~\cite{MEPS_Data_Overview} and they want to find subpopulations where a reverse trend exists, i.e., males have a
higher likelihood of feeling nervous than non-males (this phenomenon is known as
Simpson's Paradox~\cite{wagner1982simpson}). One such subpopulation is
\colorbox{Orange!30}{\strut{divorced people  aged between}} \colorbox{Orange!30}{\strut{51--63 who have a
doctor's recommendation to exercise}}, where males have a higher
likelihood (\groupOneAvgSR) of feeling nervous than non-males (\groupTwoAvgSR).
Interestingly, within this subpopulation, \colorbox{Cyan!30}{\strut{not currently
smoking}} exacerbates the situation for \colorbox{Yellow!70}{\strut
\texttt{\small males}} (increases the likelihood of feeling nervous by
\increase{\treatmentChangeOne}) but improves the situation for
\colorbox{Lavender!70}{\strut \texttt{\small non-males}} (decreases the likelihood of
feeling nervous by \decrease{\treatmentChangeTwo}). However, to discover
such a reverse trend, they must manually examine all possible subpopulations and try
out all possible treatments within each subpopulation!

% \tal{need to update the example: Among never married Southerners without a doctor’s exercise recommendation or a Diabetes diagnosis, males are slightly more likely to feel nervous (43.9\%) than non-males (43.5\%). Surprisingly, not currently smoking reduces this likelihood by 13\% for males and by 21\% for non-males, highlighting an even greater disparity within this subpopulation.}

% \brit{we need to add the causal explanation part to this example. Also, non of the examples referred to the non-overlap. }
% \af{@Brit: Not sure what you mean by adding the causal explanation part. Can you take a stab at it?}
% \af{@Brit: I put the overlap part below, in the text. See if that could work.}
\end{example}

% As database interactions become more accessible to a broader audience of data analysts and decision makers, the need for automated and insightful explanations of observed data trends has grown significantly. 

The above examples motivate investigating disparities in an aggregated
\textit{outcome variable} (e.g., TC) between \textit{two (\newtext{possibly
overlapping}) groups} of interest (e.g., analysts vs back-end developers),
aiming to identify (1)~\emph{where} the disparities are most pronounced (or
reversed), such as a specific data region or subpopulation and (2)~\emph{why},
i.e., what factors further alleviate/exacerbate the disparity. \revc{An example of a real-world use case is the famous UC Berkeley gender bias case~\cite{doi:10.1126/science.187.4175.398}, where aggregate admissions data suggested bias against women, but deeper analysis revealed Simpson's paradox driven by department-level differences. Another real-world example is the kidney stone treatment study~\cite{Julious1480}, where an inferior treatment appeared to be more effective overall, even though the superior treatment performs better for both small and large kidney stones. Our goal is to detect and explain such disparities.}

\subsubsection*{Desiderata}\label{realworld} There are three key desiderata for the
aforementioned problem. \textbf{First,} a single causal explanation rarely
accounts for the disparity observed across different subpopulations, which may exhibit disparity for different underlying reasons. Thus, the
first goal is to identify \emph{high-utility} subpopulations---groups where a
strong and meaningful causal explanation for the disparity exists.
\textbf{Second,} while small subpopulations may exhibit strong causal
explanations, insights derived from such groups are often not generalizable. To
ensure broader applicability and avoid misleading insights, chosen
subpopulations must have sufficient \emph{support}, i.e., they should cover
a reasonable portion of the overall dataset. \textbf{Third,} selecting the
top-$k$ subpopulations purely based on utility and support may lead to redundancy. E.g., ``principal engineers'' and ``people
aged between 35--45'' may comprise the same individuals, as most principal
engineers are 35--45 years old. Thus, beyond finding high-utility and
high-support subpopulations, we must minimize the overlap among the reported $k$
subpopulations, ensuring \emph{diversity}~\cite{yu2009takes, DBLP:conf/icdt/Moumoulidou0M21,
DBLP:journals/pvldb/WangMM18}.

% rarely can a single causal explanation explain the observed disparity for the entire population. In fact, in different subpopulations, disparity between two groups may be caused by different reasons, some contributing more to the disparity than others. Therefore, our first goal is to discover \emph{high-utility} subpopulations, for which a strong causal explanation exists for the observed disparity. \textbf{Second,} \newtext{small subpopulations with low support w.r.t the entire population may have a strong causal explanation; however, insights drawn from such a small subpopulation are not representative and generalizable}. Therefore, the reported subpopulations should have reasonably \textit{high support} (data coverage). \textbf{Third,} finding the top-$k$ high-utility and high-support subpopulations where the disparity is most pronounced may result in redundancy. For example, ``principal engineers'' and ``people with age between 35--45'' may comprise the same individuals as most principal engineers are 35--45 years old. Therefore, beyond finding high-utility and high-support subpopulations, we must minimize the overlap among the reported $k$ subpopulations. This alludes to the notion of \emph{diverse} selection of subpopulations~\cite{yu2009takes, DBLP:conf/icdt/Moumoulidou0M21, DBLP:journals/pvldb/WangMM18}. 

\subsubsection*{Problem} \looseness-1 Given a database $D$, an outcome variable
$O$, causal background knowledge in the form of a causal DAG \model by Pearl's
graphical causal model~\cite{pearl2009causal}, two groups of interest $g_1$ and
$g_2$, and a parameter $k$, our goal is to generate a set of $k$ \emph{disparity
explanations} that, collectively, best explain the disparities between $g_1$ and
$g_2$ w.r.t an aggregation over $O$, according to \model. In this
work, we consider \texttt{AVERAGE} as the aggregation function since it
satisfies the requirements for causal analysis (details are in
Section~\ref{sec:3dot1}).
% \brit{we need to say upfront that we only consider the average outcome. we have a paragraph explaining why in section 3.1}
Each explanation consists of two components: (1)~a \textit{subpopulation} where
the disparity is pronounced (or reversed), and (2)~a \emph{treatment pattern}
that causally affects $g_1$ and $g_2$ disparately, within that subpopulation.
The quality of a set of disparity explanations is primarily determined by the
causal strength of the treatment patterns they reveal, measured by the
\emph{Average Treatment Effect (ATE)}~\cite{pearl2009causal}. This forms our
main optimization objective: to identify a set of subpopulations for which the
associated treatments strongly explain the observed disparity. In addition to
maximizing causal explainability, two important constraints guide the selection
process:
(1)~Each explanation must have sufficient \emph{support}---i.e., the
subpopulation it describes should cover a sizable fraction of the data, ensuring
representativeness and generalizability.
(2)~The selected subpopulations must exhibit low \emph{overlap}, encouraging
\emph{diversity} in the explanation set. 

% This is quantified using the Jaccard
% similarity between subpopulations, which helps avoid redundant or overly similar
% explanations.

% The goodness of a set of explanations depends on three factors:
% (1)~strength of the treatment patterns in the explanations to cause the observed disparity, measured in terms of \textit{Average Treatment Effect} (ATE)~\cite{pearl2009causal},
% (2)~\textit{support} of the explanations (what fraction of tuples are covered by the reported subpopulations), and
% (3)~\textit{diversity} of the subpopulations in the explanations, measured in terms of Jaccard similarity among the subpopulations.

% \anna{Also, the above definition needs to include a goal of the explanation, such as the best explanation in terms of some utility such as highest discrepancies are highlighted.}

\subsubsection*{Challenges and limitations of prior work}

\looseness-1 The key challenge lies in identifying subpopulations that are associated with strong causal explanations for observed disparities, while simultaneously satisfying constraints on support and diversity. Prior work mostly focused on selecting the top-$k$ subpopulations based on observed disparities or coverage~\cite{pastor2021looking,agmon2024finding}, often neglecting causal explainability or redundancy.
% In contrast, we formalize a more challenging problem: finding a set of $k$ subpopulations that maximizes causal utility (measured by ATE), subject to constraints on minimum support and maximum pairwise similarity (i.e., diversity). This formulation significantly increases complexity, as it requires a holistic search over combinations of subpopulations, rather than evaluating them independently. 
Moreover, unlike approaches that rely solely on observed outcome disparities~\cite{DBLP:journals/corr/abs-2402-05007}, we emphasize the importance of discovering subpopulations where the disparity is causally explained by specific treatments.
We extend prior work in two key directions:
(1)~We address a more difficult variant of the problem by optimizing over a set of $k$ subpopulations under support and diversity constraints, and
(2)~We focus on identifying high-quality \emph{causal} explanations, not only data regions with high disparities. 
% The key challenge here is to strike the right balance among the
% three desiderata. Prior work focuses on finding the top-$k$ subpopulations in
% terms of support and observed disparity in the
% subpopulations~\cite{pastor2021looking,agmon2024finding}. However, finding
% subpopulations while simultaneously satisfying the three criteria (utility,
% support, and diversity) poses additional challenges. Since we want to minimize
% the overlaps among subpopulations, we must take a \textit{holistic} approach of
% finding a set of size $k$---a harder problem than the top-$k$ version---due to a
% very large search space. Furthermore, we not only focus on subpopulations with
% high disparity in the outcome variable but also the ones that entail good
% \textit{causal explanations} that can explain the observed disparity. In
% summary, we extend prior work in this space in two ways: (1)~we consider a
% harder version of the problem, finding $k$-sized set of subpopulations while
% minimizing their overlap (maximizing diversity), and (2)~we prioritize
% high-quality causal explanations while choosing the subpopulations.
Recent works \cite{youngmann2024summarized,abs-2207-12718} have used
causal inference to explain aggregate query results and disparities between two
groups within it. However, they do not support overlapping groups, which is
essential for bias debugging as demonstrated in Example~\ref{ex:bias}. 
\revc{While our form of causal explanation (the same
treatment resulting in different outcomes for different groups) is
already established in prior works, such as XInsight~\cite{abs-2207-12718},
they focus only on a global explanation. In contrast, we
expand the granularity of explanation to subpopulation level, since trends in different subpopulations can diverge significantly from global trends.}
As such,
we aim to identify causal explanations within different subpopulations rather
than providing a single explanation for the entire data or query outcome.
Finally, operating on the entire data, rather than the aggregate view, poses
another challenge as the search space becomes significantly larger (we explain
this in Section~\ref{subsec:step_2}).

\subsubsection*{Contributions} \label{contribution}
% We present a novel framework named \sysName\ to explain the disparity between the average outcome variable between two groups of interest. 
We make the following contributions:

\begin{enumerate}[leftmargin=*,topsep=0pt]

\item %We \textbf{develop a framework \sysName} that generates 
We \textbf{formalize the problem} of generating a set of causal explanations to account for disparities in outcomes between two (possibly overlapping) groups of interest. Specifically, we define an optimization problem that seeks to \emph{maximize the causal utility} of the selected explanations—i.e., their ability to causally explain the observed disparity, subject to three constraints: (1)~a bound on the number of explanations (size $k$), (2)~minimum support for each explanation to ensure representativeness, and (3) limited overlap between subpopulations to reduce redundancy. We show that this problem is NP-hard (Section~\ref{sec:framework}).

% We \textbf{formalize the problem} of generating a set of causal explanations to
% explain the disparity in outcomes between two, possibly overlapping, groups of
% interest. We define this as an optimization problem to maximize the causal
% explainability of these explanations while restricting their overlap to reduce
% redundancy, ensuring a sufficient coverage for each chosen explanation, and a size constraint on the number of explanations. We show
% its NP-hardness (Section~3). 

\item We \textbf{develop the \sysName framework}, which operates in three steps. First, \sysName finds candidate
subpopulations. Then it identifies local explanations for each
candidate subpopulation by adapting a previous work on finding treatments with substantial causal effect~\cite{youngmann2024summarized}. Finally, it uses an
effective greedy strategy to find a $k$-sized explanation set (Section~\ref{sec:algo}).

\item We present a thorough \textbf{empirical analysis} over 3 real-world
datasets and present \textbf{3 case studies} that include 5 baselines, and 4
variants. We show that
\sysName generates higher quality explanations than the existing approaches and
can find alternative explanations that existing approaches miss. We also find
\sysName scalable and efficient in practice, with its runtime being linear w.r.t
the number of data tuples (Section~\ref{sec:exp}).

\end{enumerate}

\section{Related work}
\label{sec:related}

\noindent
\textbf{Identifying interesting subgroups in high-dimensional data.}
Prior work identifies the most intriguing data subsets for exploration~\cite{sathe2001intelligent,DBLP:journals/pvldb/YoungmannAP22,DBLP:journals/tkde/JoglekarGP19,geerts2004tiling,sarawagi2000user,sarawagi2001user,sarawagi1998discovery, MoskovitchLJ23, LiMJ23, CabreraEHKMC19, asudeh2019assessing, pastor2021looking}. Other studies focus on uncovering compelling data visualizations~\cite{vartak2015seedb,zhang2021viewseeker}, or pinpointing data subsets where models underperform~\cite{chung2019automated,pastor2021looking,DBLP:journals/corr/abs-2402-05007}. One of our key objectives is to identify subpopulations with a significant disparity in the average outcomes between two groups of interest. 

\textsc{DivExplorer}~\cite{pastor2021looking} is designed to analyze the behavior of classification models, with the primary goal of identifying data regions where a performance metric (e.g., false positive rates) deviates significantly compared to the entire dataset. 
\textsc{FairDebugger}~\cite{DBLP:journals/corr/abs-2402-05007} aims to identify data subsets responsible for fairness violations in the outcomes of a random forest classifier. It pinpoints the most impactful data samples that significantly influence the model’s predictions.
However, unlike \sysName, both of these systems focus solely on detecting data regions with unexpected behavior, without uncovering the underlying \emph{causal} factors driving the observed trends. We empirically compare \sysName against these baselines in Section~\ref{sec:exp}.
% As we demonstrated in our experiment, FairDebugger was unable to identify subpopulations with a significant disparity in the average outcome of two groups.  

% Smart Drill-Down~\cite{DBLP:journals/tkde/JoglekarGP19} aims to find a list of
% rules, each describing a group of tuples, that together capture interesting
% aspects of a given relational table. However, their problem formulation ensures
% monotonicity, enabling greedy approach to have provable approximation
% guarantees. In contrast, for the task of discovering causal surprises, no
% monotonicity holds in terms of (potentially unrelated) treatment variables.
% Thus, the approach taken in~\cite{DBLP:journals/tkde/JoglekarGP19} do not
% trivially extend in our problem setting. \brit{todo: data summarization - explanation table, divexplorer, fiardebugger}

\smallskip
\noindent
\textbf{Query results explanation.}
Extensive research has been devoted to explain the results of aggregate SQL queries. 
Multiple works leverage \emph{data provenance} to generate explanations for query results~\cite{bidoit2014query,chapman2009not,meliou2010complexity, meliou2009so,DeutchFG20,lee13approximate,ten2015high,li2021putting}, while some rely on non-causal interventions~\cite{wu2013scorpion,roy2014formal,roy2015explaining,tao2022dpxplain,DBLP:journals/pvldb/DeutchGMMS22,DeutchGM15, BourhisDM20}, entropy-based techniques~\cite{el2014interpretable}, and counterbalancing patterns~\cite{miao2019going}.
This line of work differs from ours, as our goal is to explain the disparity among two, possibly overlapping, groups of interest via a small set of causal explanations.

\looseness-1 Recent works \cite{salimi2018bias, youngmann2022explaining, youngmann2024summarized,abs-2207-12718} use causal inference to explain aggregate query results. Prior work~\cite{salimi2018bias,youngmann2022explaining} proposed methods to find confounding variables that explain the correlation between the grouping attribute and the outcome in group-by-average queries. Both provide the same explanation for all groups in the query results. 
\revb{CauSumX's~\cite{youngmann2024summarized} goal is to provide explanations for an entire aggregate view by combining similar explanations for brevity. To adapt it to our setting, it could be used to explain the aggregate results of two groups. However, our objective differs: rather than identifying treatments that influence each group \textit{individually}, we focus on finding what \textit{differentiates} the groups.
Specifically, we search for treatments that benefit one group but have the opposite effect on the other. In contrast, CauSumX simply aims to identify what influences the outcome within each group. Consequently, the treatments identified by CauSumX cannot serve as explanations for the observed disparity between the two groups.
Furthermore, CauSumX operates only on aggregate views with a relatively small number of grouping patterns, making it unlikely to scale in our setting. }

% CauSumX~\cite{youngmann2024summarized} focuses on providing causal explanations for group-by-average queries, aiming to explain the overall aggregate view by identifying factors influencing the outcome within each group in the query results. In contrast, our objective is to identify subpopulations within the data where the disparity between two groups is most pronounced and to offer specific causal explanations for the observed disparity within these subpopulations. 
% % , as detailed in Section~\ref{subsec:step_2}.
% \newtext{%Our novelty lies in the fact that, 
% Moreover, unlike CauSumX, which derives explanations from the aggregate view, we analyze the entire dataset, which is typically much larger. As a result, we introduce multiple optimizations to improve runtime efficiency.} While our goals differ, we adapt their treatment mining algorithm for our approach (Section~\ref{subsec:step_2}).

XInsight~\cite{abs-2207-12718} identifies both causal and non-causal patterns to explain disparities between two groups in aggregate queries. In contrast, we support overlapping groups and identify specific causal explanations within different subpopulations, rather than seeking a single treatment for the entire dataset. As our experiments confirm (Section~\ref{subsec:exp_quality}), in many cases, no universal explanation suffices, and disparities are better understood through localized causal insights.

\smallskip
\noindent
\textbf{Rule mining.} \newtext{Association rule mining is a widely studied problem~\cite{hipp2000algorithms,kumbhare2014overview}, which aims to identify frequent relationships in datasets. Rule-based interpretable models utilize these techniques to derive predictive rules, aiming to balance accuracy and interpretability~\cite{sagi2021approximating,schielzeth2010simple,lou2013accurate,kim2014bayesian}. Recent works have explored rule generation from causal relationships~\cite{plecko2022causal,plecko2023causal,sun2021treatment} and heterogeneous treatment effect estimation~\cite{xie2012estimating,wang2022causal}. However, they do not address our problem of identifying where disparities are significant and their causes.}

\revb{\looseness-1 CURLS~\cite{10.1145/3637528.3671951} aims at finding a set of causal rules that delineate subgroups with a significant treatment effect and small outcome variances. Like how we limit \#explanations and \#predicates within each explanation, CURLS limits \#rules and their size to ensure interpretability. While both works aim to minimize the overlap between explanations (rules in CURLS), our goal is to expose and explain disparities between groups, while CURLS focuses on identifying subpopulations with significant effects. Notably, subpopulations with strong treatment effects do not necessarily show high disparities across groups. In fact, subpopulations with weaker overall effects can reveal stronger disparities, as opposing effects between the groups within a subpopulation may cancel each other out.}

\vspace{-2mm}
\section{Background on Causal Inference}
\label{sec:prelim} 

We use Pearl's model for {\em observational causal analysis}
\cite{pearl2009causal} and present below a few concepts according to it. The
broad goal of {\em causal inference} is to estimate the effect of a {\em
treatment variable} $T$ on an {\em outcome variable} $O$ (e.g., the effect of
\verb|YrsProfCoding| on \verb|TC|).

The Average Treatment Effect (ATE) quantifying the difference in expected outcomes between treated and untreated groups~\cite{rubin2005causal, pearl2009causal}. ATE conceptually assumes a scenario where treatment is assigned randomly. However, in observational data, treatment is not assigned randomly, and \emph{confounding variables} that influence both treatment and outcome must be accounted for.  
To estimate ATE for a binary treatment $T$ on an outcome $O$, the following definition is used:  
\begin{equation*}
    {\small ATE(T,O) =\mathbb{E}_Z \left[ \mathbb{E}[O \mid T=1,\mathbf{Z}] -  
    \mathbb{E}[O \mid T=0, \mathbf{Z}] \right]}
\label{eq:ate}
\end{equation*}
Here, $\mathbf{Z}$ represents the set of confounding variables, ensuring that the causal effect is isolated from confounding influences.

% The gold standard of causal inference is by doing {\em randomized controlled experiments},
% where the population is randomly divided into a {\em treated} group that receives the treatment (denoted by $\text{do}(T = 1)$ for a binary treatment $T$) and the {\em control} group ($\text{do}(T = 0)$). One popular measure of causal estimate is {\em Average Treatment Effect} (ATE). In a randomized experiment, ATE is the difference in the average outcomes of the treated and control groups~\cite{rubin2005causal, pearl2009causal}, and is defined as:
% \begin{equation}
%     {\small ATE(T,O) = \mathbb{E}[O \mid \text{do}(T=1)] -  
%     \mathbb{E}[O \mid \text{do}(T=0)]}
% \label{eq:ate}
% \end{equation}
% The above definition assumes that the treatment assigned to one causal unit does not affect the outcome of another unit (called the {Stable Unit Treatment Value Assumption (SUTVA)) \cite{rubin2005causal}}\footnote{This assumption does not hold for causal inference on multiple tables and even on a single table where tuples depend on each other.}. 

\begin{example}
Suppose that we want to estimate the causal effect of $\;$\verb|YrsProfCoding| on
\verb|TC| from the Stack Overflow (SO) dataset. Since the values of
$\;$\verb|YrsProfCoding| were not assigned at random, and having more or fewer years
of professional coding experience and obtaining a high total compensation may
depend on other attributes like \verb|Ethnicity|, \verb|Education|, and
\verb|Role|, we must control for these confounding variables when estimating  $ATE(\verb|YrsProfCoding|, \verb|TC|)$.
% \af{Update based on new causal DAG that revolve around attributes shown in the
% exp results and Example 1.}
%
\end{example}

% However, randomized experiments where treatments are assigned at random cannot be done in many practical scenarios due to ethical or feasibility issues 
% (e.g., the effect of
% higher education on salary). In these scenarios, pearl's causal model allows sound causal inference under additional assumptions. Randomization in controlled trials mitigates the effect of {\em confounding factors}, i.e., variables that can affect the treatment assignment and outcome.

Pearl's model provides ways to account for these confounders $\mathbf{Z}$ to get an unbiased causal estimate under additional assumptions: (1)~The unconfoundedness assumption states that if we condition on $\mathbf{Z}$, then $T$ is independent of $O$, given $\mathbf{Z}$. Intuitively, it means that after conditioning on $\mathbf{Z}$, $T$ is as good as randomly assigned.
%
% then the treatment $T$ and the outcome $O$ are independent. \brit{I removed the example from here}
% In the Stack Overflow dataset, assuming that only $\mathbf{Z}$ = \{\verb|Gender|, \verb|Ethinicity|, \verb|YrsCoding|\} affects $T = $ \verb|Role|, if we condition on a fixed set of values of $Z$, i.e., consider people of a given gender and ethnicity with the same experience in coding, then $T = $ \verb|Role| and $O = $ \verb|TC| are independent. 
(2)~The Overlap assumption ensures that for every combination of confounders, there is a nonzero probability of receiving the treatment, allowing for valid comparisons across groups.

% ($\independent$ denotes independence):
% % \vspace{-2mm}
% \begin{eqnarray}
%     \textbf{Unconfoundedness}: & O \independent T | Z {=} z \label{eq:unconfoundedness}\\
%     \textbf{Overlap}: & 0 < Pr(T {=} 1 |Z {=} z)< 1 \label{eq:overlap}
% \end{eqnarray}
% The unconfoundedness assumption, Eq. (\ref{eq:unconfoundedness}), states that if we condition on $\mathbf{Z}$, then treatment $T$ and the outcome $O$ are independent. In Stack Overflow, assuming that only $\mathbf{Z}$ = \{\verb|Gender|, \verb|Age|, \verb|Country|\} affects $T = $ \verb|Education|, if we condition on a fixed set of values of $Z$, i.e., consider people of a given gender and age, from a given country, then $T = $ \verb|Education| and $O = $ \verb|TC| are independent. 

Pearl's model gives a systematic way (e.g., the backdoor criterion~\cite{pearl2009causal}) to find a sufficient set of confounding variables $\mathbf{Z}$ when a \emph{causal DAG} is available. 
A causal DAG is a specific type of Bayesian network, where nodes represent random variables (i.e., data attributes) and
edges signify potential direct causal influence, that provides a simple way to represent causal relationships among variables. Causal DAGs can be constructed by a domain expert or
using {\em causal discovery} algorithms~\cite{glymour2019review}. In line with prior
work~\cite{galhotra2022hyper,youngmann2024summarized,li2025fair}, we assume the causal DAG is given as part of the input.

% \paratitle{Causal DAG}
% Pearl's Probabilistic Graphical Causal Model model \cite{pearl2009causal} can be written as a tuple $(\exo, \edvar, Pr_{\exo}, \psi)$, where $\exo$ is a set of {\em unobserved exogenous (noise)} variables or attributes, $\Pr_{\exo}$ is the joint distribution of \exo, and $\mathcal{N}$ is a set of {\em observed endogenous variables}. %associated with observed attributes the $A \in \attrset$, 
% Here  $\psi$ is a set of structural equations that encode dependencies among variables. The equation for $A \in \edvar$ takes the following form:
% $$\psi_{A}: 
% \dom(Pa_{\exo}(A)) {\times} \dom(Pa_{\edvar}(A)) \to \dom(A)$$
% Here $Pa_{\exo}(A) {\subseteq} {\exo}$ and $Pa_{\edvar}(A) {\subseteq} \edvar \setminus \{A\}$ respectively denote the exogenous and endogenous parents of $A$. A causal relational model is associated with a {\em causal DAG}, $G$, whose nodes are the endogenous variables $\edvar$ and whose edges are all pairs $(X,Y)$ (directed edges from $X$ to $Y$) such that $Y {\in} \edvar$ and $X {\in} Pa_{\edvar}(Y)$. The causal DAG obfuscates exogenous variables as they are unobserved. Any given set of values for the exogenous variables completely determine the values of the endogenous variables by the structural equations (we do not need any known closed-form expressions of the structural equations in this work).
% The probability distribution $\Pr_{\exo}$ on exogenous variables $\exo$ induces a probability distribution  
% on the endogenous variables $\mathcal{N}$ by the structural equations $\psi$. 

\begin{example}
  Figure \ref{fig:causal_DAG} depicts a causal DAG for the SO dataset
  over a subset of attributes in Table \ref{tab:data}, which indicates that \verb|YrsProfCoding| depends on an
  individual's \verb|Role|, \verb|Ethnicity|, and \verb|Education|. 
\end{example}

% To compute the causal effect of $T$ on $O$ (as per ATE, Equation (\ref{eq:ate})), it is crucial to identify and adjust for confounders (the set of variables $\mathbf{Z}$). The backdoor criterion~\cite{pearl2009causal} provides a sufficient condition by identifying such a set of variables $\mathbf{Z}$ which can be checked against a causal DAG. 

For this work, we focus on estimating the causal effect of a treatment $T$ on an outcome $O$ within a specific subpopulation, characterized by a \emph{pattern} $\pattern$. (We define the set of patterns in Section \ref{sec:framework}.) Consequently, our goal is to compute the Conditional Average Treatment Effect (CATE)~\cite{rubin1971use,holland1986statistics} rather than the ATE, as CATE captures the treatment effect within a targeted subpopulation. To estimate the CATE for a binary treatment $T$ on an outcome $O$ for a subpopulation $\pattern$, we use the following definition:
\begin{equation*}
    {\small CATE(T,O| \pattern) =\mathbb{E}_Z \left[ \mathbb{E}[O \mid T=1,\mathbf{Z}, \pattern] -  
    \mathbb{E}[O \mid T=0, \mathbf{Z},\pattern] \right]}
% \label{eq:cate}
\end{equation*}
where $\mathbf{Z}$ represents a sufficient set of confounding variables.

% In Pearl's model, a treatment $T = t$ (on one or more variables) is considered as an {\em intervention} to a causal DAG by mechanically changing the DAG such that the values of node(s) for $T$ in $G$ are set to the value(s) in $t$, which is denoted by $\doop(T = t)$. Following this operation, the probability distribution of the nodes in the graph changes as the treatment nodes no longer depend on the values of their parents. Pearl's model gives an approach to estimate the new probability distribution by identifying the confounding factors $Z$ described earlier using conditions such as {\em d-separation} and {\em backdoor criteria} \cite{pearl2009causal}, which we do not discuss in this paper. 

%!TEX root=main.tex

\section{Discovering Disparity Explanations}
\label{sec:framework}
We consider a single-relation database instance \db over a schema $\attrset {=} (A_1, {\ldots}, A_m)$, where each attribute $A_i$ is associated with a domain $\dom(A_i)$. The outcome attribute $O {\in} \attrset$ can be either categorical or continuous.
We also assume all other attributes in $\attrset$ are categorical (if not, we can discretize them to make them categorical).
A tuple in \db is denoted as $t {=} (a_1, {\ldots}, a_m)$ where $a_i {\in} \dom(A_i)$. We use bold letters to represent a subset of attributes $\mathbf{A} {\subseteq} \attrset$.  
To specify a \emph{subpopulation} (subset of tuples) from \db, we use
{\em patterns}~\cite{roy2014formal,wu2013scorpion,el2014interpretable,lin2021detecting, li2021putting,youngmann2024summarized} that comprise conjunctive \emph{predicates} on attribute values.

\begin{definition}[Pattern]\label{def:pattern}
Given a database instance \db over schema \attrset,
a simple predicate $\varphi$ is an expression of the form $A_i \; \mathtt{op} \; a_i$, 
where $A_i {\in} \attrset$, 
$a_i {\in} \dom(A_i)$, and $\; \mathtt{op} \in \{=,\neq\}$.  
A {\em pattern} $\pattern$ is a conjunction of simple  predicates, i.e., $\pattern = \varphi_1 \land \ldots \land \varphi_k$. We use $\pattern(\db) \subseteq \db$ to denote the subpopulation within $\db$ defined by $\pattern$.
\end{definition}

\begin{example} Examples of two simple predicates in the SO dataset (Table~\ref{tab:data}) are {\tt{Ethnicity}} {=} \verb|Asian| and {\tt{Role}} $=$ \verb|Data analyst|. An example of a pattern is {\tt{Ethnicity}} {=} \verb|Asian| $\wedge$ {\tt{Role}} $=$ \verb|Data analyst|.
\end{example}

The two groups of interest, $g_1$ and $g_2$, are defined by the patterns $\pattern_{g_1}$ and $\pattern_{g_2}$ respectively. In this work, we only consider equality or inequality predicates, in line with previous work on explanations that deem such predicates intuitive and understandable~\cite{el2014interpretable, roy2015explaining,agmon2024finding}. 

% We leave for future work the consideration of a richer class of predicates, as even this limited space of patterns is sufficiently large and handling them is computationally challenging.  

% Our objective is to identify disparities between these groups with respect to the outcome $O$ in different scopes (subpopulations) and provide causal explanations for them.

\subsection{Disparity Explanations}\label{sec:3dot1}
For a database $\db$, we aim to discover explanations for an observed disparity in average outcome \texttt{AVG}($O$) between $g_1$ and $g_2$.
Our building blocks are \emph{disparity explanations} that identify \textit{where} the average outcomes for $g_1$ and $g_2$ differ significantly and \textit{why}.

\begin{example}
An analyst over the SO dataset (Table~\ref{tab:data}) is interested in finding
explanations of a surprising disparate observation: the average \verb|TC| of
data or business analysts ($\$106k$) is $\$10k$ higher than the average
\verb|TC| of back-end developers ($\$96k$). 

% \af{Update the TC numbers based on
% Example 1 update after new results.}
%
\end{example}

% For a database $\db$ defined over a schema $\attrset$, let $O \in \attrset$ denote the outcome attribute. Our objective is to 

We focus on comparing the average outcomes of two groups. Our framework is designed to uncover {\em causal explanations} for differences in these averages, leveraging the concept of CATE as discussed in Section~\ref{sec:prelim}. Since CATE inherently relies on {\em expectations} (weighted averages), it is particularly suited for analyzing aggregate averages of outcomes. 
Aggregate functions such as {\tt SUM} or {\tt COUNT}, on the other hand, depend on the {\em number of tuples} in the data, which does not directly align with causal effect estimates. While non-causal approaches to explanations \cite{wu2013scorpion, roy2014formal, miao2019going, li2021putting, lakshmanan2002quotient, wen2018interactive} can support a variety of aggregate functions, methods that are based on causal estimates typically focus on averages~\cite{salimi2018bias,youngmann2022explaining,abs-2207-12718, youngmann2024summarized,li2025fair}.
% \tal{is it still relevant? we now pick all the results above support threshold, without considering the averages of the groups}

\begin{figure}[t]
    \centering
    \footnotesize
    \resizebox{0.18\textwidth}{!}{
    \begin{tikzpicture}[node distance=1cm and 1cm, every node/.style={minimum size=0.9cm}]
        \tikzset{vertex/.style = {draw, rectangle, align=center}}

        % \node[vertex] (Country) {\small \bf{Country}};
        \node[vertex] (Gender) {\small \bf{Ethnicity}};
        \node[vertex, right=0.3cm of Gender] (Education) {\small \bf{Education}};
        \node[vertex, below right=0.3cm of Education] (Role) {\small \bf{{Role}}};
        \node[vertex, below=0.3cm of Gender] (YrsCoding) {\small \bf{YrsProfCoding}};        
        \node[vertex, below=1.5cm of Gender] (TC) {\small \bf{TC}};

        % \draw[->] (Country) -- (TC);
        
        \draw[->] (Gender) -- (YrsCoding);
        \draw[->] (Education) -- (YrsCoding);
        \draw[->] (Role) -- (YrsCoding);
        \draw[->] (Role) -- (TC);
        % \draw[->] (Country) -- (Education);
        % \draw[->, bend right=15] (Country) to (Education);
          \draw[->, bend right=15] (Gender) to (Role);
        % \draw[->] (Gender) -- (Education);
        \draw[->] (Education) -- (Role);
        \draw[->] (YrsCoding) -- (TC);
        \draw[->] (Gender) to[bend right=85] (TC);
        % \draw[->] (Age) -- (TC);
          \draw[->, bend left=25] (Education) to (TC);
    \end{tikzpicture}}
	\vspace{-3mm}
    \caption{\small Partial causal DAG for the Stack Overflow dataset. }
	\vspace{-3mm}
    \label{fig:causal_DAG}
\end{figure}
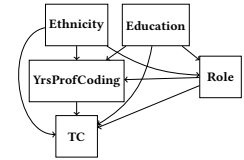

\noindent\subsubsection{\revc{Mutable and immutable attributes}} \label{sec:mutable}
We assume the attributes set $\attrset \setminus \{O\}$ is partitioned into two
disjoint sets: \textit{mutable} attributes that can be used to define what
affects the outcome (e.g., years of coding professionally, education)
and \textit{immutable} attributes, which are inherent and cannot be changed
(e.g., ethnicity, gender). We use immutable attributes to define the
subpopulations. \revc{Formally, $\immutable {=} \{I_1, I_2, \dots\} \subset \attrset$ denotes the set of
immutable attributes and $\mutable {=} \{M_1, M_2, \dots\} \subset \attrset$ denotes the set of mutable
attributes, where $\mutable \cap \immutable = \emptyset$} and the outcome $O
\notin \mutable \cup \immutable$. This makes sure that
explanations consist solely of mutable attributes that can imply corrective
measures to reduce the disparity. We assume that a domain
expert provides this categorization of attributes. This is similar to prior work
on counterfactual explanations, where certain attributes are excluded in the
explanation as they are
non-actionable~\cite{GalhotraPS21,KarimiSV21,KarimiBSV23}.

% \vspace{1mm}
% Given two groups $g_1$ and $g_2$ defined by the patterns $\pattern_{g_1}$ and $\pattern_{g_2}$, a \fact\ is defined as follows: 

\begin{definition}[\fact]\label{def:4dot2}
Given a database instance $\db$ over a schema $\attrset$, \revc{two disjoint sets of mutable and immutable attributes $\mutable \subset \attrset$ and $\immutable  \subset \attrset$}, an outcome variable $O \in \attrset$, and two groups of interest $g_1$ and $g_2$, a \fact\ $\phi$ is defined as a pair of patterns $(\pattern_g,\pattern_e)$ where:

\indent (1)~$\pattern_g$ is defined by attributes in $\immutable {\subset} \attrset$, highlighting a subpopulation with significant disparity between $g_1$ and $g_2$ in terms of \verb|AVG|($O$). 

\indent (2)~$\pattern_e$ is defined by  attributes in $\mutable{\subset} \attrset$, indicating a treatment that can explain the disparity between $g_1$ and $g_2$ within  $\pattern_g(\db)$. 

\end{definition}

To assess the impact of the treatment $\pattern_e$ on the outcome $O$ within the subpopulation $\pattern_g(\db)$, we compare the causal effect of $\pattern_e$ on $O$ within the two subpopulations: $(\pattern_g \wedge \pattern_{g_1})(\db)$ and $(\pattern_g \wedge \pattern_{g_2})(\db)$.

\begin{example}
\label{ex:disparity_explanation}
\begin{sloppypar}
Continuing with our running example, where $g_1$ is \colorbox{Yellow!70}{\strut
\texttt{\small analysts}} and $g_2$ is \colorbox{Lavender!70}{\strut \texttt{\small back-end
developers}}, an example disparity explanation is: Among
\colorbox{Orange!30}{\strut{white individuals aged between 25--34}},
\colorbox{Cyan!30}{\strut {having 6-8 years of professional
coding experience}} boosts the \verb|TC| for
\colorbox{Yellow!70}{\strut \texttt{\small analysts}} more than
\colorbox{Lavender!70}{\strut \texttt{\small back-end developers}}. Here, the
subpopulation pattern $\pattern_g$ is defined by
\colorbox{Orange!30}{\strut{$\mathtt{Ethnicity} = \mathtt{White} \wedge
\mathtt{Age} = \mathtt{25 - 34}$}} and the treatment
pattern $\pattern_e$ is \colorbox{Cyan!30}{\strut{\texttt{YrsProfCoding} = 6--8}}. Within this subpopulation, the average \verb|TC| for
\colorbox{Yellow!70}{\strut \texttt{\small analysts}} and \colorbox{Lavender!70}{\strut
\texttt{\small back-end developers}} are \$115K and \$105K, respectively (gap is
\$10K). 
% \af{Update this example based on new results.}
%
\end{sloppypar}
 % This highlights the disparity within this specific subpopulation of \texttt{orange}{white males}, where having some coding experience proves to be more beneficial for analysts than for back-end developers.
\end{example}

\subsection{Problem Formulation}
\label{subsec:problem_for}
We formally define the problem of finding \facts. We assume that we are
given a database instance \db\ with schema \attrset, a causal model \model\ on
\attrset, outcome $O \in \attrset$, and two groups $g_1$ and $g_2$ defined by
the patterns $\pattern_{g_1}$ and $\pattern_{g_2}$. Let $\{\phi_1, \phi_2,
\ldots, \phi_l\}$ be a set of possible \facts. Our goal is to find
a bounded-sized set of disparity explanations $\Phi$ to identify subsets of the
data that (1)~provide insights into the disparity between $g_1$ and $g_2$, and
(2)~avoid redundancy across subsets to cover different data regions.

% To this end, we define the \emph{\score}, which measures the usefulness of a \fact, the \emph{support} of a \fact that allows us to ignore \facts\ that constitute only minor portions of the data, and the \emph{similarity} between \facts\ to measure the redundancy across different \facts.

% To this end, we define the \$\Delta$\ and $NOverlap$ of $\Phi$. Intuitively, the \$\Delta$\ of $\Phi$ measures the usefulness of $\Phi$ in explaining the disparity by combining the \textit{supports} in the data of all \facts\ in $\Phi$ and their explainability scores. $NOverlap$ is used to measure the quality of lacking redundancy across different \facts. We next formally define the \$\Delta$\ and $NOverlap$ of a set of \facts\ $\Phi$. 

% \smallskip\noindent\emph{Explainability score.} 
\smallskip\noindent\emph{Usefulness of an explanation.}
% \brit{Inspired by \cite{pradhan2022interpretable}, a \fact\ is considered to be intersng if it displays a significant disparity between two groups, but also covers a significant portion of the data. Inspired by \cite{pradhan2022interpretable}, we therefor combine the }
A \fact\ $\phi = (\pattern_g,\pattern_e)$ is considered useful if  %($i$) it is \emph{surprising} or \emph{interesting}, i.e., 
($i$)~it reveals a significantly different effect of the treatment  $\pattern_e$ on the outcome $O$ for the subpopulation $\pattern_g(D)$ between $g_1$ and $g_2$, and ($ii$) it constitutes a significant portion of the data. 
To this end, we define the \emph{\score}, which measures the magnitude of the disparity, and the \emph{support} of a \fact that allows us to eliminate \facts\ that constitute only minor portions of the data.

% We next formalize these terms.
% The utility of \fact integrates the 
% \paragraph*{Interestingness score} 
% A \fact\ $\phi = (\pattern_g, \pattern_e)$ is \emph{surprising} or \emph{interesting} with respect to an outcome $O$ and two groups $g_1$ and $g_2$ if the treatment effect of $\pattern_e$ on $O$ within $\pattern_g(\db)$ significantly differs between $g_1$ and $g_2$. 
% We formally define the interestingness score as follows.

% Given a database instance \db\ with schema \attrset, a causal model \model\ on \attrset, and outcome $O \in \attrset$ and two groups $g_1$ and $g_2$ defined by the patterns $\pattern_{g_1}$ and $\pattern_{g_2}$. 
% The unnormalized interestingness score ($\Delta$) 

% \smallskip

The \score\ $\disp$ of a \fact\ $\phi = (\pattern_g, \pattern_e)$ measures the
absolute difference between the two CATE values: one computed over the
subpopulation $(\pattern_g \land \pattern_{g_1})(\db)$ and the other over
$(\pattern_g \land \pattern_{g_2})(\db)$. The difference is normalized by the
maximal outcome value. Formally,
$$\disp(\phi) = \frac{\left|CATE_{\model}(\pattern_e,O| \pattern_g \wedge \pattern_{g_1}) - CATE_{\model}(\pattern_e,O| \pattern_g \wedge \pattern_{g_2}) \right|}{\max\{|o|\mid o\in O\}}$$

In order to prioritize \facts\ that cover a large portion of the given database, we use the notion of \emph{support}. 
% \noindent
% \textbf{Support}:
The support of a \fact\ $\phi = (\pattern_g, \pattern_e)$ is defined by the
fraction of tuples $\in \db$ that take part in the explanation, namely, tuples
that satisfy the patterns in the \fact. Formally, $$support(\phi) =
\frac{|\pattern_{g\wedge g_1}(\db)\cup \pattern_{g\wedge g_2}(\db)|}{|\db|}$$
Intuitively, the higher the support of a \fact, the more interesting it is, as
it applies to a larger portion of the population. We prefer disparity
explanations with high support.

\begin{example}
Continuing from Example \ref{ex:disparity_explanation}, the support for the
disparity explanation, over the subpopulation \colorbox{Orange!30}{\strut{White
individuals aged}} \colorbox{Orange!30}{\strut{between 25 to 34}}---working as either analysts or back-end
developers---is $\frac{16,508}{47,702} = 34.6\%$. The \score\ for the
corresponding explanation---with the identified cause of \colorbox{Cyan!30}{\strut{having
6--8 years of coding experience}}---is $\frac{\left|44,058 -
10,552 \right|}{2,000,000} = 0.016$. Note that this is a real
explanation found during our empirical analysis (Table~\ref{tab:results_so}, row~5).

% Finally, we obtain the \$\Delta$\ by the multiplication of the \nscore\ and the support, which is 0.0032185.
%% ANNA: I removed the $\Delta$ computation because it doesn't make sense to do it for only one disparity explanaiton, we need a set. I think we don't need example for $\Delta$, it is pretty straightforward.
\end{example}

\noindent\emph{Diversity among the disparity explanations.} We are interested in
a diverse set of disparity explanations to reveal and explain the difference in
outcome for the two groups of interest. Given two groups of interest $g_1$ and
$g_2$, we use $D_{g_1\cup g_2}$ to denote the subset of $D$ containing tuples
that belong to at least one of the groups. More formally, $D_{g_1\cup g_2} =
\pattern_{g_1}(D)\cup\pattern_{g_2}(D)$. Given two disparity explanations $\phi
= (\pattern_{g}, \pattern_{e})$ and $\phi' = (\pattern_{g'}, \pattern_{e'})$,
defined over subpopulations $g$ and $g'$, respectively, and the same outcome
variable $O$, we use the Jaccard similarity between $\pattern_{g}(D_{g_1\cup
g_2})$ and $\pattern_{g'}(D_{g_1\cup g_2})$ to measure the similarity between
$\phi$ and $\phi'$. Formally:
%
% $\pattern_{g}(\db)$ and
% $\pattern_{g'}(\db)$ to measure the similarity between $\phi$ and $\phi'$.
 % \af{This is not aligned with support formula, it is not factoring the two groups of interest in.}
%
$$\expSim(\phi, \phi') = %J(\pattern_{g}(D_{g_1\cup g_2}), \pattern_{g'}(D_{g_1\cup g_2})) =
\frac{|\pattern_{g}(D_{g_1\cup g_2})\cap \pattern_{g'}(D_{g_1\cup g_2})|}{|\pattern_{g}(D_{g_1\cup g_2})\cup
\pattern_{g'}(D_{g_1\cup g_2})|}$$

We are now ready to formally define the problem of selecting disparity
explanations. At a high level, our goal is to select a bounded-sized diverse set
of disparity explanations with support above a given threshold, such that their
combined disparity score is maximized, with bounded pairwise similarity to
reduce redundancy.

% We are now ready to formally define the problem of selecting disparity explanations. At a high level, our goal is to select a bounded-sized set of disparity explanations such that the overlap is minimized (i.e., $NOverlap$ is maximized) and the \$\Delta$\ is maximized. We use a parameter $\alpha$ to strike a balance between the two objectives.

\begin{problem}[Disparity Explanation Selection]\label{prblem}
Given a database instance \db\ with schema \attrset, a causal model \model\ on \attrset,  outcome $O {\in} \attrset$, 
two groups of interest $g_1$ and $g_2$, a set of possible \facts\ $\Phi_c$, 
a budget $k \in \mathbb{N}^+$, a support threshold $\sigma$, and a similarity threshold $\tau$, 
% and a balance parameter $\alpha \in [0, 1]$,
select a disparity explanation set  $\Phi \subseteq \Phi_c$,
such that:
 \begin{enumerate}[leftmargin=1em,labelwidth=*,align=left]
     \item (\textbf{size constraint}) $|\Phi| \leq k$,
     \item (\textbf{support constraint}) $\forall \phi_i\in \Phi,~ support(\phi) \geq \sigma$,
     \item (\textbf{diversity constraints}) $\forall \phi_i,\phi_j\in \Phi,~\expSim(\phi_i,\phi_j) \leq \tau$, and
     \item (\textbf{objective}) $\Delta(\Phi) = \sum_{\phi\in \Phi}\Delta(\phi)$ is maximized.
     % \item  (\textbf{objective}) $f(\Phi) = \alpha \cdot \$\Delta$(\Phi) + (1-\alpha) \cdot NOverlap(\Phi)$ is maximized. 
 \end{enumerate}
 \end{problem}

\noindent
\revb{\emph{Complexity Analysis.} A na\"ive solution to Problem~\ref{prblem} requires (1)~materializing all subsets $\Phi {\subseteq} \Phi_c$ s.t $|\Phi| \le k$, (2)~validating each subset w.r.t the support and diversity constraints, and (3)~finding the valid subset that maximizes the objective function. Note that the number of possible patterns grows exponentially with the number of attributes and their domain sizes, leading to an exponential explosion of the number of candidate \facts $\Phi_c$. Concretely, $|\Phi_c| = \Pi_{A_i\in \mutable}(|dom(A_i)| + 1)\cdot\Pi_{A_j\in \immutable}(|dom(A_j)| + 1)$, 
% $|\Phi_c| = \Pi_{j=1}^{|\mutable|}(|dom(M_j)| + 1)\cdot\Pi_{j=1}^{|\immutable|}(|dom(I_j)| + 1)$, 
rendering enumeration of all possible explanations  infeasible. Next, we show that even if the full search space could be materialized, finding the optimal solution remains intractable.}

\begin{proposition}\label{prop:hard}
Given a set of candidate \facts\ $\Phi_c$, a budget $k$, a support threshold
$\sigma$, a similarity threshold $\tau$, and a bound $B$, determining whether
$\exists\Phi \subseteq \Phi_c$ s.t $|\Phi| \leq k$, $\forall \phi_i\in \Phi,~
support(\phi) \geq \sigma$, $\forall \phi_i,\phi_j\in \Phi,~\expSim(\phi_i,\phi_j)
\leq \tau$ and $\sum_{\phi\in \Phi}\Delta(\phi) \geq B$ is NP-hard.
\end{proposition}

The proof (given in the Appendix) is based on a reduction from the Independent Set problem, \revb{indicating that the problem is hard w.r.t the number of \facts, which itself is exponential in the number of attributes, as explained above.}

\subsection{\common{System Parameters}}\label{sysparams}

\noindent
\common{\emph{Data-Specific and Scenario-Specific Parameters.} We draw distinction between two types of users: an ``admin'' or domain expert (equivalent to DB administrator in traditional RDBMS), who will set up \sysName for the ``end-users'' (equivalent to SQL programmers). The admin is responsible for setting up \textit{data-specific parameters}---such as mutable and immutable attributes \mutable and \immutable, the causal DAG \model, and the support threshold $\sigma$. These parameters are scenario-agnostic and are set according to the data properties such as whether it is realistic to treat certain data attributes and the known causal relationships among data attributes. The distinction between immutable and mutable attributes ensures explanations are actionable for policymakers (e.g., race cannot be changed) but typically requires external knowledge to identify. However, the recent rise of powerful LLMs with general-domain knowledge can help bypass the need for domain expertise in identifying immutable attributes.

In contrast, the \textit{scenario-specific parameters} are specified by the end-users. These include (1)~the outcome attribute $O$ and two groups of interest $g_1$ and $g_2$ based on the desired scenario, (2) the desired number of explanations $k$, and (3) the diversity threshold $\tau$.}

\smallskip

\noindent\common{\emph{Parameter Tuning.}\label{paramtune}
Several established ways exist for setting the support threshold $\sigma$. E.g., 
domain experts often set $\sigma$ based on what counts as ``frequent enough'' to be meaningful for their applications. For instance, in retail data mining, $1\%$ of transactions are often considered actionable, while in medical data, even rare patterns can be important. Alternatively, in an exploratory tuning, a common practice is to start with a relatively high $\sigma$ (to keep the search space manageable), then gradually lower it until the number of candidate explanation patterns becomes too high (resulting in performance regression). Finally, an automated method is the elbow method~\cite{bholowalia2014ebk}, where one can plot the number of candidate patterns vs. $\sigma$ to observe a sharp increase to choose $\sigma$ near that point. 
Tuning the diversity threshold ($\tau$) can be guided by the budget parameter $k$: a high value for $\tau$ may result in fewer than $k$ explanations, which indicates infeasibility ($k$ explanations do not exist that satisfy the diversity requirement). Thus, in an exploratory setting, a practical way to tune $\tau$ is to set it such that $k$ explanations are retrieved.}

\section{The \sysName Algorithm}
\label{sec:algo}

Since the number of possible \facts\ can grow exponentially with the number of attributes and their domain values, enumerating all possible explanations is infeasible. Moreover, as shown in Proposition~\ref{prop:hard}, even if the full search space could be materialized, finding the optimal solution remains NP-hard. We therefore propose a highly scalable heuristic approach.
\reva{Our method builds on prior work in subpopulation mining~\cite{agrawal1994fast} (Section~\ref{subsec:step_1}) and explanation generation~\cite{youngmann2024summarized} (Section~\ref{subsec:step_2}), but our problem formulation is novel and introduces new challenges. A key contribution of our work is the integration of pattern mining with causal analysis in a practically efficient way. Specifically, our framework must explore a vast search space to identify subpopulations that exhibit both significant disparities and high-quality causal explanations. In addition, it must satisfy a diversity constraint that requires computing similarity between \emph{every} pair of explanations. To improve efficiency, we employ clustering. As we demonstrate in Section~\ref{subsec:exp_quality}, no existing method can be directly applied to solve this problem.}
Although \sysName\ does not provide theoretical guarantees, our experiments show that it achieves performance comparable to that of an exhaustive search that computes the optimal solution.

The \sysName framework comprises three steps:
(1)~the \textit{subpopulation miner}, which identifies subpopulations with sufficient support;
(2)~the \textit{explanation miner}, which uncovers causal explanations for each candidate subpopulation; and
(3) the \emph{greedy search}, which efficiently selects $k$ explanations adhering the diversity constraint.
% We build upon and adapt existing methods (e.g.,\cite{agrawal1994fast,youngmann2024summarized}) where applicable, and introduce novel techniques when necessary.

% While the first two modules leverage existing solutions to find candidate disparity explanations, the third module offers a novel, efficient, and effective approach to find a feasible solution.

% Note that \sysName lacks theoretical guarantees due to its design, which avoids generating all possible disparity explanations (as their number grows exponentially with the database size). If steps 1 and 2 were replaced by a brute-force approach that generates all disparity explanations, then a greedy approach for selecting a set of disparity explanations could approximate the optimal solution as the objective is a non-negative, monotone submodular function (as shown in Proposition~\ref{prop:submodular}).  

\vspace{-4mm}
\subsection{Subpopulation Miner}
\label{subsec:step_1}

\revc{Our first objective is to identify candidate subpopulations (data regions) where the disparity between groups $g_1$ and $g_2$ is significant and supported by sufficient data. To this end, we adapt the classical \emph{Apriori} algorithm~\cite{agrawal1994fast} to efficiently discover all subpopulations whose support exceeds a given threshold $\sigma$. This procedure is restricted to the set of immutable attributes $\immutable$, ensuring that the resulting candidate subpopulations are defined exclusively by immutable attributes. The full pseudocode is given in Algorithm~\ref{alg:apriori}, which modifies the standard Apriori algorithm to restrict candidate generation and frequency counting to the immutable attributes~$\immutable$. }

\setlength{\textfloatsep}{0pt}
\begin{algorithm}[t]
\small
\caption{Apriori-Based Subpopulation Miner}
\label{alg:apriori}
\KwIn{Dataset $\db$, immutable attributes $\immutable$, min support threshold $\sigma$}
\KwOut{Set of candidate subpopulations $\mathcal{C}$}
$\mathcal{C}_1 \leftarrow \{ \{(A=v)\} \mid A \in \immutable, v \in \text{dom}(A), \text{support}(A=v) \ge \sigma \}$\\
$\mathcal{C} \leftarrow \mathcal{C}_1$\\
$k \leftarrow 2$\\

\While{$\mathcal{C}_{k-1} \neq \emptyset$}{
    $\mathcal{L}_k \leftarrow$ \revc{\textsc{GenerateCandidates}($\mathcal{C}_{k-1}$)}\\
    $\mathcal{C}_k \leftarrow \emptyset$\\
    \ForEach{$c \in \mathcal{L}_k$}{
        $\text{support}(c) \leftarrow \frac{|\{ r \in \db \mid r \text{ satisfies } c \}|}{|\db|}$\\
        \If{$\text{support}(c) \ge \sigma$}{
            $\mathcal{C}_k \leftarrow \mathcal{C}_k \cup \{c\}$\\
        }
    }
    $\mathcal{C} \leftarrow \mathcal{C} \cup \mathcal{C}_k$\\
    $k \leftarrow k + 1$\\
}
\Return{$\mathcal{C}$}\\
\BlankLine
\SetKwFunction{GenCand}{GenerateCandidates}
\SetKwProg{Fn}{Function}{:}{}
\Fn{\GenCand{$\mathcal{C}_{k-1}$}}{
    \Return\ all $k$-itemsets formed by joining pairs of $(k-1)$-itemsets in $\mathcal{C}_{k-1}$ that share the first $k-2$ items, and whose all $(k-1)$ subsets are in $\mathcal{C}_{k-1}$\\
}
\end{algorithm}

\sysName supports various use cases, such as investigating surprising
observations, debugging fairness issues, or identifying reverse trends. Each
scenario may require search for specific subpopulations where the average outcome
for $g_1$ is either higher or lower than that for $g_2$. To accommodate this, we
introduce a filtering step, which retains subpopulations that meet the
relevant condition.

% \begin{figure}[b]
%     \centering    
%     \vspace{3mm}
%     \resizebox{0.75\columnwidth}{!}{
%     \includegraphics[width=\columnwidth]{SIGMOD 2026 - Full Paper/figs/exdis.pdf} 
%     }
%     \vspace{2mm}
%     \caption{\small The \sysName architecture comprising three modules. The user provides a database and two groups of interest to \sysName, which returns a set of disparity explanations of size $k$.}
%     \label{fig:system_architecture}
% \end{figure}

\subsection{Explanation Miner}
\label{subsec:step_2}
For each subpopulation pattern $\pattern_g$ found in the previous step, given an outcome variable $O$, and two groups $g_1$ and $g_2$, the next step is to explain the disparity within the subpopulation, i.e., to identify the treatment pattern $\pattern_e$ with the highest disparity score.
\newtext{To this end, we adapt the treatment mining step of CauSumX~\cite{youngmann2024summarized} to our setting.} CauSumX provides causal explanations for the results of aggregate queries. An explanation pattern consists of a set of tuples from the aggregate view (i.e., the query results), and a treatment pattern is used to quantify the causal effect (in terms of CATE) of the treatment on the outcome within the relevant subview. \newtext{CauSumX employs a heuristic lattice traversal approach to identify promising treatment patterns with high CATE values. However, unlike CauSumX, which estimates the CATE value, we estimate the disparity score of the treatment pattern under consideration.} 
Furthermore, we adjust the search to the context, e.g., when exploring parts of the data where the average outcome for $g_1$ exceeds that of $g_2$, the explanation should elucidate this phenomenon---specifically, we aim to identify treatments that favor $g_1$ over $g_2$, and filter out the others.

% Furthermore, if, for instance, we are exploring parts of the data where the average outcome for $g_1$ exceeds that of $g_2$, the explanation should elucidate this phenomenon---specifically, identifying treatments that favor $g_1$ compared to $g_2$. This can be addressed by filtering the examined treatments in Step 2 of \sysName.

\newtext{We note that CauSumX operates only on aggregate views (i.e., the results of aggregate queries), where considering a relatively small number of grouping patterns is sufficient (the average was 24~\cite{youngmann2024summarized}). In contrast, identifying a treatment for each subpopulation requires searching for subpopulations with significant disparity across the entire dataset. This requires consideration of a much larger number of potential grouping patterns (in our experiments, the average was 184).  
As a result, this module is the primary bottleneck of \sysName. To improve runtime, we introduce the following optimizations:}

\smallskip

\noindent\newtext{\textbf{Limiting patterns.} To ensure conciseness (and thus
interpretability) of explanations while reducing runtime, we restrict the
search for treatment patterns to at most two predicates (similar
to~\cite{agmon2024finding}).} \revc{However, this is \emph{not} a fundamental limitation of our technique. If the user wishes to get longer explanations, \sysName\ can explore patterns with any number of predicates. This, of course, comes with a cost of a longer runtime. In our evaluation, even when allowing 3 predicates, most explanations had at most two predicates. Thus, to optimize the runtime, we limit the number of predicates to 2.}

% In our experiments, increasing the number of patterns beyond two did not affect the results, except for increasing the runtime. Therefore, by default, we set this bound to two, although it can be user-defined.}.

\smallskip

\noindent\textbf{Parallelization.} \reva{The independence of subpopulations allows treatment pattern discovery to be performed in parallel. We exploit this property to parallelize the computation, thereby improving scalability and reducing runtime.}

\smallskip

\noindent\textbf{Caching.} Computing the disparity score for a given treatment
pattern requires adding it as a node to the underlying causal
DAG~\cite{youngmann2024summarized}. \newtext{Often, this may lead to a DAG that
has been previously encountered. To avoid redundant computations, we cache
results related to previously encountered causal DAGs.}

\smallskip

\noindent\textbf{Sampling.} As was done in ~\cite{youngmann2024summarized},
instead of focusing on obtaining precise CATE values, we estimate CATE from a
random sample of the data. We use a fixed sample size of $50{,}000$ tuples,
guided by our empirical findings, which indicates that this sampling size
achieves highly accurate CATE estimations while maintaining a relatively low
runtime. However, the sampling ratio is a customizable system parameter and the
user is free to tune it for more accurate results.

% \textbf{Function:} Find\_Treatment

% \textbf{Input:}
% \begin{itemize}
%     \item \( S \): A set of subpopulations
%     \item \( D \): The database
%     \item \( L_2 \): A set of potential treatments
% \end{itemize}

% \textbf{Output:}
% \begin{itemize}
%     \item \( A \): A set of pairs (facts) \((s, t)\), where \( s \in S \) and \( t \in L_2 \), representing the optimal treatment \( t \) for each subpopulation \( s \) based on CATE.
% \end{itemize}

% \textbf{Algorithm:}
% \begin{enumerate}
%     \item For each \( s \in S \):
%     \begin{itemize}
%         \item Find the treatment \( t \in L_2 \) that maximizes the CATE:
%         \[
%         t = \arg \max_{t \in L_2} \, \text{CATE}(D, G_1, G_2)
%         \]
%     \end{itemize}
%     \item Return \( A = \{(s, t) \mid s \in S\} \)
% \end{enumerate}

\subsection{Greedy Search}
\label{subsec:step_3}

% \brit{
% In step 3 we run before hierarchical clustering based on semantic distance of subpopulations definitions
% after clustering we run greedy based on jaccard distance, and pass the candidates based on $\Delta$ in descending order, and not let the explanations in the results to be from the same cluster}

Given the set of candidate disparity explanations $\{\phi_i\}_{i = 1}^l$ obtained in the previous two steps, our goal is to identify a set of $k$  explanations with the highest disparity scores, adhering to the constraints of Problem~\ref{prblem}. 
\reva{A key challenge is balancing scalability and diversity in selecting candidate disparity explanations. To address this, we introduce a clustering step: rather than evaluating all candidate explanations individually, we group similar subpopulations and assign a representative explanation to each cluster, thereby reducing redundancy and improving efficiency while maintaining coverage of distinct subpopulations.
The rationale is twofold: (1)~it enables us to handle a large number of candidate subpopulations without exhaustive enumeration, and (2)~selecting a representative from each cluster ensures that the final set of explanations is diverse, avoiding multiple explanations that represent the same subpopulation.}

Specifically, we proceed as follows. We first cluster the candidate explanations using a hierarchical clustering algorithm (using Scipy~\cite{virtanen2020scipy} implementation) based on the symmetric difference among the subpopulations. Then, from each cluster, we select a random representative explanation and assign its disparity score to the entire cluster. We then iteratively select $k$ disparity explanations.
At the first iteration, we pick a random explanation from the cluster with the highest disparity score. 
At the $j{-}th$ iteration (for $1<j\leq k$), we select the explanation $\phi^*$ such that:
$$
\phi^* = \argmax_{\phi \in \Phi \land \expSim(\phi, \phi') < \tau } \Delta(\phi), \quad \text{for } \phi' \in \Phi_{j-1}.
$$
% $$
% \phi^* {=} \argmax_{\phi \in \{\phi_i\}_{i=1}^l} f(\Phi_{j-1} \cup \{\phi\})
% $$
where $\Phi$ is the set of candidate explanations that consist of a random explanation from each cluster, and $\Phi_{j-1}$ is the set of explanations selected up to iteration $j$.

\subsubsection*{Complexity analysis}
The maximum number of disparity explanations in a database $\db$ with attributes
$\attrset$ is bounded by $|\db|^{|\attrset|}$ (considering both subpopulation
and treatment patterns), which is polynomial in terms of data complexity,
assuming a fixed schema~\cite{Vardi82}. Our greedy search is also polynomial in
the number of explanations considered. Additional operations, such as
calculating CATE values, are polynomial in $\db$, leading to worst-case
polynomial data complexity. As we demonstrate in \cref{subsec:scalability},
\sysName is capable of efficiently handling large, high-dimensional datasets in
practice.

\begin{table*}[t]
\begin{center}
\caption{\small Details of the datasets for experiments and case studies.}
\vspace{-3mm}	
\resizebox{1\textwidth}{!}{
\begin{tabular}{@{}lrrrllrrlrr@{}}
	
\toprule
\textbf{Dataset} &
\multicolumn{1}{c}{\#Tuples} & 
$|\mathbf{I}|$ & 
$|\mathbf{M}|$ &  
\multicolumn{1}{c}{$g_1$} & 
\multicolumn{1}{c}{$g_2$} & 
$|g_1|$ & 
$|g_2|$ & 
\multicolumn{1}{c}{$O$} & 
\texttt{AVG}$_{g_1}$ & 
\texttt{AVG}$_{g_2}$ \\ 

\midrule

Stack Overflow (SO)~\cite{stackoverflowreport} 
& 47,702 & 4 & 6 
& Data/business analysts & Back-end developers 
& 4,088 & 28,987 
& Total Compensation (TC) & \$106K & \$96K \\
%%%%
American Community (ACS)~\cite{ACS_Data} 
& 1,420,652 & 15 & 10 
& Manual labor & Overall data 
& 99,790 & 1,420,652 
& Likelihood of having a health insurance & 78.6\% & 91.5\% \\
%%%%
Medical Expenditure Panel (MEPS)~\cite{MEPS_Data_Overview} 
& 13,528 & 7 & 5 
& Males & Non-males 
& 5,731 & 7,797 
& Likelihood of feeling nervous frequently & 41.6\% & 46.9\% \\

\bottomrule
\end{tabular}}
\vspace{-3mm}
\label{tab:datasets}
\end{center}
\end{table*}

\section{Experimental Evaluation}
\label{sec:exp}
We present an experimental evaluation of the effectiveness of
\sysName in practical settings. We aim to address the following questions:

\begin{itemize}[leftmargin=*,topsep=0pt]

\item \textbf{Q1}: How does the quality of \sysName-generated disparity
explanations compare to that of existing methods?
(Section~\ref{subsec:exp_quality})

\item \textbf{Q2}: How is the quality of the explanations
affected by various parameters, and how to tune the system parameters?
(Section~\ref{subsec:parameters})

\item \textbf{Q3}: How efficient and scalable is \sysName?
(Section~\ref{subsec:scalability})

\item \textbf{Q4}: How do our proposed optimizations affect \sysName' runtime
performance? (Section~\ref{subsec:ablation})

\end{itemize}

\subsection{Experimental Setup} \label{subsec:exp_setup}
All experiments were performed on a Windows computer, Intel CPU, with 16 GB
memory. We implemented \sysName in Python 3 and
used DoWhy library~\cite{dowhypaper} to compute the CATE values. CATE values are estimated using a linear regression model. For simplicity of visualization, we omit $p$-values, but note that all reported causal effects are statistically significant (i.e., $p < 0.05$). Our source code is publicly available~\cite{repolink}.

\subsubsection{Datasets, causal DAGs, and preprocessing} We used three popular
datasets (\cref{tab:datasets}) and obtained the corresponding causal DAGs given
in prior work~\cite{youngmann2023causal}. To process continuous
numerical attributes, we applied equal-width binning across 10 bins.

\smallskip

\noindent\textbf{SO:} The Stack Overflow Developer
Survey~\cite{stackoverflowreport} dataset contains responses from developers
worldwide, covering topics such as professional experience, education,
technologies used, and employment-related information, such as annual total
compensation (TC).

\smallskip

\noindent\textbf{ACS:} The American Community Survey (ACS)~\cite{ACS_Data} is a
nationwide survey conducted by the U.S. Census Bureau, with demographic, social,
economic, and housing data. We focused on 7 states: California, Texas, Florida,
New York, Pennsylvania, Illinois, and Ohio.

\smallskip

\noindent\textbf{MEPS:} The Medical Expenditure Panel Survey
(MEPS)~\cite{MEPS_Data_Overview} dataset provides information on healthcare
utilization, expenditures, insurance coverage, and demographics of individuals
in the U.S.

\subsubsection{System parameters} Unless otherwise specified, we used the following default parameters: the explanation set size $k = 5$; minimum support threshold $\sigma = 0.05$ (considering only groups covering at least 5\% of the data); maximum similarity threshold for the diversity constraint $\tau = 0.55$ based on Jaccard similarity; and $10$ clusters in the greedy search phase.

% \begin{itemize}[leftmargin=*,topsep=0pt]
    
%     \item The desired size of the explanation set $k = 5$.

%     \item The minimum support threshold for the support constraint: $\sigma =
%     0.05$. This means that we consider only groups that account for at least 5\%
%     of the data.
	
%     \item The maximum similarity threshold for diversity constraint: $\tau = 0.55$, using Jaccard similarity.

%     \item The number of clusters for the greedy search phase = $10$.
    
% \end{itemize}

\subsubsection{Use cases} Throughout our experimental evaluation, we examine
three scenarios that represent different use cases of \sysName (as mentioned in
Section~\ref{sec:introduction}). The description of the groups, the outcome
variables, and relevant statistics are given in Table \ref{tab:datasets}.

\smallskip \noindent\textbf{(1)~Investigating a surprising fact.} \looseness-1
In tech, developers generally earn more than analysts. However, an analysis of
the SO dataset revealed a surprising trend: data analysts ($g_1$)
earn more on average than backend developers ($g_2$). To analyze this, we fix
the outcome $O$ as total compensation (TC). We aim to identify which
subpopulations significantly contribute to this disparity and the
underlying explanations in terms of treatments that favor $g_1$ over $g_2$.

\smallskip\noindent\textbf{(2)~Fairness debugging.}
In the ACS dataset, we observe people in certain type of occupation to have
lower than average rate of health-insurance coverage. Specifically, we fix the
outcome $O$ to indicate whether an individual holds a health insurance; $g_1$
represents individuals employed in manual-labor occupations (cleaning,
maintenance, farming, fishing, construction, etc.), where the health insurance
coverage rate is only 78.6\%; and $g_2$ represents the entire dataset across all
occupations, with a coverage rate of 91.5\%. Our objective is to investigate the
underlying factors contributing to this discrepancy by identifying
subpopulations for which the average insurance coverage rate of $g_1$ is
significantly lower than that of $g_2$, along with a causal explanation specific
to each subpopulation.

\smallskip\noindent\textbf{(3)~Finding reverse trends.}
We investigate an intriguing observation in the MEPS dataset: while typically
males ($g_1$) have a lower likelihood of feeling nervous frequently than
non-males ($g_2$), our analysis reveals subpopulations where a reverse trend
holds. To dig deeper, we fix the outcome $O$ to indicate whether an individual
feels nervous frequently. Our goal is to pinpoint subpopulations where the
relative trend involving $g_1$ and $g_2$ is reverse compared to the global trend
and explore the underlying causes.

\subsubsection{Baselines} \label{sec:six:one} We consider the following baselines:

\noindent\textbf{Brute Force.} We employ an exhaustive Brute Force
algorithm, which considers all possible $k$-sized explanation sets. Note that in the
absence of an absolute ground-truth, the results obtained from this
technique can be treated as the optimal solution for
Problem~\ref{prblem}.

\noindent\textbf{Top-$\mathbf{k}$.} This baseline ignores the
diversity constraint and simply returns the top-k explanations ranked by their
disparity scores.

\noindent\textbf{XInsight.}
XInsight~\cite{abs-2207-12718} is designed to identify both causal and
non-causal patterns to explain disparities between two groups in aggregate SQL
queries. Unlike \sysName, which provides local (possibly different) explanations
for each subpopulation, XInsight provides a single global explanation for the
entire data.
Since XInsight includes a causal discovery phase, we ensure a fair comparison as
follows: for each treatment pattern (i.e., explanation) identified by \sysName,
we report its causal effect over the entire dataset (in the ``Global'' column in
Tables~\ref{tab:results_so}--\ref{tab:results_meps}). We aim to empirically
demonstrate that subpopulation-specific explanations may not be valid globally.

%!TEX root=main.tex

%%%%%%%%%%%%%%%%%%%%%%%%%%%%%%%%%%%%%%%%%%%%%%%%%%%%%%%%%%%%%%%%%%%%%%%%%%%%%%%%
\begin{table*}[t] 
	
	\caption{\small Overall disparity scores ($\Delta$), runtimes, and diversity
visualizations for explanations generated by various baselines and \sysName
across three use cases. We report the disparity scores w.r.t the Brute Force
baseline since it gives the ground truths. In the heat-map, the diagonal is
black, denoting 0 self distance from an explanation to itself. Lighter colors
denote less similar pair of explanations, offering diversity.}
\vspace{-2mm}
\centering
\resizebox{\textwidth}{!}{
    \begin{tabular}{l
				r@{\hskip 3mm}r@{\hskip 3mm}cc
                r@{\hskip 3mm}r@{\hskip 3mm}cc
                r@{\hskip 3mm}r@{\hskip 3mm}cc
		}
        \toprule
    	\multirow{2}{*}{\diagbox{\textbf{Approach}}{\textbf{Dataset}}}  & 
    	\multicolumn{4}{c}{\cellcolor{blue!10}\textbf{SO}} &
    	\multicolumn{4}{c}{\cellcolor{red!20}\textbf{ACS}} &
    	\multicolumn{4}{c}{\cellcolor{green!5}\textbf{MEPS}} \\
        
        & 
	\cellcolor{blue!10}{\textbf{$\Delta$ (\%)}} &
        \cellcolor{blue!10}{\textbf{Runtime (s)}} &
        \cellcolor{blue!10}{\textbf{\#Explanations}} &
        \cellcolor{blue!10}{\textbf{Diversity}} &
        
        \cellcolor{red!20}{\textbf{$\Delta$ (\%)}} &
        \cellcolor{red!20}{\textbf{Runtime (s)}} &
        \cellcolor{red!20}{\textbf{\#Explanations}} &
        \cellcolor{red!20}{\textbf{Diversity}} &
        
        \cellcolor{green!5}{\textbf{$\Delta$ (\%)}} &
        \cellcolor{green!5}{\textbf{Runtime (s)}} &
        \cellcolor{green!5}{\textbf{\#Explanations}} &
        \cellcolor{green!5}{\textbf{Diversity}} \\
        
        \midrule
        
        \textbf{Brute Force}     
        & \cellcolor{blue!10}{100} & \cellcolor{blue!10}{181} & \cellcolor{blue!10}{5} & 	\cellcolor{blue!10}\raisebox{-3mm}
        {\includegraphics[width=0.06\textwidth]{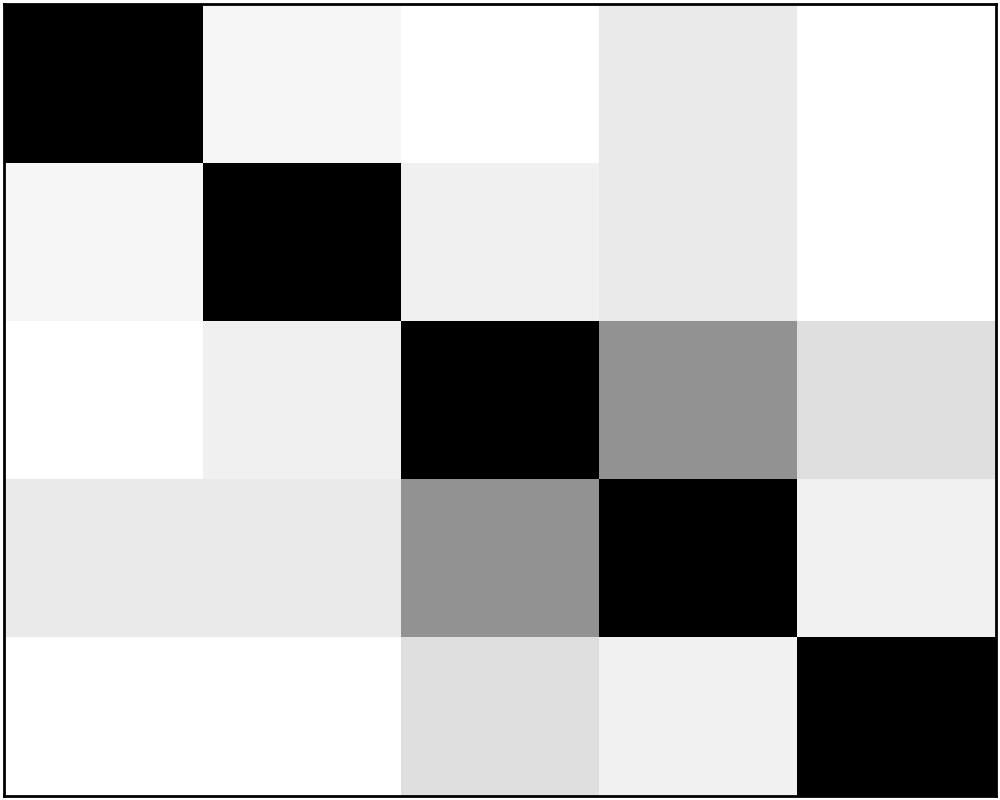}}
        & \cellcolor{red!20}{100} & \cellcolor{red!20}{3130} & \cellcolor{red!20}{5} & 		\cellcolor{red!20}\raisebox{-3mm}{\includegraphics[width=0.06\textwidth]{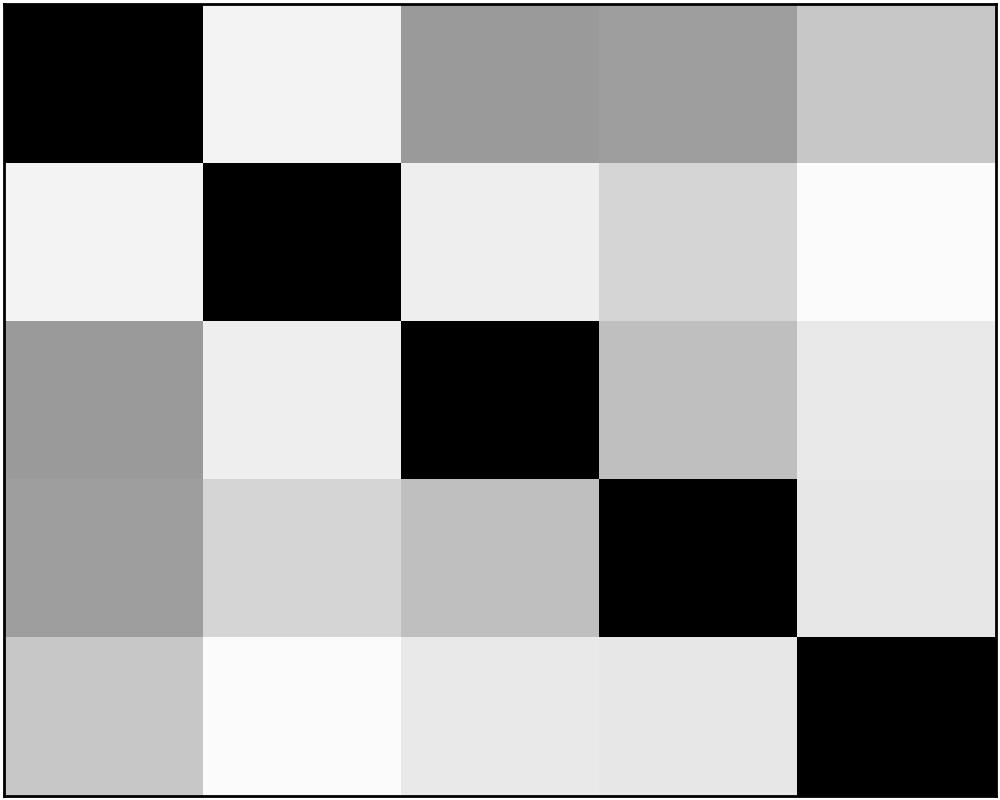}}
        & \cellcolor{green!5}{100} & \cellcolor{green!5}{19} & \cellcolor{green!5}{5} & 	\cellcolor{green!5}\raisebox{-3mm}{\includegraphics[width=0.06\textwidth]{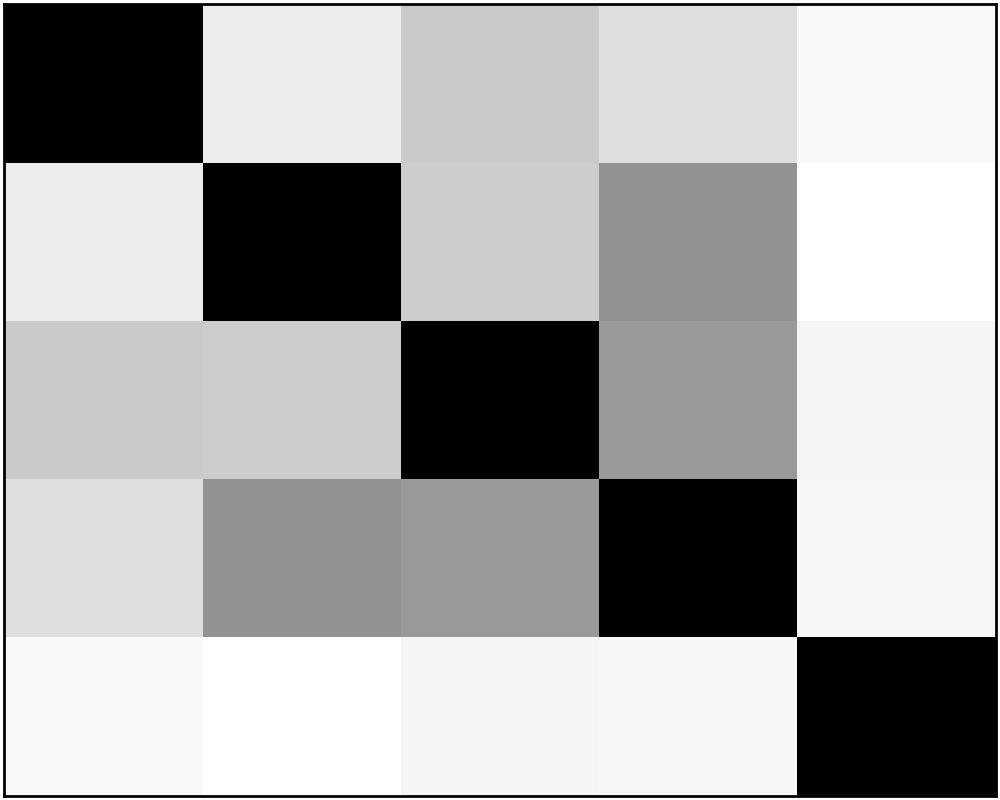}} \\[3mm]
                
        \textbf{Top-K}           
        & \cellcolor{blue!10}{118} & \cellcolor{blue!10}{180} & \cellcolor{blue!10}{5} & 		\cellcolor{blue!10}\raisebox{-3mm}{\includegraphics[width=0.06\textwidth]{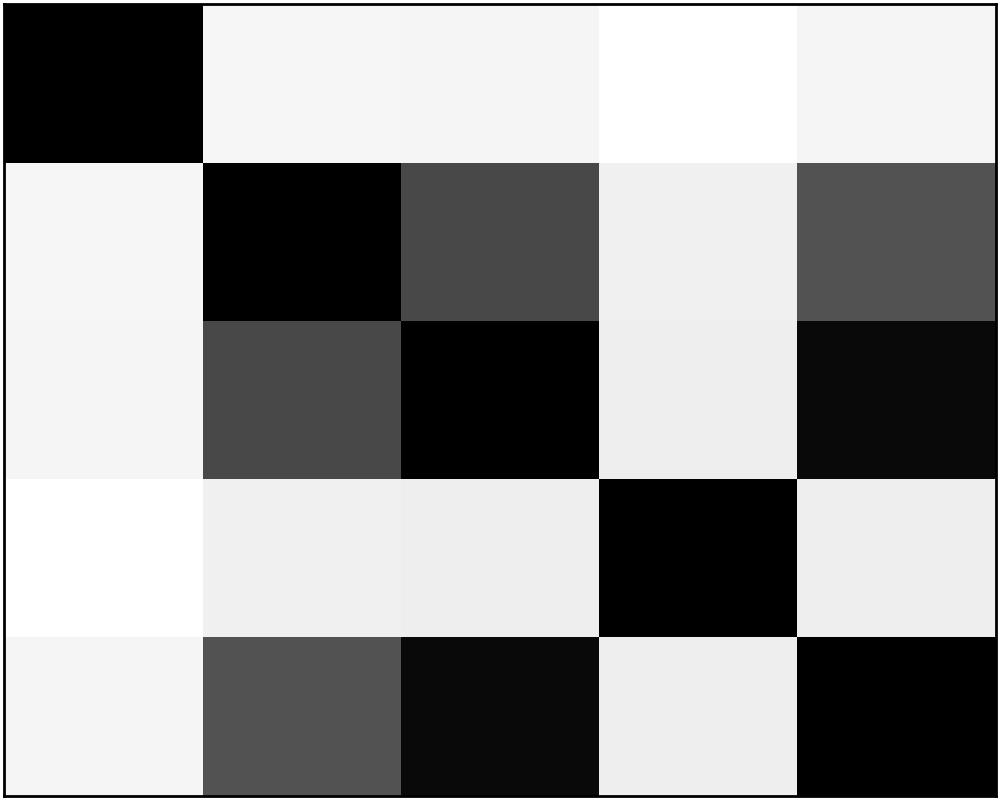}}
        & \cellcolor{red!20}{121} & \cellcolor{red!20}{1514} & \cellcolor{red!20}{5} & 		\cellcolor{red!20}\raisebox{-3mm}{\includegraphics[width=0.06\textwidth]{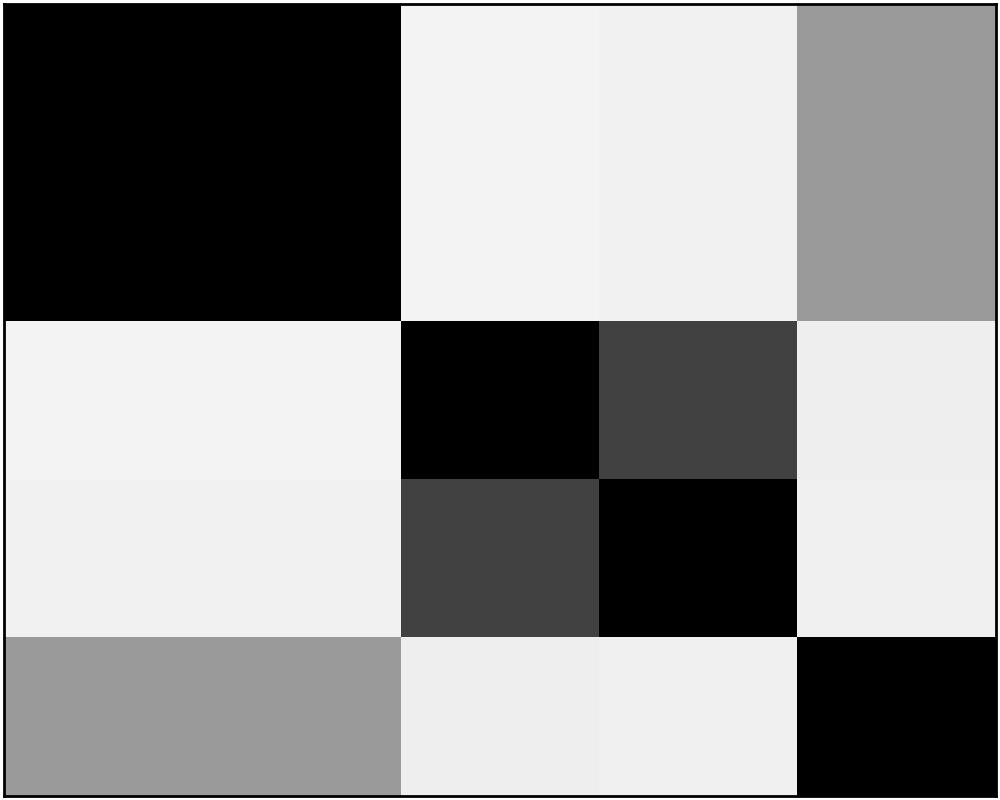}}
        & \cellcolor{green!5}{166} & \cellcolor{green!5}{14} & \cellcolor{green!5}{5} & 	\cellcolor{green!5}\raisebox{-3mm}{\includegraphics[width=0.06\textwidth]{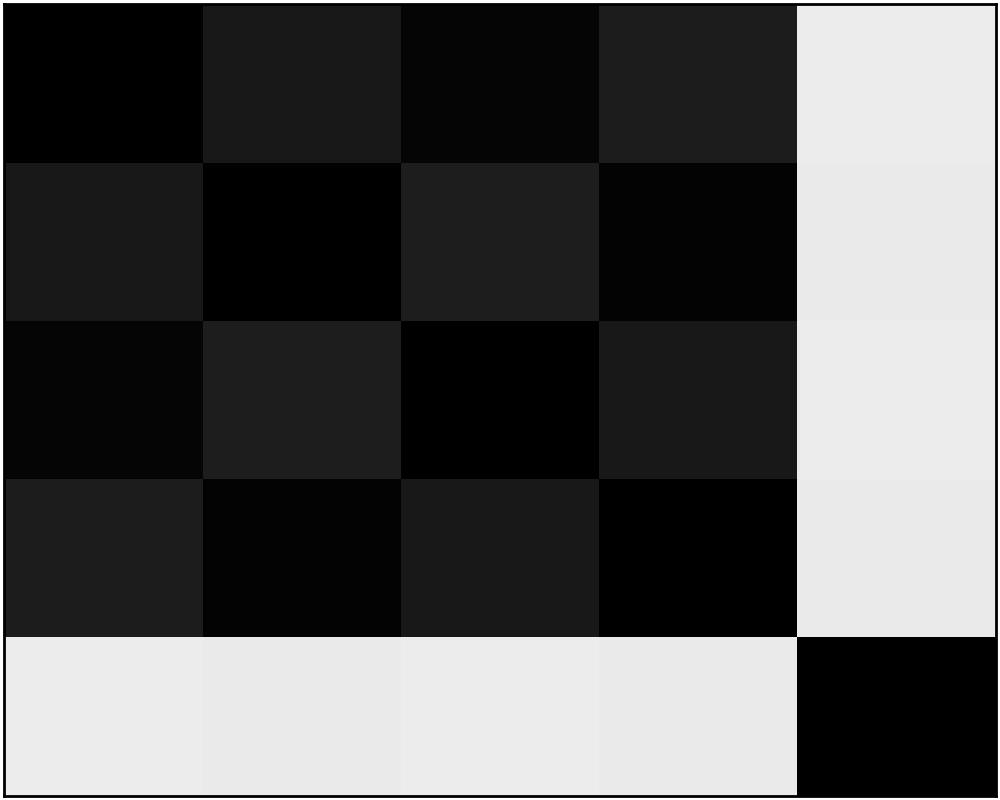}} \\[3mm]
        
        \textbf{DivExplorer}     
        & \cellcolor{blue!10}{N/A} & \cellcolor{blue!10}{33} & \cellcolor{blue!10}{0} & 		\cellcolor{blue!10}{N/A}
        & \cellcolor{red!20}{86} & \cellcolor{red!20}{1055} & \cellcolor{red!20}{5} & 		\cellcolor{red!20}\raisebox{-3mm}{\includegraphics[width=0.06\textwidth]{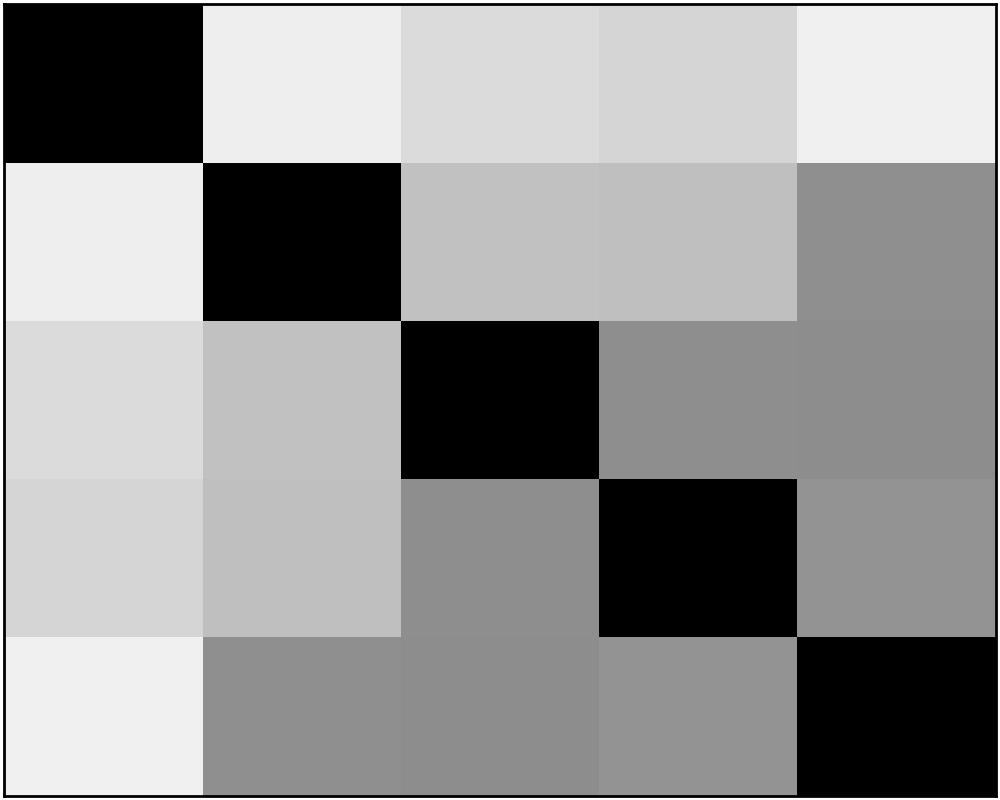}}
        & \cellcolor{green!5}{N/A} & \cellcolor{green!5}{4} & \cellcolor{green!5}{0} &	  	\cellcolor{green!5}{N/A}\\[3mm]
        
        \textbf{FairDebugger}    
        & \cellcolor{blue!10}{N/A} & \cellcolor{blue!10}{2258} & \cellcolor{blue!10}{0} & 	\cellcolor{blue!10}{N/A}
        & \cellcolor{red!20}{26} & \cellcolor{red!20}{717} & \cellcolor{red!20}{2} & 			\cellcolor{red!20}\raisebox{-3mm}{\includegraphics[width=0.06\textwidth]{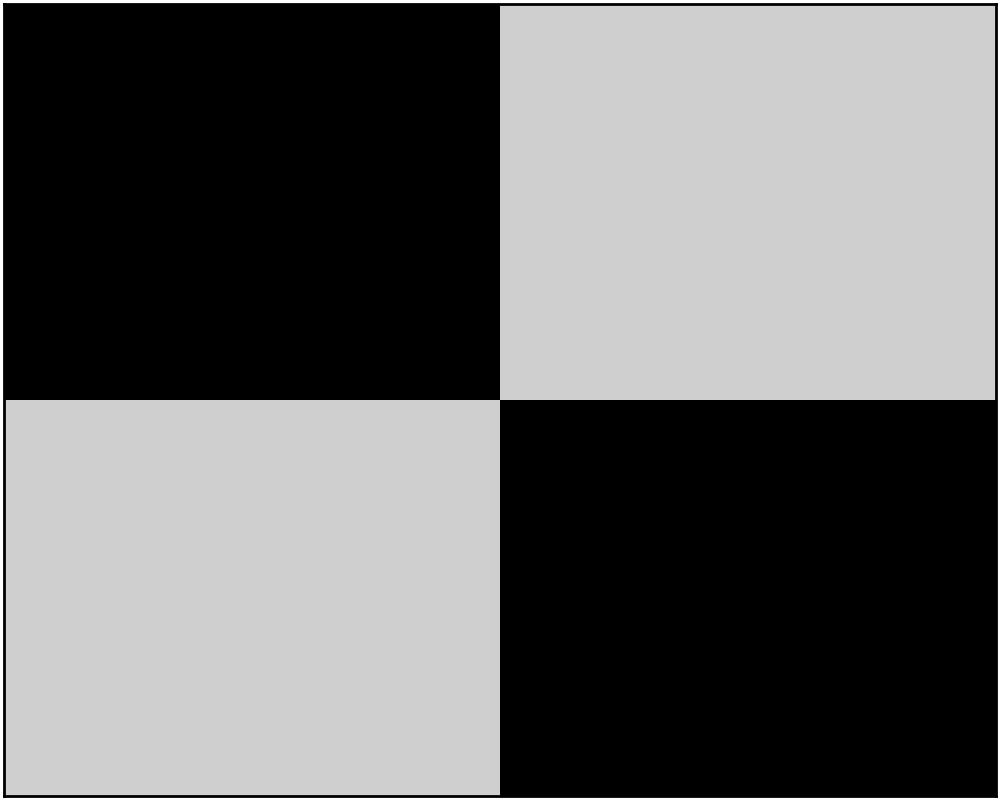}}
        & \cellcolor{green!5}{N/A} & \cellcolor{green!5}{103} & \cellcolor{green!5}{0} & 		\cellcolor{green!5}{N/A} \\[3mm]
		
        \textbf{\sysName (this paper)}        
        & \cellcolor{blue!10}{55} & \cellcolor{blue!10}{62} & \cellcolor{blue!10}{5} & 		\cellcolor{blue!10}\raisebox{-3mm}{\includegraphics[width=0.06\textwidth]{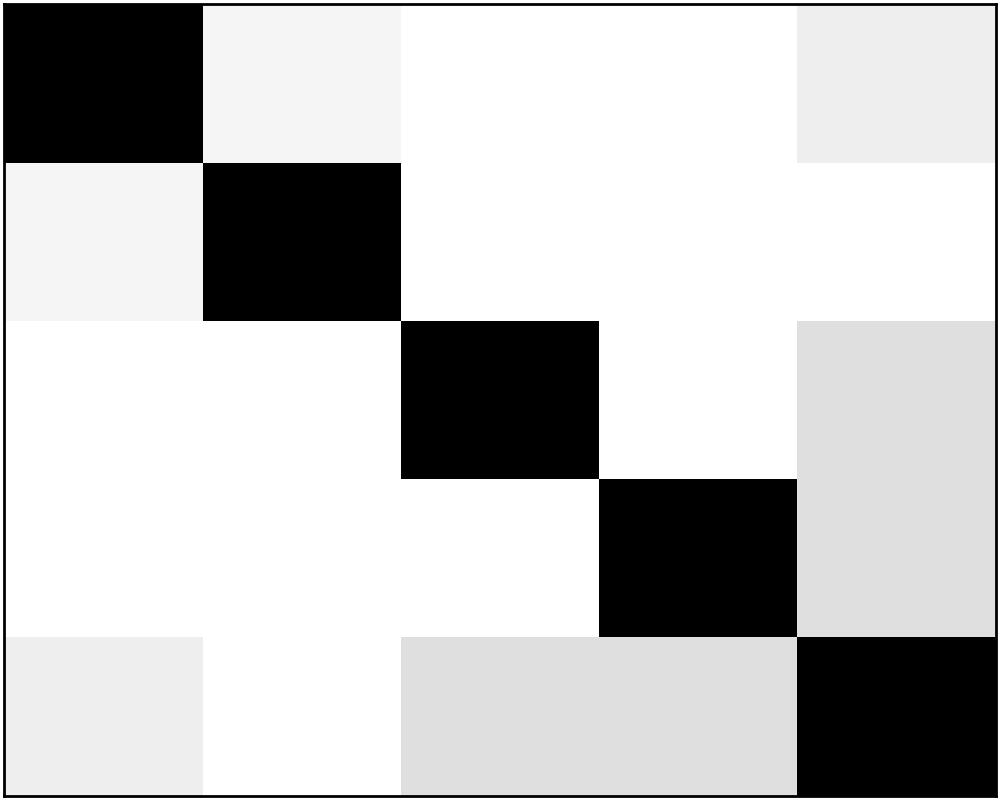}}
        & \cellcolor{red!20}{84} & \cellcolor{red!20}{1170} & \cellcolor{red!20}{5} & 		\cellcolor{red!20}
        \raisebox{-3mm}{\includegraphics[width=0.06\textwidth]{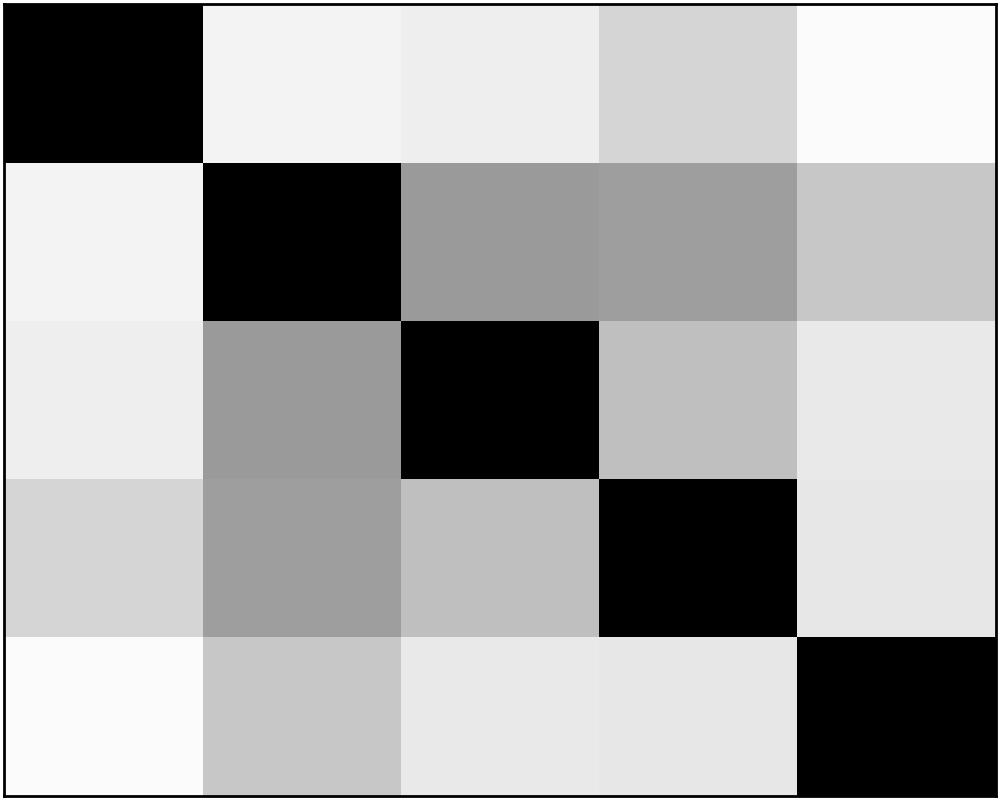}}
        & \cellcolor{green!5}{100} & \cellcolor{green!5}{17} & \cellcolor{green!5}{5} & 	\cellcolor{green!5}\raisebox{-3mm}{\includegraphics[width=0.06\textwidth]{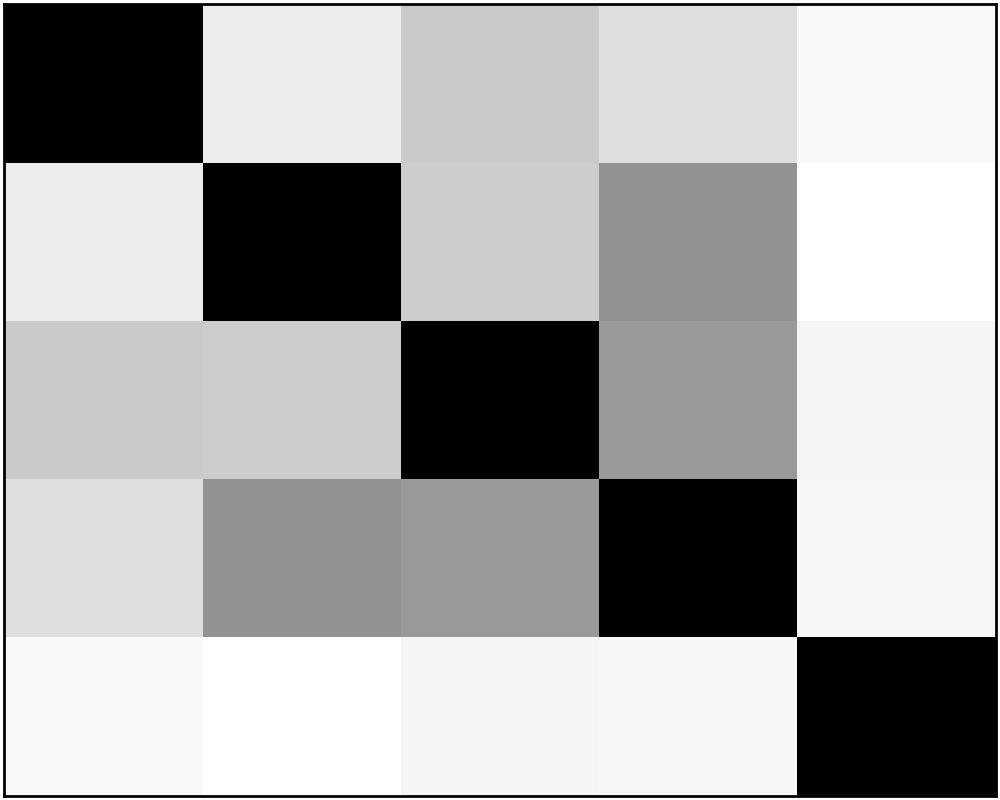}} \\
		        
        \bottomrule
    \end{tabular}
}
\label{tab:quality_metrics_new}
\end{table*}
%%%%%%%%%%%%%%%%%%%%%%%%%%%%%%%%%%%%%%%%%%%%%%%%%%%%%%%%%%%%%%%%%%%%%%%%%%%%%%%%

\noindent\textbf{DivExplorer.}
DivExplorer~\cite{pastor2021looking} analyzes the behavior of classification models to identify data regions where a performance metric (e.g., false positive rate) deviates significantly from its value over the entire dataset. Given a divergence metric, it identifies data subsets, defined by patterns, where the metric shows a significant disparity compared to the overall population. We use
DivExplorer as a baseline, setting the divergence function to the disparity
score. To adapt it to our setting, we use a fixed treatment. This treatment is selected as the one that yields
the highest disparity score between the two groups of interest in the
overall data. We then used the last step of ExDis to find a $k$-size solution.

% compute the disparity scores for each subpopulation and return
% the top-k subpopulations with the highest scores.
% \tal{instead took the top-k I improved that and as in the brute force - return the group with highest score from all combinations in size k}

\noindent\textbf{FairDebugger.} \looseness-1
FairDebugger~\cite{DBLP:journals/corr/abs-2402-05007} identifies training data
subsets that contribute to fairness violations in random forest models by
evaluating how their removal affects model outcomes. To adapt FairDebugger to
our setting, we fix the treatment pattern to the one exhibiting most 
disparity across the entire dataset. We then assess the influence of each
subpopulation by measuring the change in disparity after removing it. The
difference in disparity scores quantifies the subpopulation's contribution to
the overall disparity.  We then used the last step of ExDis to find a $k$-size solution.

% FairDebugger returns the top-k highest contributing
% subpopulations.
% \tal{same as DivExplorer baseline - instead took the top-k I improved that and as in the brute force - return the group with highest score from all combinations in size k}

% \underline{Why Not Compare Against CauSumX?}	
% \newtext{We do not consider CauSumX \cite{youngmann2024summarized} as a
% baseline for the following reasons: \textbf{(i) Lack of support for
% overlapping groups:} CauSumX does not support overlapping groups, which is
% essential for our second use case (Fairness Debugging). It can only generate
% explanations for groups within the results of a group-by query. \textbf{(ii)
% Focus on outcome influence rather than disparity:} The explanations provided
% by CauSumX focus on identifying factors that mostly affect the outcome.
% However, our work aims to explain what causes disparities between groups,
% making CauSumX unsuitable for our purpose.}

%!TEX root=main.tex

\begin{table*}[t]
\centering
\caption{\small Disparity explanations discovered by \sysName for the SO
dataset. Subpopulation patterns  are highlighted in \colorbox{Orange!30}{\strut{Orange}},
while treatment patterns are highlighted in \colorbox{Cyan!30}{\strut{Cyan}}. We highlight the two groups
of interest using \colorbox{Yellow!70}{\strut{Yellow}} and \colorbox{Lavender!70}{\strut{Pink}}.
The first explanation highlights that for
\colorbox{Orange!30}{\strut{White heterosexual individuals whose parents attended
secondary school}}, the average TC for \colorbox{Yellow!70}{\strut
\texttt{analysts}} observes an increase of \$154,024 when the treatment
\colorbox{Cyan!30}{\strut{hoping to become a manager in the next 5 years}} is applied. In
contrast, the same treatment yields only an increase of \$31,354 for
\colorbox{Lavender!70}{\strut \texttt{back-end developers}}. We observe that
this treatment is not a good explanation globally due to not being statistically
significant.}
% \af{Update this caption and SO results after the exercise treatment
% is removed from the causal DAG.}
%
\vspace{-3mm}
\resizebox{1\textwidth}{!}{
\renewcommand{\arraystretch}{1.1}
\begin{tabular}[t]{|p{80mm}|r|>{\centering}m{25mm}c|@{}>{\centering}m{25mm}c|r|}
\hline
\multirow{3}{*}{\textbf{Disparity Explanation}} & 
\multirow{3}{*}{\textbf{Support}} & 
\multicolumn{4}{c|}{\textbf{Total Compensation (TC)}} &
\multirow{3}{*}{\textbf{$\Delta$}\phantom{D}} \\
\cline{3-6}
& &
\multicolumn{2}{c|}{\textbf{Subpopulation (\sysName)}} &
\multicolumn{2}{c|}{\textbf{Global (XInsight)}}& \\
\cline{3-6}
&& 
\textbf{Average} & \textbf{CATE} & 
\textbf{Average} & \textbf{CATE} &\\
\hline
\hline
% Row 1
\begin{tabular}{@{}p{80mm}}
For \colorbox{Orange!30}{\strut{White heterosexual individuals whose parents attended}} \colorbox{Orange!30}{\strut{ secondary school}}, TC growth is more influenced by \colorbox{Cyan!30}{\strut{hoping to become a manager in the next 5 years}} for \colorbox{Yellow!70}{\strut{analysts}} compared to \colorbox{Lavender!70}{\strut{back-end developers}}.
\end{tabular}
&
11.14\%&
\multicolumn{2}{l|}{
\begin{tabular}{rr}
\$129,749 \hspace{2mm}
\resizebox{.08\textwidth}{!}{
\raisebox{-0.3\height}{
\begin{tikzpicture}
\draw[fill=Yellow!70] (0,-25) rectangle (72, -45);
\draw[fill=Yellow!10] (72, -25) rectangle (100, -45);
\end{tikzpicture}}}
& \$154,024 \uparrowbold
\\
\$110,026 \hspace{2mm}
\resizebox{.08\textwidth}{!}{
\raisebox{-0.3\height}{
\begin{tikzpicture}
\draw[fill=Lavender!70] (0,-25) rectangle (61, -45);
\draw[fill=Lavender!10] (61, -25) rectangle (100, -45);
\end{tikzpicture}}}
& \$31,354 \uparrowbold\\
\end{tabular}
}
&
\multicolumn{2}{c|}{not statistically significant}
& 0.061\\
\hline
% Row 2
\begin{tabular}{@{}p{80mm}}
For \colorbox{Orange!30}{\strut{White males aged between 18-24 years old}}, TC growth is more influenced by \colorbox{Cyan!30}{\strut{have 3-5 years in professional coding}} for \colorbox{Yellow!70}{\strut{analysts}} compared to \colorbox{Lavender!70}{\strut{back-end developers}}.
\end{tabular}
&
10.64\%&
\multicolumn{2}{l|}{
\begin{tabular}{rr}
\phantom{5}\$90,909 \hspace{2mm}
\resizebox{.08\textwidth}{!}{
\raisebox{-0.3\height}{
\begin{tikzpicture}
\draw[fill=Yellow!70] (0,-25) rectangle (50, -45);
\draw[fill=Yellow!10] (50, -25) rectangle (100, -45);
\end{tikzpicture}}}
& \phantom{5}\$75,692 \uparrowbold
\\
\$67,754 \hspace{2mm}
\resizebox{.08\textwidth}{!}{
\raisebox{-0.3\height}{
\begin{tikzpicture}
\draw[fill=Lavender!70] (0,-25) rectangle (37, -45);
\draw[fill=Lavender!10] (37, -25) rectangle (100, -45);
\end{tikzpicture}}}
& \$24,164 \uparrowbold\\
\end{tabular}
}
&
\multicolumn{2}{c|}{not statistically significant}
& 0.025\\
\hline
% Row 3
\begin{tabular}{@{}p{80mm}}
For \colorbox{Orange!30}{\strut{males aged between 25-34 years old whose parents hold a}} \colorbox{Orange!30}{\strut{ bachelor degree}}, TC growth is more influenced by \colorbox{Cyan!30}{\strut{working in a company size between 100 - 499 workers}} for \colorbox{Yellow!70}{\strut{analysts}} compared to \colorbox{Lavender!70}{\strut{back-end developers}}.
\end{tabular}
&
13.21\%&
\multicolumn{2}{l|}{
\begin{tabular}{rr}
\$105,694 \hspace{2mm}
\resizebox{.08\textwidth}{!}{
\raisebox{-0.3\height}{
\begin{tikzpicture}
\draw[fill=Yellow!70] (0,-25) rectangle (58, -45);
\draw[fill=Yellow!10] (58, -25) rectangle (100, -45);
\end{tikzpicture}}}
& \phantom{5}\$70,069 \uparrowbold
\\
\$96,085 \hspace{2mm}
\resizebox{.08\textwidth}{!}{
\raisebox{-0.3\height}{
\begin{tikzpicture}
\draw[fill=Lavender!70] (0,-25) rectangle (53, -45);
\draw[fill=Lavender!10] (53, -25) rectangle (100, -45);
\end{tikzpicture}}}
& \$19,807 \uparrowbold\\
\end{tabular}
}
&
\multicolumn{2}{c|}{not statistically significant}
& 0.025\\
\hline
% Row 4
\begin{tabular}{@{}p{80mm}}
For \colorbox{Orange!30}{\strut{White heterosexual males aged between 25-34 years old}} \colorbox{Orange!30}{\strut{ whose parents hold a Master degree}}, TC growth is more influenced by \colorbox{Cyan!30}{\strut{not having coding as a hobby}} for \colorbox{Yellow!70}{\strut{analysts}} compared to \colorbox{Lavender!70}{\strut{back-end developers}}.
\end{tabular}
&
7.56\%&
\multicolumn{2}{l|}{
\begin{tabular}{rr}
\$117,879 \hspace{2mm}
\resizebox{.08\textwidth}{!}{
\raisebox{-0.3\height}{
\begin{tikzpicture}
\draw[fill=Yellow!70] (0,-25) rectangle (65, -45);
\draw[fill=Yellow!10] (65, -25) rectangle (100, -45);
\end{tikzpicture}}}
& \phantom{5}\$84,648 \uparrowbold
\\
\$103,957 \hspace{2mm}
\resizebox{.08\textwidth}{!}{
\raisebox{-0.3\height}{
\begin{tikzpicture}
\draw[fill=Lavender!70] (0,-25) rectangle (57, -45);
\draw[fill=Lavender!10] (57, -25) rectangle (100, -45);
\end{tikzpicture}}}
& \$43,292 \uparrowbold\\
\end{tabular}
}
&
\multicolumn{2}{c|}{not statistically significant}
& 0.020\\
\hline
% Row 5
\begin{tabular}{@{}p{80mm}}
For \colorbox{Orange!30}{\strut{White individuals aged between 25–34}}, TC growth is more influenced by \colorbox{Cyan!30}{\strut{working 6-8 years in professional coding}} for \colorbox{Yellow!70}{\strut{analysts}} compared to \colorbox{Lavender!70}{\strut{back-end developers}}.
\end{tabular}
&
34.69\%&
\multicolumn{2}{l|}{
\begin{tabular}{rr}
\$115,777 \hspace{2mm}
\resizebox{.08\textwidth}{!}{
\raisebox{-0.3\height}{
\begin{tikzpicture}
\draw[fill=Yellow!70] (0,-25) rectangle (63, -45);
\draw[fill=Yellow!10] (63, -25) rectangle (100, -45);
\end{tikzpicture}}}
& \phantom{5}\$44,058 \uparrowbold
\\
\$105,988 \hspace{2mm}
\resizebox{.08\textwidth}{!}{
\raisebox{-0.3\height}{
\begin{tikzpicture}
\draw[fill=Lavender!70] (0,-25) rectangle (58, -45);
\draw[fill=Lavender!10] (58, -25) rectangle (100, -45);
\end{tikzpicture}}}
& \$10,552 \uparrowbold\\
\end{tabular}
}
&
\multicolumn{2}{c|}{not statistically significant}
& 0.016\\
\hline
\end{tabular}}
\label{tab:results_so}
\end{table*}
%!TEX root=main.tex

\begin{table*}[t]
\caption{\small Disparity explanations discovered by \sysName for the ACS dataset.}
\vspace{-3mm}
\centering
\resizebox{0.97\textwidth}{!}{
\renewcommand{\arraystretch}{1.1}
\begin{tabular}[t]{|p{87mm}|r|>{\centering}m{25mm}c|@{}>{\centering}m{25mm}c|r|}
\hline
\multirow{3}{*}{\textbf{Disparity Explanation}} & 
\multirow{3}{*}{\textbf{Support}} & 
\multicolumn{4}{c|}{\textbf{Likelihood of having a health insurance}} &
\multirow{3}{*}{\textbf{$\Delta$}\phantom{D}} \\
\cline{3-6}
& &
\multicolumn{2}{c|}{\textbf{Subpopulation (\sysName)}} &
\multicolumn{2}{c|}{\textbf{Global (XInsight)}}& \\
\cline{3-6}
&& 
\textbf{Average} & \textbf{CATE} & 
\textbf{Average} & \textbf{CATE} &\\
\hline
\hline
% Row 1
\begin{tabular}{@{}p{87mm}}
For \colorbox{Orange!30}{\strut{White individuals from the Southern region who speak Spanish}}, the likelihood of having a health insurance decreases when they \colorbox{Cyan!30}{\strut{have no personal earnings}} for \colorbox{Yellow!70}{\strut{manual labor occupations}}, whereas it increases for \colorbox{Lavender!70}{\strut{all occupations}}.
\end{tabular}
&
8.18\%&
\multicolumn{2}{l|}{
\begin{tabular}{p{25mm}r}
\drawbar{52.42}{7.74}{\%}{Yellow}{\downarrowbold}\\
\drawbar{75.17}{4.87}{\%}{Lavender}{\uparrowbold}\\
\end{tabular}
}
&
\multicolumn{2}{l|}{
\begin{tabular}{p{25mm}r}
\drawbar{78.61}{8.02}{\%}{Yellow}{\downarrowbold}\\
\drawbar{91.58}{3.63}{\%}{Lavender}{\downarrowbold}\\
\end{tabular}
}
& 0.126\\
\hline
% Row 2
\begin{tabular}{@{}p{87mm}}
For \colorbox{Orange!30}{\strut{individuals from the Southern region who were born in USA}}, the likelihood of having a health insurance decreases when they \colorbox{Cyan!30}{\strut{have no personal earnings}} for \colorbox{Yellow!70}{\strut{manual labor occupations}}, whereas it increases for \colorbox{Lavender!70}{\strut{all occupations}}.
\end{tabular} &
25.09\% &
\multicolumn{2}{l|}{
\begin{tabular}{p{25mm}r}
\drawbar{72.77}{6.38}{\%}{Yellow}{\downarrowbold}\\
\drawbar{81.24}{1.83}{\%}{Lavender}{\uparrowbold}\\
\end{tabular}
}
&
\multicolumn{2}{l|}{
\begin{tabular}{p{25mm}r}
\drawbar{78.61}{8.02}{\%}{Yellow}{\downarrowbold}\\
\drawbar{91.58}{3.63}{\%}{Lavender}{\downarrowbold}\\
\end{tabular}
} & 0.082 \\
\hline
% Row 3
\begin{tabular}{@{}p{87mm}}
For \colorbox{Orange!30}{\strut{White natives individuals from Texas who were born in USA}}, the likelihood of having a health insurance decreases when they \colorbox{Cyan!30}{\strut{have no personal earnings}} for \colorbox{Yellow!70}{\strut{manual labor occupations}}, whereas it increases for \colorbox{Lavender!70}{\strut{all occupations}}.
\end{tabular} &
12.21\% &
\multicolumn{2}{l|}{
\begin{tabular}{p{25mm}r}
\drawbar{71.86}{5.20}{\%}{Yellow}{\downarrowbold}\\
\drawbar{80.52}{2.83}{\%}{Lavender}{\uparrowbold}\\
\end{tabular}
}
&
\multicolumn{2}{l|}{
\begin{tabular}{p{25mm}r}
\drawbar{78.61}{8.02}{\%}{Yellow}{\downarrowbold}\\
\drawbar{91.58}{3.63}{\%}{Lavender}{\downarrowbold}\\
\end{tabular}
} & 0.080 \\
\hline
% Row 4
\begin{tabular}{@{}p{87mm}}
For \colorbox{Orange!30}{\strut{White males from the Southern region}}, the likelihood of having a health insurance decreases when they \colorbox{Cyan!30}{\strut{have no personal earnings}} for \colorbox{Yellow!70}{\strut{manual labor occupations}}, whereas it increases for \colorbox{Lavender!70}{\strut{all occupations}}.
\end{tabular} &
18.00\% &
\multicolumn{2}{l|}{
\begin{tabular}{p{25mm}r}
\drawbar{67.56}{4.15}{\%}{Yellow}{\downarrowbold}\\
\drawbar{74.20}{2.69}{\%}{Lavender}{\uparrowbold}\\
\end{tabular}
}
&
\multicolumn{2}{l|}{
\begin{tabular}{p{25mm}r}
\drawbar{78.61}{8.02}{\%}{Yellow}{\downarrowbold}\\
\drawbar{91.58}{3.63}{\%}{Lavender}{\downarrowbold}\\
\end{tabular}
}
& 0.068 \\
\hline
% Row 5
\begin{tabular}{@{}p{87mm}}
For \colorbox{Orange!30}{\strut{individuals who were born in USA}}, the likelihood of having a health insurance decreases when they \colorbox{Cyan!30}{\strut{have no personal earnings}} for \colorbox{Yellow!70}{\strut{manual labor occupations}}, whereas it increases for \colorbox{Lavender!70}{\strut{all occupations}}.
\end{tabular} &
75.67\% &
\multicolumn{2}{l|}{
\begin{tabular}{p{25mm}r}
\drawbar{84.11}{3.46}{\%}{Yellow}{\downarrowbold}\\
\drawbar{88.97}{1.36}{\%}{Lavender}{\uparrowbold}\\
\end{tabular}
}
&
\multicolumn{2}{l|}{
\begin{tabular}{p{25mm}r}
\drawbar{78.61}{8.02}{\%}{Yellow}{\downarrowbold}\\
\drawbar{91.58}{3.63}{\%}{Lavender}{\downarrowbold}\\
\end{tabular}
}
 & 0.048 \\
\hline
\end{tabular}}
\label{tab:results_ACS}
\end{table*}

\subsection{Explanation Quality}
\label{subsec:exp_quality}
Table~\ref{tab:quality_metrics_new} shows a quantitative comparison contrasting
the explanations generated by \sysName with those of the baselines over three use
cases. Details of the explanations generated by each of these approaches are
in the Appendix.

% \af{Brit: It would be great if you can take care of Section 6.2. I am seeing
% some inconsistencies and I don't know why some of the results are the way they
% are. For example, why fairdebugger is taking longer time for SO but ACS.}

% {Why FairDebugger is failing to produce 5 results for SO and ACS. Why are we not
% observing enough diversity in terms of the subpopulation descriptions in ExDis
% as we promised to show. Why divexplorer is failing to produce 5 results on SO
% and failing to produce anything for MEPS.}

% \af{For the ACS dataset, we are seeing an interesting trend of opposite
% direction of treatment everywhere. Should we highlight this?} \brit{done}

% \af{Also, change the term ``general population'' to ``all occupations''
% (including in appendixes), to not mislead the reader to believe general
% population means all subpopulations, i.e., the entire data. What we mean by
% ``general population'' here is the entire subpopulation not just one group.}

% \af{The second thing I am struggling with is
% MEPS dataset experiment. My understanding was, globally, males feel LESS
% nervous than non-males. We want to find subpopulations where males feel MORE
% nervous and find causes that exacerbates this for males (making them even more
% nervous) and alleviates it for non-males (making them less nervous). But the
% results shown in Table 6 is not supporting the claims I wanted to make.}

% s

\subsubsection{Investigating a surprising fact (SO)} The explanations generated
by \sysName\ are shown in Table \ref{tab:results_so}. Notably, \sysName\
identifies subpopulations where the average salary of analysts is higher than
that of back-end developers, with the salary gap often exceeding that observed in the overall population (analysts: \$106,542; back-end developers: \$96,609). For almost all cases,
\sysName provides a different causal explanation to highlight the factor
contributing to the disparity within the subpopulations. In all five explanations, the causal effect of the chosen treatment on the entire
population is not statistically significant \newtext{as shown in the ``Global''
column}, emphasizing the importance of providing ``local'' explanations for each
subpopulation.
% Additionally, in some cases, a
% treatment pattern can explain why the average outcome of one group is higher
% than another within a subpopulation. However, this relationship may not hold at
% the global population level, as the treatment may benefit the second group more
% overall. This is evident in the 3rd disparity explanation. 
\emph{This
observation highlights the distinction between our approach and XInsight
(\newtext{that provides a global explanation, as shown in the Global column}),
demonstrating that subpopulation-level explanations differ from those provided
at the entire population level}.

% Another advantage that the in this use case we want to look on the places where the effect on the analysts is bigger than the back-end's effect, in the table we can see that the treatment = YearsCodeingProf = 24-26 years, bring the wanted requirement in the local level, but in the global level it brings the opposite -  so again looking for the treatments only in the global dimension can causes missing interesting information. Another conclusion that for all the given treatments the differences between groups' effects is bigger in the local than in the global level. Also the consideration in the support factor in the utility, bring big subpopulations in the result set.

In this scenario, \sysName\ shares only 1 out of 5 explanations with the optimal
solution by Brute Force. However, \sysName\ selects a highly diverse set of
explanations, achieving 55\% of the optimal disparity score produced by Brute
Force. In contrast, although the Top-k baseline achieves high disparity scores,
the selected explanations are highly similar to one another (as indicated by the
heatmap in Table \ref{tab:quality_metrics_new}), highlighting the importance of
incorporating similarity-awareness when selecting explanations.

% Comparison to the Brute-Force we got the same set here, which show that the greedy brings great result.
% Comparison to Top-K baseline, we got 4 of 5 objects in their resulted set, the missing object is because his small support, which caused to small utility, that caused our algorithm to pick another insight instead.

In this scenario, both DivExplorer and FairDebugger failed to identify any valid
solution. This is because they relied on a fixed treatment (as discussed in
Section~\ref{sec:six:one}). In every subpopulation they identified, either the
average income of backend developers exceeded that of analysts, or the treatment
disproportionately benefited backend developers. According to our problem
definition, such cases do not qualify as valid solutions.

\subsubsection{Fairness debugging (ACS)} \looseness-1 The disparity explanations
generated by \sysName\ for the second use case are presented in
Table~\ref{tab:results_ACS}. \sysName\ identified subpopulations in which the
percentage of manual-labor workers with health insurance was lower than in the
overall population. Notably, all explanations (i.e., treatments) involved
personal earnings, underscoring the strong causal relationship between income
level and the likelihood of having health insurance in the U.S. Interestingly,
in this case, a single explanation was sufficient to account for the disparity
across all identified subpopulations. However, this explanation behaves
differently when applied to the entire dataset. Specifically, for the overall
population, not having personal earnings decreases the likelihood of having
health insurance across all occupations. In contrast, for each of the reported
subpopulations, not having personal earnings actually increases the probability
of having health insurance when considering all individuals within those
subpopulations. We note that even in the optimal solution found by Brute Force
(as well as in the Top-K solution), only two distinct treatments were selected.
This suggests that, in this case, the space of relevant treatments is limited,
indicating few plausible explanations for the observed disparity.

\looseness-1 Compared to the Brute-Force baseline, \sysName\ recovered 2 out of
the 5 disparity explanations while exploring a significantly smaller search
space. Top-K overlapped with Brute-Force in just one explanation and exhibited
higher similarity among its selected explanations compared to \sysName. Both
DivExplorer and FairDebugger identified different subpopulations than \sysName;
FairDebugger detected only two (highly overlapping) subpopulations, while
DivExplorer found five, though with high similarity among them.

\subsubsection{Finding reverse trends (MEPS)} Table~\ref{tab:results_meps} shows
the disparity explanations for MEPS. \sysName identifies subpopulations where
the average likelihood of experiencing nervous attacks very frequently for males
($g_1$) is higher than that of non-males ($g_2$), in contrast to the opposite
trend observed in the overall population. \sysName\ generated a solution
identical to that of the Brute-Force approach, while Top-$k$ identified
explanations with substantial overlap among themselves. In this scenario, both
DivExplorer and FairDebugger failed to produce meaningful results. The
subpopulations they identified either did not align with the use case.
Specifically, they did not exhibit a trend that reversed the one observed in the
overall dataset, or, in the few relevant subpopulations they did find, the fixed
treatment yielded a disparity score of zero, rendering the explanation
ineffective for explaining the observed disparity.

% Again, the subpopulations identified by
% FairDebugger differed from those found by \sysName and it failed to generate the
% desired number of $5$ explanations \af{Why?}. The results computed by
% DivExplorer were not statistically significant.

\begin{tcolorbox}[colback=gray!10, colframe=gray, boxsep=1mm, left=1mm, right=1mm, top=1mm, bottom=1mm, title=Results Summary]
    $\bullet$ The top-$k$ baseline results in overlapping disparity explanations, demonstrating the need to consider the similarity among selected explanations.\\
    $\bullet$ The output of \sysName\ closely matches that of Brute Force, demonstrating that our solution prioritizes efficiency without compromising quality.\\
    $\bullet$ Explanations for disparity at the entire population level (as generated by XInsight) do not necessarily account for disparities within subpopulations, highlighting the need to find a specific local explanation for each subpopulation.\\ 
    $\bullet$ DivExplorer and FairDebugger were less successful in identifying subpopulations with high disparity between $g_1$ and $g_2$, resulting in low-score explanations.
\end{tcolorbox}

\smallskip\noindent\emph{\underline{Remark.}} Note that DivExplorer and
FairDebugger serve as baselines only for identifying subpopulations that exhibit
significant disparities between the two groups of interest. However, unlike
\sysName, they do not offer any causal explanations. To enable comparison, we
use a fixed treatment while using these two baselines, which is not necessarily
the optimal treatment that maximizes subpopulation-level disparity scores.
Therefore, these approaches result in a low total disparity score (which is
expected) when the subpopulation-level disparity scores are added to obtain the
score for the objective function (Problem~\ref{prblem}). Nonetheless, the
comparison allows for a qualitative contrast between the subpopulations
identified by the baselines and those discovered by \sysName.

%!TEX root=main.tex

\begin{table*}[t]
\centering
\caption{\small Disparity explanations discovered by \sysName for the MEPS dataset.}
\vspace{-3mm}
\resizebox{1\textwidth}{!}{
\renewcommand{\arraystretch}{1.1}
\begin{tabular}[t]{|p{85mm}|r|>{\centering}m{30mm}c|@{}>{\centering}m{30mm}c|c|}
\hline
\multirow{3}{*}{\textbf{Disparity Explanation}} & 
\multirow{3}{*}{\textbf{Support}} & 
\multicolumn{4}{c|}{\textbf{Likelihood of feeling nervous frequently}} &
\multirow{3}{*}{\textbf{$\Delta$}} \\
\cline{3-6}
& &
\multicolumn{2}{c|}{\textbf{Subpopulation (\sysName)}} &
\multicolumn{2}{c|}{\textbf{Global (XInsight)}}& \\
\cline{3-6}
&& 
\textbf{Average} & \textbf{CATE} & 
\textbf{Average} & \textbf{CATE} &\\
\hline
\hline
% Row 1
\begin{tabular}{@{}p{85mm}}
For individuals who \colorbox{Orange!30}{\strut{were never married, are from the Southern}} \colorbox{Orange!30}{\strut{
region, don't have a doctor's recommendation to exercise, and}} \colorbox{Orange!30}{\strut{ aren't diagnosed
with Diabetes}}, the likelihood of feeling nervous frequently decreases less for
\colorbox{Yellow!70}{\strut{males}} who \colorbox{Cyan!30}{\strut{do not currently smoke}}
compared to \colorbox{Lavender!70}{\strut{non-males}}.
\end{tabular} &
5.78\% &
\multicolumn{2}{l|}{
\begin{tabular}{p{30mm}r}
\drawbar{46.08}{13.06}{\%}{Yellow}{\downarrowbold}\\
\drawbar{41.71}{20.73}{\%}{Lavender}{\downarrowbold}\\
\end{tabular}
}
&
\multicolumn{2}{l|}{
\begin{tabular}{p{30mm}r}
\drawbar{37.58}{3.39}{\%}{Yellow}{\downarrowbold}\\
\drawbar{45.10}{3.94}{\%}{Lavender}{\downarrowbold}\\
\end{tabular}
} & 0.076 \\
\hline
% Row 2
\begin{tabular}{@{}p{85mm}}
For individuals who \colorbox{Orange!30}{\strut{are White, were never married, are under
29,}} \colorbox{Orange!30}{\strut{ and don't have Asthma or Diabetes}}, the likelihood of feeling nervous
frequently decreases less for \colorbox{Yellow!70}{\strut{males}} than
\colorbox{Lavender!70}{\strut{non-males}} when they are \colorbox{Cyan!30}{\strut{uninsured for the
whole year}}.
\end{tabular} &
8.41\% &
\multicolumn{2}{l|}{
\begin{tabular}{p{30mm}r}
\drawbar{49.20}{14.16}{\%}{Yellow}{\downarrowbold}\\
\drawbar{48.95}{18.48}{\%}{Lavender}{\downarrowbold}\\
\end{tabular}
}
&
\multicolumn{2}{l|}{
\begin{tabular}{p{30mm}r}
\drawbar{37.58}{7.02}{\%}{Yellow}{\downarrowbold}\\
\drawbar{45.10}{4.86}{\%}{Lavender}{\downarrowbold}\\
\end{tabular}
} & 0.043 \\
\hline
% Row 3
\begin{tabular}{@{}p{85mm}}
For individuals who \colorbox{Orange!30}{\strut{were never married, don't have a
doctor's}} \colorbox{Orange!30}{\strut{recommendation to exercise, and were born in USA}}, the
likelihood of feeling nervous frequently increases more for
\colorbox{Yellow!70}{\strut{males}} who \colorbox{Cyan!30}{\strut{have private insurance}}
compared to \colorbox{Lavender!70}{\strut{non-males}}.
\end{tabular}
&
14.25\%&
\multicolumn{2}{l|}{
\begin{tabular}{p{30mm}r}
\drawbar{45.90}{10.84}{\%}{Yellow}{\uparrowbold}\\
\drawbar{45.74}{7.51}{\%}{Lavender}{\uparrowbold}\\
\end{tabular}
}
&
\multicolumn{2}{l|}{
\begin{tabular}{p{30mm}r}
\drawbar{37.58}{4.63}{\%}{Yellow}{\uparrowbold}\\
\drawbar{45.10}{5.76}{\%}{Lavender}{\uparrowbold}\\
\end{tabular}
}
& 0.033\\
\hline
% Row 4
\begin{tabular}{@{}p{85mm}}
For individuals who are \colorbox{Orange!30}{\strut{White, were never married, don't
have}} \colorbox{Orange!30}{\strut{ a doctor's recommendation to exercise, and don't have Asthma}}, the
likelihood of feeling nervous frequently decreases less for
\colorbox{Yellow!70}{\strut{males}} who \colorbox{Cyan!30}{\strut{were uninsured the whole year}}
compared to \colorbox{Lavender!70}{\strut{non-males}}.
\end{tabular} &
9.74\% &
\multicolumn{2}{l|}{
\begin{tabular}{p{30mm}r}
\drawbar{47.25}{12.52}{\%}{Yellow}{\downarrowbold}\\
\drawbar{46.64}{14.09}{\%}{Lavender}{\downarrowbold}\\
\end{tabular}
}
&
\multicolumn{2}{l|}{
\begin{tabular}{p{30mm}r}
\drawbar{37.58}{7.02}{\%}{Yellow}{\downarrowbold}\\
\drawbar{45.10}{4.86}{\%}{Lavender}{\downarrowbold}\\
\end{tabular}
} & 0.015 \\
\hline
% Row 5
\begin{tabular}{@{}p{85mm}}
For individuals \colorbox{Orange!30}{\strut{aged between 30-42 years who don't have
Diabetes}}, the likelihood of feeling nervous frequently increases more for
\colorbox{Yellow!70}{\strut{males}} who \colorbox{Cyan!30}{\strut{have health insurance}}
compared to \colorbox{Lavender!70}{\strut{non-males}}.
\end{tabular} &
19.64\% &
\multicolumn{2}{l|}{
\begin{tabular}{p{30mm}r}
\drawbar{43.93}{\phantom{7}7.47}{\%}{Yellow}{\uparrowbold}\\
\drawbar{43.57}{\phantom{7}6.47}{\%}{Lavender}{\uparrowbold}\\
\end{tabular}
}
&
\multicolumn{2}{r|}{
\begin{tabular}{p{30mm}r}
\drawbar{37.58}{2.34}{\%}{Yellow}{\uparrowbold}\\
\drawbar{45.10}{4.54}{\%}{Lavender}{\uparrowbold}\\
\end{tabular}
} & 0.010 \\
\hline
\end{tabular}}
\label{tab:results_meps}
\end{table*}

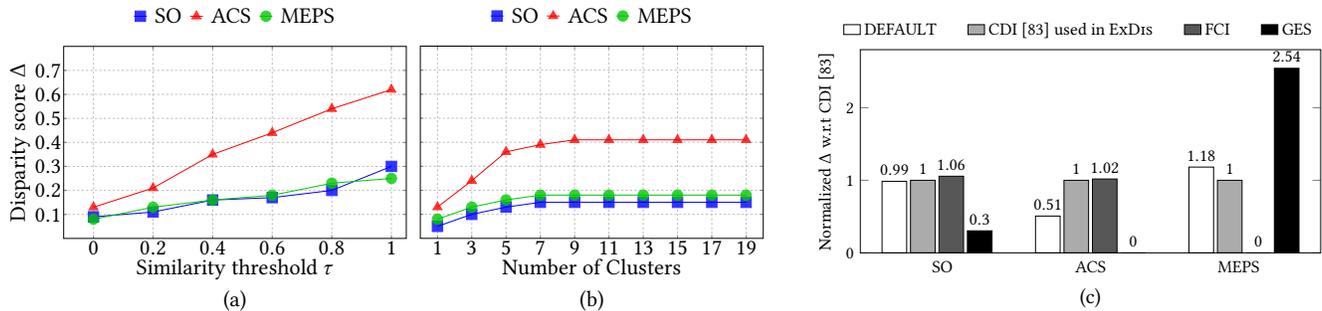
\begin{figure*}[t]
    \centering
	%!TEX root=main.tex

\resizebox{0.3\textwidth}{!}{
\begin{tikzpicture}
    \begin{axis}[
        width=\textwidth,
        height=0.6\textwidth,
        xlabel={Similarity threshold $\tau$},
        ylabel={Disparity score $\Delta$},
        xmin=-0.1, xmax=1.05,
        ymin=0, ymax=0.8,
        xtick={0.0,0.2,0.4,0.6,0.8,1.0},
        ytick={0.1,0.2,0.3,0.4,0.5,0.6,0.7},
        ymajorgrids=true,
        grid=major,
        grid style=dashed,
        tick label style={font=\fontsize{30}{36}\selectfont},
        label style={font=\fontsize{30}{36}\selectfont},
        legend style={
            at={(0.5,1.1)},
            anchor=south,
            font=\fontsize{30}{36}\selectfont,
            legend columns=3,
            column sep=1em,
            draw=none
        },
		after end axis/.code={
		            \node[below=7.5em of current axis.south, font=\fontsize{30}{36}\selectfont, align=center] 
		            {(a)};
		},
    ]

    % SO - blue
    \addplot[
        color=blue,
        mark=square*,
        thick,
        mark size=7,
        mark options={opacity=0.8},
    ] coordinates {
        (0.0,0.09) (0.2,0.11) (0.4,0.16) (0.6,0.17) (0.8,0.20) (1.0,0.30)
    };
    \addlegendentry{SO}

    % ACS - red
    \addplot[
        color=red,
        mark=triangle*,
        thick,
        mark size=7,
        mark options={opacity=0.8},
    ] coordinates {
        (0.0,0.13) (0.2,0.21) (0.4,0.35) (0.6,0.44) (0.8,0.54) (1.0,0.62)
    };
    \addlegendentry{ACS}

    % MEPS - dark green
    \addplot[
        color=darkgreen,
        mark=*,
        thick,
        mark size=7,
        mark options={opacity=0.8},
    ] coordinates {
        (0.0,0.08) (0.2,0.13) (0.4,0.16) (0.6,0.18) (0.8,0.23) (1.0,0.25)
    };
    \addlegendentry{MEPS}
    \end{axis}	
\end{tikzpicture}
}
\raisebox{0mm}{\resizebox{0.261\textwidth}{!}{
\begin{tikzpicture}
    \begin{axis}[
        width=\textwidth,
        height=0.6\textwidth,
        xlabel={Number of Clusters},
        xmin=0, xmax=20,
        ymin=0, ymax=0.7,
        xtick={1,3,5,7,9,11,13,15,17,19},
        ymin=0, ymax=0.8,
        ytick={0.1,0.2,0.3,0.4,0.5,0.6,0.7},
		yticklabels={},
        ymajorgrids=true,
        grid=major,
        grid style=dashed,
        tick label style={font=\fontsize{30}{36}\selectfont},
        label style={font=\fontsize{30}{36}\selectfont},
        legend style={
            at={(0.5,1.1)},
            anchor=south,
            font=\fontsize{30}{36}\selectfont,
            legend columns=3,
            column sep=1em,
            draw=none
        },
		after end axis/.code={
		            \node[below=7.5em of current axis.south, font=\fontsize{30}{36}\selectfont, align=center] 
		            {(b)};
		},
    ]

    % SO - blue
    \addplot[
        color=blue,
        mark=square*,
        thick,
        mark size=7,
        mark options={opacity=0.8},
    ] coordinates {
        (1,0.05) (3,0.10) (5,0.13) (7,0.15) (9,0.15) (11,0.15)
        (13,0.15) (15,0.15) (17,0.15) (19,0.15)
    };
    \addlegendentry{SO}

    % ACS - red
    \addplot[
        color=red,
        mark=triangle*,
        thick,
        mark size=7,
        mark options={opacity=0.8},
    ] coordinates {
        (1,0.13) (3,0.24) (5,0.36) (7,0.39) (9,0.41) (11,0.41)
        (13,0.41) (15,0.41) (17,0.41) (19,0.41)
    };
    \addlegendentry{ACS}

    % MEPS - dark green
    \addplot[
        color=darkgreen,
        mark=*,
        thick,
        mark size=7,
        mark options={opacity=0.8},
    ] coordinates {
        (1,0.08) (3,0.13) (5,0.16) (7,0.18) (9,0.18) (11,0.18)
        (13,0.18) (15,0.18) (17,0.18) (19,0.18)
    };
    \addlegendentry{MEPS}

    \end{axis}
\end{tikzpicture}}}
\hspace{5mm}
\resizebox{0.38\textwidth}{!}{%
\begin{tikzpicture}
\begin{axis}[
    ybar,
	axis on top,
    bar width=0.45cm,
    width=10cm,
    height=5.3cm,
    ymin=0,
    ymax=2.8,
    enlarge x limits=false, % control manually via xmin/xmax
    xmin=0.5,               % add space before first group
    xmax=3.5,               % optional: space after last group
    ylabel={Normalized $\Delta$ w.r.t CDI~\cite{youngmann2023causal}},
    xtick={1,2,3},
    xtick style={draw=none},
    xticklabels={SO, ACS, MEPS},
    x tick label style={rotate=0, anchor=east, yshift=-3pt, xshift=4mm},
    legend style={
        at={(-0.05,1.02)},       % left-aligned (x=0)
        anchor=south west,  % anchor bottom-left corner of legend box
        legend columns=4,
        draw=none,
        /tikz/every even column/.append style={column sep=0.5cm} 
    },
	legend image code/.code={
	  	\draw[draw=black,fill opacity=1] (0cm,-0.08cm) rectangle (0.3cm,0.1cm);
	},
	every node near coord/.append style={black},
    nodes near coords,
    nodes near coords align={vertical},
	after end axis/.code={
	            \node[below=1.7em of current axis.south, font=\fontsize{11}{19}\selectfont, align=center] 
	            {(c)};
	},
]

\addplot+[
    ybar,
    fill=black!00,
    draw=black,    
] coordinates { (1, 0.986) (2, 0.507) (3, 1.182)};
\addlegendentry{DEFAULT}

\addplot+[
    ybar,
    fill=black!33,
    draw=black,    
] coordinates { (1, 1.000) (2, 1.000) (3, 1.000)};
\addlegendentry{CDI~\cite{youngmann2023causal} used in \sysName~}

\addplot+[
    ybar,
    fill=black!66,
    draw=black,    
] coordinates { (1, 1.056) (2, 1.018) (3, 0.000)};
\addlegendentry{FCI}

\addplot+[
    ybar,
    fill=black!100,
    draw=black,    
] coordinates { (1, 0.304) (2, 0.000) (3, 2.543)};
\addlegendentry{GES}

\end{axis}
\end{tikzpicture}}	
	\vspace{-4mm}
    \caption{\small Effect of various system parameters on the disparity score.
    (a) \& (b)~The absolute disparity scores are reported here to show direct
    impact of the similarity threshold $\tau$ and the number of clusters on the
    disparity score $\Delta$. (c)~Effect of casual DAG modification. Disparity
    scores here are shown as a relative value w.r.t the disparity score of
    CDI~\cite{youngmann2023causal} (which \sysName uses).}
	\vspace{-1mm}
    \label{fig:paramquality}
\end{figure*}

\subsection{Parameters Sensitivity}
\label{subsec:parameters}

% \af{Optional: add more experiments: one varying k, one varying min-sup threshold}

%
In this section, we investigate the impact of various parameters on our
objective function, namely disparity score $\Delta$. Our goal is to gain
insights into effective default parameter settings.

\subsubsection{Robustness to similarity threshold $\tau$}
Figure~\ref{fig:paramquality}~(a) shows how our objective function, i.e., the
disparity score $\Delta$ changes with varying similarity threshold $\tau$ across
three datasets: SO, ACS, and MEPS. As $\tau$ increases, all datasets exhibit a
rising trend in disparity. This is expected because low similarity threshold
significantly restricts the feasible solution space. For a fixed budget $k$ (we
used $k = 11$ for this experiment, to observe the impact of $\tau$ without any
other restriction), relaxing \(\tau\) allows for the inclusion of a larger
number of disparity explanations, which can cumulatively increase the overall
disparity score. This is because more similar subpopulations are permitted to
coexist, potentially capturing less diverse and more overlapping causes of
disparity. However, setting \(\tau\) too high undermines the goal of maintaining
diversity among the explanations, as excessive overlap can dilute the quality
and reduce the distinctiveness of the explanation set. Thus, \(\tau\) must be
carefully chosen to balance comprehensiveness and diversity.

\subsubsection{Robustness to number of clusters}
Figure~\ref{fig:paramquality}~(b) shows how the disparity score $\Delta$ varies
with the number of clusters across three datasets: SO, ACS, and MEPS. As the
number of clusters increases, all datasets exhibit an initial rise in disparity,
which eventually plateaus. Notably, the ACS dataset shows the most significant
increase, with the disparity growing rapidly up to 5 clusters before stabilizing
around 0.4. In contrast, SO and MEPS show more modest increases, leveling off at
lower disparity scores. These results suggest that finer-grained clustering help
improve the disparity scores, but has diminishing return.

\subsubsection{Robustness to the Causal DAG}
\label{subsec:causal_DAG_robustness} The quality of the solution may depend on
the quality of the underlying causal DAG. To evaluate this, we assess the impact
of using different causal DAGs, generated by commonly used causal discovery methods. We consider the following DAGs: (1)~\textbf{DEFAULT}, a default two-layer causal DAG where immutable attributes affect both mutable attributes and the outcome, and mutable attributes also affect the outcome
(2)~Causal Data Integration (\textbf{CDI})~\cite{youngmann2023causal}, which
\sysName uses, (3)~\textbf{FCI}~\cite{spirtes2000causation}, and
(4)~\textbf{GES}~\cite{chickering2002optimal}.
The results are depicted in Figure \ref{fig:paramquality}~(c). We report the
relative disparity scores ($\Delta$) computed using the different DAGs, w.r.t. the disparity score of CDI. Observe that for the default DAG, the results closely resemble those obtained using the CDI DAG employed by \sysName. In contrast, the FCI and GES DAGs produce more varied results across different use cases. This variability is due to the challenges that causal discovery algorithms face when applied to real-world data~\cite{glymour2019review,o2006incorporating}, often leading to noisier and less reliable DAGs.
Nonetheless, the disparity explanations generated using FCI and GES DAGs were largely consistent with those derived from CDI, with similar explanations selected. This suggests that \sysName\ remains robust and capable of producing meaningful explanations even when the underlying causal DAG is noisy or imperfect.

% In the SO and MEPS datasets, the
% results closely matched those obtained with the causal DAG used by \sysName (CDI). However, for the ACS dataset,
% we observed a larger deviation. This is likely due to the challenges causal
% discovery algorithms face when handling high-dimensional data like
% ACS~\cite{glymour2019review,o2006incorporating}, resulting in noisier DAGs.
% Nevertheless, the disparity explanations generated using each of the discovered
% DAGs were consistent with those from the original \af{what do you mean by
% ``original''? What is DEFAULT then?} DAG, with similar patterns selected for
% both treatment and subpopulations. 

\begin{tcolorbox}[colback=gray!10, colframe=gray, boxsep=1mm, left=1mm, right=1mm, top=1mm, bottom=1mm, title=Results Summary]
    $\bullet$ A higher \(\tau\) can increase the disparity score by allowing
    overlapping explanations at the cost of compromising diversity. \\
    $\bullet$ Finer-grained clustering helps improve the disparity scores, but
    has diminishing return.\\
    $\bullet$ Even with imperfect causal DAGs, \sysName is able to produce meaningful and
robust disparity explanations.
\end{tcolorbox}

\subsubsection{\revb{Robustness to ATE estimator}} \label{sec:robustate}
 \revb{By default, we use DoWhy~\cite{dowhypaper} to estimate ATE values via a linear regression model. Next, we evaluate the robustness of the generated explanations under alternative ATE estimators. Specifically, we examine three widely used methods: propensity score stratification, propensity score weighting, and propensity score matching,  all implemented in DoWhy~\cite{dowhypaper}, following standard causal inference practices~\cite{rosenbaum1983central,austin2011introduction}.}
 \revb{The results are shown in Table \ref{tab:exdis_estimators}. 
Our experiments reveal three main findings:}
\revb{\begin{enumerate}[leftmargin=1em,labelwidth=*,align=left]
    \item As expected, the linear regression estimator is the fastest and is therefore used by default in \sysName.
    \item Explanations generated using different estimators are largely consistent, typically selecting the same treatments with only minor differences in the estimated CATE values, demonstrating that \sysName\ is robust to the choice of ATE estimator. 
    \item The lack of results for MEPS when using the prosperity-score stratification estimator stems from the high variance of stratified ATE estimates which causes the CATE values to be not statistically significant. This is caused by sample fragmentation and limited overlap across prosperity-score bins. This behavior is consistent with standard findings in causal inference: while stratification enhances interpretability, it can reduce statistical power when strata are small or poorly balanced~\cite{miratrix2013adjusting}.
\end{enumerate}}

\subsubsection{\revb{Robustness to seed selection}}\label{subsec:seed}
\revb{As discussed in Section \ref{subsec:step_3}, we select a random representative from each cluster and assign its disparity score to the entire cluster, which introduces some randomization. By default, we used a fixed seed. Next, we examine how different seeds affect the results. 
In particular, we run each experiment on the SO and ACS datasets using seven different random seeds, and report the average and variance of the objective value, diversity, and runtime. The results are shown in Table \ref{tab:exdis_randomness}. Our results show that the results were largely consistent across different seeds. Our results suggest that \sysName\ is highly robust to randomization: across all seeds, the variance in both the objective value and diversity is very small (e.g., $\leq$ 0.09 for SO and $\leq$ 0.04 for MEPS), indicating that the clustering-based optimization introduces minimal noise. Similarly, runtime remains stable, with negligible variation across runs (< 0.1 s on average). These findings confirm that the random seed used for representative selection has a minor impact on both the quality and efficiency of the results.}

\begin{table}[t]
\caption{\small \revb{Results using different ATE estimators across SO and MEPS. LR denotes Linear Regression and PS denotes Propensity Score.}}
\label{tab:exdis_estimators}
\vspace{-4mm}
\centering
\resizebox{0.47\textwidth}{!}{
    \begin{tabular}{l
				r@{\hskip 3mm}r@{\hskip 3mm}c
                r@{\hskip 3mm}r@{\hskip 3mm}c
		}
        \toprule
    	\multirow{2}{*}{\revb{\diagbox{\textbf{Estimator}}{\textbf{Dataset}}}}  & 
    	\multicolumn{3}{c}{\cellcolor{blue!10}\textbf{SO}} &
    	\multicolumn{3}{c}{\cellcolor{green!5}\textbf{MEPS}} \\
        
        & 
	\cellcolor{blue!10}{\textbf{$\Delta$ (\%)}} &
        \cellcolor{blue!10}{\textbf{Runtime (s)}} &
        \cellcolor{blue!10}{\textbf{\#Exp}} &
        
        \cellcolor{green!5}{\textbf{$\Delta$ (\%)}} &
        \cellcolor{green!5}{\textbf{Runtime (s)}} &
        \cellcolor{green!5}{\textbf{\#Exp}} \\
        
        \midrule

        \textbf{\revb{\sysName\ + LR}}
        & \cellcolor{blue!10}{55} & \cellcolor{blue!10}{62} & \cellcolor{blue!10}{5}
        & \cellcolor{green!5}{100} & \cellcolor{green!5}{17} & \cellcolor{green!5}{5} \\[2mm]

        \textbf{\revb{\sysName + PS Stratification}}        
        & \cellcolor{blue!10}{100} & \cellcolor{blue!10}{4870} & \cellcolor{blue!10}{5}
        & \cellcolor{green!5}{0} & \cellcolor{green!5}{331} & \cellcolor{green!5}{0} \\[2mm]

        \textbf{\revb{\sysName + PS Weighting}}        
        & \cellcolor{blue!10}{94} & \cellcolor{blue!10}{4655} & \cellcolor{blue!10}{5}
        & \cellcolor{green!5}{100} & \cellcolor{green!5}{1070} & \cellcolor{green!5}{5} \\[2mm]

        \textbf{\revb{\sysName + PS Matching}}        
        & \cellcolor{blue!10}{92} & \cellcolor{blue!10}{67379} & \cellcolor{blue!10}{5}
        & \cellcolor{green!5}{100} & \cellcolor{green!5}{7265} & \cellcolor{green!5}{2} \\[2mm]
		        
        \bottomrule
    \end{tabular}
}
\vspace{4mm}
\end{table}

\subsection{Efficiency \& Scalability} 
\label{subsec:scalability}
In this section, we present results demonstrating the efficiency of various
components of \sysName, analyze the impact of different parameters on its
runtime, and evaluate its scalability as the dataset grows both vertically (in
tuples) and horizontally (in attributes).

\begin{table}[t]
\caption{\revb{Average and variance (over 7 random seeds) of objective value ($\Delta$), diversity, and runtime for \sysName.}}
\label{tab:exdis_randomness}
\vspace{-4mm}
\centering
\resizebox{0.49\textwidth}{!}{
\begin{tabular}{
    l
    r@{\hskip 3mm}r@{\hskip 3mm}r
}
\toprule
\textbf{Dataset} &
\textbf{$\Delta$ (\%) (avg $\pm$ var)} &
\textbf{Diversity (avg $\pm$ var)} &
\textbf{Runtime (avg $\pm$ var)} \\
\midrule

\cellcolor{blue!10}\textbf{SO} 
& \cellcolor{blue!10}{51.79 (7.86)} 
& \cellcolor{blue!10}{0.2996 (0.0136)} 
& \cellcolor{blue!10}{60.283 (0.038)} \\[1mm]

\cellcolor{green!5}\textbf{MEPS} 
& \cellcolor{green!5}{92.63 (3.25)} 
& \cellcolor{green!5}{0.2437 (0.0091)} 
& \cellcolor{green!5}{10.25 (0.065)} \\

\bottomrule
\end{tabular}
}
\vspace{4mm}
\end{table}

%!TEX root=main.tex

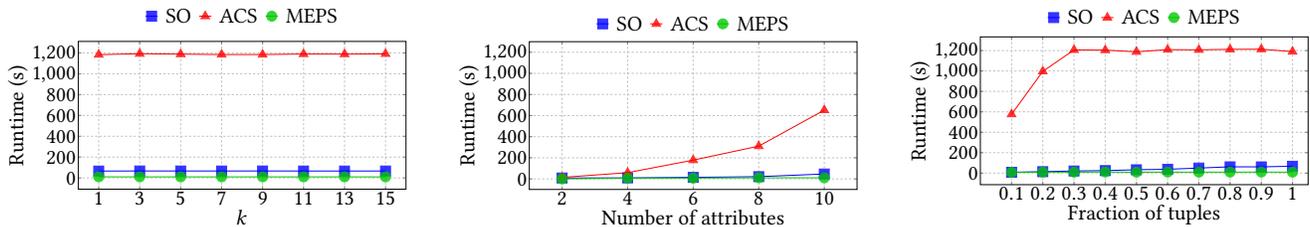
\begin{figure*}[t]
    \centering    
    \raisebox{\depth}{\resizebox{0.3\textwidth}{!}{%
		\begin{tikzpicture}
	    \begin{axis}[
	        width=\textwidth,
	        height=0.5\textwidth,
	        xlabel={$k$},
	        ylabel={Runtime (s)},
	        xmin=0, xmax=16,
	        grid=major,
	        xtick={1,3,5,7,9,11,13,15},
	        ymajorgrids=true,
	        grid style=dashed,
	        tick label style={font=\fontsize{30}{36}\selectfont},
	        label style={font=\fontsize{30}{36}\selectfont},
	        legend style={
	            at={(0.5,1.1)},
	            anchor=south,
	            font=\fontsize{30}{36}\selectfont,
	            legend columns=3,
	            column sep=1em,
	            draw=none
	        },
	    ]

	    % SO - blue
	    \addplot[
	        color=blue,
	        mark=square*,
	        thick,
	        mark size=7,
	        mark options={opacity=0.8},
	    ] coordinates {
	        (1,66.01) (3,66.83) (5,66.6) (7,66.56) (9,66.79) (11,66.44) (13,66.35) (15,66.61)
	    };
	    \addlegendentry{SO}

	    % ACS - red
	    \addplot[
	        color=red,
	        mark=triangle*,
	        thick,
	        mark size=7,
	        mark options={opacity=0.8},
	    ] coordinates {
	        (1,1183.25) (3,1192.12) (5,1188.15) (7,1184.32) (9,1184.28) (11,1189.35) (13,1188.28) (15,1190.82)
	    };
	    \addlegendentry{ACS}

	    % MEPS - dark green
	    \addplot[
	        color=darkgreen,
	        mark=*,
	        thick,
	        mark size=7,
	        mark options={opacity=0.8},
	    ] coordinates {
	        (1,10.68) (3,10.23) (5,10.19) (7,10.22) (9,10.12) (11,10.12) (13,10.15) (15,10.02)
	    };
	    \addlegendentry{MEPS}

	    \end{axis}
	\end{tikzpicture}
	}}
	\hspace{5mm}
	\raisebox{\depth}{\resizebox{0.3\textwidth}{!}{%
		\begin{tikzpicture}
	    \begin{axis}[
	        width=\textwidth,
	        height=0.5\textwidth,
	        xlabel={Number of attributes},
	        ylabel={Runtime (s)},
	        xmin=1, xmax=11,
			ymin=-100, ymax=1300,
	        grid=major,
	        xtick={2,4,6,8,10},
	        ymajorgrids=true,
	        grid style=dashed,
	        tick label style={font=\fontsize{30}{36}\selectfont},
	        label style={font=\fontsize{30}{36}\selectfont},
	        legend style={
	            at={(0.5,1.05)},
	            anchor=south,
	            font=\fontsize{30}{36}\selectfont,
	            legend columns=3,
	            column sep=1em,
	            draw=none
	        },
	    ]

	    % SO - blue
	    \addplot[
	        color=blue,
	        mark=square*,
	        thick,
	        mark size=7,
	        mark options={opacity=0.8},
	    ] coordinates {
	        (2,8.64) (4,9.43) (6,14.69) (8,21.2) (10,46.43)
	    };
	    \addlegendentry{SO}

	    % ACS - red
	    \addplot[
	        color=red,
	        mark=triangle*,
	        thick,
	        mark size=7,
	        mark options={opacity=0.8},
	    ] coordinates {
	        (2,14.2) (4,59.16) (6,177.73) (8,310.46) (10,651.44)
	    };
	    \addlegendentry{ACS}

	    % MEPS - dark green
	    \addplot[
	        color=darkgreen,
	        mark=*,
	        thick,
	        mark size=7,
	        mark options={opacity=0.8},
	    ] coordinates {
	        (2,0.11) (4,7.21) (6,7.37) (8,10.53) (10,10.17)
	    };
	    \addlegendentry{MEPS}

	    \end{axis}
	\end{tikzpicture}
	}}
	\hspace{5mm}
	\raisebox{0mm}{\resizebox{0.3\textwidth}{!}{%
		\begin{tikzpicture}
        \begin{axis}[
            width=\textwidth,
            height=0.5\textwidth,
            xlabel={Fraction of tuples},
            ylabel={Runtime (s)},
            xmin=0, xmax=1.05,
            grid=major,
            xtick={0.1,0.2,...,1.0},
            ymajorgrids=true,
            grid style=dashed,
            tick label style={font=\fontsize{30}{36}\selectfont},
            label style={font=\fontsize{30}{36}\selectfont},
            legend style={
                at={(0.5,1.06)},
                anchor=south,
                font=\fontsize{30}{36}\selectfont,
                legend columns=3,
                column sep=1em,
				draw=none
            },
        ]

        % SO - blue
        \addplot[
            color=blue,
            mark=square*,
            thick,
            mark size=7,
			mark options={opacity=0.8},
        ] coordinates {
            (0.1,8.47) (0.2,14.26) (0.3,21.1) (0.4,24.87) (0.5,32.92)
            (0.6,37.77) (0.7,49.72) (0.8,61.74) (0.9,61.44) (1.0,67.43)
        };
        \addlegendentry{SO}

        % ACS - dark gray
        \addplot[
            color=red,
            mark=triangle*,
            thick,
            mark size=7,
			mark options={opacity=0.8},
        ] coordinates {
            (0.1,576.39) (0.2,996.2) (0.3,1205.9) (0.4,1203.92) (0.5,1187.62)
            (0.6,1208.79) (0.7,1206.81) (0.8,1211.99) (0.9,1213.03) (1.0,1189.48)
        };
        \addlegendentry{ACS}

        % MEPS - dark green
        \addplot[
            color=darkgreen,
            mark=*,
            thick,
            mark size=7,
			mark options={opacity=0.8},
        ] coordinates {
            (0.1,7.6) (0.2,9.43) (0.3,9.36) (0.4,10.86) (0.5,9.19)
            (0.6,10.06) (0.7,9.77) (0.8,10.19) (0.9,10.39) (1.0,9.66)
        };
        \addlegendentry{MEPS}

        \end{axis}
    \end{tikzpicture}	
	}}
\vspace{-3mm}
    \caption{Effects of various parameters on runtime: (left)~the budget
    parameter $k$, (center)~number of attributes, and (right)~fraction of
    data.}	
\vspace{-2mm}
    \label{fig:parameters_runtime}
\end{figure*}

\smallskip\noindent\textbf{Step-by-step breakdown of runtime}: We present a
step-wise breakdown of runtime in Table~\ref{fig:runtime_per_step}. Not
surprisingly, The step ``explanation miner'', which focuses on identifying the
causal explanation for each subpopulation, is the most computationally expensive
one, accounting for over 80\% of the total runtime in all examined scenarios.
Nevertheless, \sysName generates the solution within a reasonable time, even for
large, high-dimensional datasets like ACS.

\smallskip\noindent\textbf{Effect of various parameters on runtime}: Next, we
analyze how various parameters impact runtime. Since parameter variations
involve sampling, we repeat each experiment $5$ times and report the average
runtime across all runs.

\smallskip\noindent\emph{\underline{Solution size $k$.}} Figure~\ref{fig:parameters_runtime}
(left) shows the impact of the solution size $k$ on the runtime. Note that $k$
only affects the last step (fast greedy search), which selects the explanations
from the candidates mined in the previous steps. Recall that this step evaluates
the pairwise intersection between the subpopulations corresponding to the
disparity explanations to account for the diversity constraint. As $k$
increases, the number of pairwise comparisons grows and so does the runtime. The
effect of $k$ on runtime is negligible for the smaller datasets (MEPS and SO)
compared to the larger dataset ACS, where computing intersections among
subpopulations takes longer due to the dataset's size.

\smallskip\noindent\emph{\underline{Number of attributes.}}
\looseness-1 Figure~\ref{fig:parameters_runtime} (center) shows the impact of
the number of attributes on runtime. In this experiment, we randomly sampled
subsets of attributes to retain and removed the rest from the dataset. The
number of attributes influences the search-space size, as more attributes result
in a larger set of subpopulations and treatment patterns to consider.
Theoretically, runtime should grow exponentially with the number of attributes.
However, this worst-case behavior was not always observed in practice, as
several factors influence computation—such as the structure of the underlying
causal DAG, the choice of mutable and immutable attributes, and other system
parameters. Nevertheless, we observe that, as expected, the number of attributes
in the dataset significantly affects runtime. For the largest dataset, ACS, we
do in fact observe exponential growth in runtime as the number of attributes
increases.

\smallskip\noindent\emph{\underline{Dataset cardinality.}}
Figure~\ref{fig:parameters_runtime} (right) illustrates the impact of the
dataset cardinality (in number of tuples) on runtime. In this experiment, we
varied the dataset size using specific-sized horizontal dataset slices. We find
that the runtime is linearly influenced by the number of tuples. For the ACS
dataset, we applied the sampling optimization (during the Explanation Miner
phase, as explained in \cref{subsec:step_2}). We found that the growth in
runtime is more moderate up to around 30\% of the ACS
data (which is 500{,}000 tuples). Beyond this point, the Explanation Miner module operates on a random
fixed-sized sample of 500{,}000 tuples.

\begin{tcolorbox}[colback=gray!10, colframe=gray,
	boxsep=1mm,       % Reduce space between text and box border
	  left=1mm,         % Reduce left padding
	  right=1mm,        % Reduce right padding
	  top=1mm,          % Reduce top padding
	  bottom=1mm,        % Reduce bottom padding
	  title=Results Summary]

$\bullet$ The explanation miner step of \sysName, which focuses on identifying
the causal explanation for each subpopulation, takes the longest time,
accounting for over 80\% of the total runtime. \\
$\bullet$ The runtime grows (almost linearly) with the solution size $k$ because
of the pair-wise similarity computation. \\
$\bullet$ The runtime is greatly influenced by the number of attributes in the
dataset, as it affects the size of the search space.\\
$\bullet$ The runtime grows linearly with the number of dataset tuples.
\end{tcolorbox}

\begin{table}
    \centering
    \small
         \caption{\small Breakdown of runtime by steps (seconds).} \label{fig:runtime_per_step}
         \vspace{-3mm}
        \resizebox{\columnwidth}{!}{
        \begin{tabular}{lrrr}
            \toprule
            \textbf{Dataset} & \textbf{Subpopulation Miner} &
            \textbf{Explanation Miner} &
            \textbf{Fast Greedy Search}\\
            \midrule
            SO  & 0.4  & 70.1 & 0.5 \\
            ACS  & 3.0 & 1294.6 & 13.0 \\
            MEPS  & 0.3 & 11.3 & 0.1 \\
            \bottomrule
        \end{tabular}}
        \label{fig:runtime_per_step}
        \vspace{4mm}
\end{table}

\subsection{Ablation Study}
\label{subsec:ablation}
%\af{Optional: update with better ablation study if possible.}

To assess the impact of our proposed optimizations on the runtime of \sysName,
we compare 5 variants: 
(1)~None, implying no optimization was applied,
(2)~No Parallel, denoting the setting where parallelization was removed,
(3)~No Cache, denoting the setting where caching was removed,
(4)~No Clustering, denoting the setting where clustering was removed,
(5)~All, denoting the \sysName setting where all optimizations were applied.

Figure~\ref{fig:optimization_tests} illustrates how the removal of individual optimization techniques affects runtime performance across three datasets: SO, ACS, and MEPS. The y-axis is plotted on a logarithmic scale to clearly illustrate differences in runtimes. As expected, without any optimization, we observe the highest runtimes across all three datasets. We also observe that \sysName, with all optimizations, yield the best runtime performance. Removing clustering increases runtimes moderately, indicating clustering provides a noticeable but relatively modest speedup. Eliminating caching also causes a moderate performance degradation. While caching was introduced to avoid
redundant computations, it offers relatively smaller gains. This is because when causal DAGs are small, modifying them may require less overhead than reading and writing them from cache. However, we observe substantial reduction in runtime using caching for the larger datasets SO and ACS. The absence of parallelization significantly worsens performance, especially on SO and ACS datasets, indicating its significant contribution in boosting the runtime performance across the board.

Overall, while all optimizations contribute to performance gains, parallelization yields the most substantial runtime improvements. Clustering also helps, though to a lesser extent. This experiment underscores the critical value of each optimization, highlighting how removal of any of them can degrade runtime efficiency.

%!TEX root = main.tex

%!TEX root=main.tex

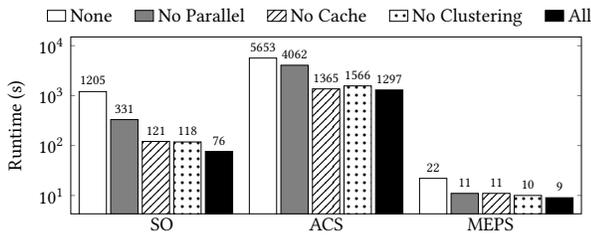
\begin{figure}[t]
\centering
\resizebox{0.95\columnwidth}{!}{
\begin{tikzpicture}
\begin{axis}[
    ybar,
    axis on top,
    bar width=0.45cm,
    width=10cm,
    height=4.5cm,
    ymin=0,
    ymax=15000,
    ymode=log,
    ytick={10,100,1000,10000},
    tick style={major tick length=2pt},
    minor tick style={draw=none},    
    enlarge x limits=false, 
    xmin=0.5,               % add space before first group
    xmax=3.5,               % optional: space after last group
    ylabel={Runtime (s)},
    xtick={1,2,3},
    xtick style={draw=none},
    xticklabels={SO, ACS, MEPS},
    x tick label style={rotate=0, anchor=east, yshift=-3pt, xshift=4mm},
    legend style={
        at={(-0.06,1.02)},       % left-aligned (x=0)
        anchor=south west,  % anchor bottom-left corner of legend box
        legend columns=5,
        draw=none,
        /tikz/every even column/.append style={column sep=0.3cm} 
    },
	legend image code/.code={
	  	\draw[draw=black,fill opacity=1] (0cm,-0.08cm) rectangle (0.3cm,0.1cm);
	},
	every node near coord/.append style={font=\scriptsize, black}
]

\addplot+[
    ybar,
    fill=black!0,
    draw=black, 
    point meta=explicit,
    nodes near coords*={
	\pgfmathprintnumber[
	    use comma=false,
	    1000 sep={}
	]{\pgfplotspointmeta}	
	},
] coordinates {
    (1, 1205) [1205]
    (2, 5653) [5653]
    (3, 22) [22]
};
\addlegendentry{None}

\addplot+[
    ybar,
    fill=black!50,
    draw=black, 
    point meta=explicit,
    nodes near coords*={
	\pgfmathprintnumber[
	    use comma=false,
	    1000 sep={}
	]{\pgfplotspointmeta}	
	},
] coordinates {
(1, 331.0) [331]
(2, 4062.0) [4062]
(3, 11.0) [11]
};
\addlegendentry{No Parallel}

\addplot+[
    ybar,
    pattern=north east lines, 
    draw=black,    
    point meta=explicit,
    nodes near coords*={
	\pgfmathprintnumber[
	    use comma=false,
	    1000 sep={}
	]{\pgfplotspointmeta}	
	},
] coordinates
{
(1, 121.0) [121]
(2, 1365.0) [1365]
(3, 11.0) [11]
};
\addlegendentry{No Cache}

\addplot+[
    ybar,
    pattern=dots,   
    pattern color=black,
    draw=black,    
    point meta=explicit,
    nodes near coords*={
	\pgfmathprintnumber[
	    use comma=false,
	    1000 sep={}
	]{\pgfplotspointmeta}	
	},
] coordinates
{
(1, 118.0) [118]
(2, 1566.0) [1566]
(3, 10.0) [10]
};
\addlegendentry{No Clustering}

\addplot+[
    ybar,
    fill=black!100,
    draw=black,    
    point meta=explicit,
    nodes near coords*={
	\pgfmathprintnumber[
	    use comma=false,
	    1000 sep={}
	]{\pgfplotspointmeta}	
	},
] coordinates {
(1, 76.0) [76]
(2, 1297.0) [1297]
(3, 9.0) [9]
};
\addlegendentry{All}
\end{axis}
\end{tikzpicture}}
\vspace{-3mm}
\caption{\small Effect of various settings of using optimization techniques on
runtime across three datasets. Note that the y-axis is in log scale.}
\vspace{3mm}
\label{fig:optimization_tests}
\end{figure}

\section{Conclusions and Future Work}
\label{sec:conc}

We have presented \sysName, a framework for discovering causal explanations for
disparities between two groups of interest. \sysName identifies data regions
where disparities are most pronounced (or reversed), and associates specific
factors that causally contribute to the disparity. We acknowledge that several
factors can influence the quality of the disparity explanations, including data
quality, the quality of the underlying causal model, and system parameters.
\newtext{In Section \ref{subsec:causal_DAG_robustness}, we showed that
meaningful results can still be obtained even with imperfect causal DAGs}. We
also provided insights on tuning system parameters to achieve satisfactory
results across different use cases and datasets.

% \looseness-1 
\revb{\sysName\ currently operates on single-relation databases, assuming no dependencies among tuples. This design ensures compliance with SUTVA~\cite{rubin2005causal}, a standard assumption in causal inference~\cite{youngmann2024summarized,galhotra2022hyper,li2025fair,zeng2025causal}. While this simplifies analysis, it limits applicability to multi-relation databases where tuple dependencies naturally arise. As discussed in~\cite{youngmann2024summarized}, extending treatment and grouping patterns to multi-table settings introduces substantial complexity and remains an open research direction. Nevertheless, \sysName\ can operate on normalized data by performing joins at additional computational cost, assuming no tuple-dependencies.
We also assume that a causal DAG is provided by domain experts, a common practice in causal analysis~\cite{youngmann2024summarized,galhotra2022hyper,li2025fair,zeng2025causal}. In practice, this requirement can be alleviated by leveraging existing causal discovery algorithms~\cite{glymour2019review}. As shown in Section.~6.3 (Fig.~2c), \sysName\ remains robust even when the DAG is noisy or imperfect.
Finally, the current implementation focuses on two-group comparisons. While extending to multiple groups is conceptually straightforward, it may introduce cognitive complexity for users, reducing interpretability. Supporting multi-group and multi-relation analyses thus presents a promising direction for future work.}

\small
\bibliographystyle{ACM-Reference-Format}
\balance
\bibliography{references}


\begin{thebibliography}{87}

%%% the Web address.

\ifx \showCODEN    \undefined \def \showCODEN     #1{\unskip}     \fi
\ifx \showDOI      \undefined \def \showDOI       #1{#1}\fi
\ifx \showISBNx    \undefined \def \showISBNx     #1{\unskip}     \fi
\ifx \showISBNxiii \undefined \def \showISBNxiii  #1{\unskip}     \fi
\ifx \showISSN     \undefined \def \showISSN      #1{\unskip}     \fi
\ifx \showLCCN     \undefined \def \showLCCN      #1{\unskip}     \fi
\ifx \shownote     \undefined \def \shownote      #1{#1}          \fi
\ifx \showarticletitle \undefined \def \showarticletitle #1{#1}   \fi
\ifx \showURL      \undefined \def \showURL       {\relax}        \fi
% invisible to TeX
\providecommand\bibfield[2]{#2}
\providecommand\bibinfo[2]{#2}
\providecommand\natexlab[1]{#1}
\providecommand\showeprint[2][]{arXiv:#2}

\bibitem[sta(2021)]%
        {stackoverflowreport}
 \bibinfo{year}{2021}\natexlab{}.
\newblock \bibinfo{title}{2021 Stackoverflow Developer Survey}.
\newblock
\newblock
\newblock
\shownote{\url{https://insights.stackoverflow.com/survey/2021}}.


\bibitem[{Agency for Healthcare Research and Quality (AHRQ)}(2024)]%
        {MEPS_Data_Overview}
\bibfield{author}{\bibinfo{person}{{Agency for Healthcare Research and Quality (AHRQ)}}.} \bibinfo{year}{2024}\natexlab{}.
\newblock \bibinfo{title}{Medical Expenditure Panel Survey (MEPS) - Data Overview}.
\newblock
\newblock
\urldef\tempurl%
\url{https://meps.ahrq.gov/mepsweb/data_stats/data_overview.jsp}
\showURL{%
\tempurl}
\newblock
\shownote{Accessed: 2024-01-30}.


\bibitem[Agmon et~al\mbox{.}(2024)]%
        {agmon2024finding}
\bibfield{author}{\bibinfo{person}{Shunit Agmon}, \bibinfo{person}{Amir Gilad}, \bibinfo{person}{Brit Youngmann}, \bibinfo{person}{Shahar Zoarets}, {and} \bibinfo{person}{Benny Kimelfeld}.} \bibinfo{year}{2024}\natexlab{}.
\newblock \showarticletitle{Finding Convincing Views to Endorse a Claim}.
\newblock \bibinfo{journal}{\emph{Proceedings of the VLDB Endowment}} \bibinfo{volume}{18}, \bibinfo{number}{2} (\bibinfo{year}{2024}), \bibinfo{pages}{439--452}.
\newblock


\bibitem[Agrawal et~al\mbox{.}(1994)]%
        {agrawal1994fast}
\bibfield{author}{\bibinfo{person}{Rakesh Agrawal}, \bibinfo{person}{Ramakrishnan Srikant}, {et~al\mbox{.}}} \bibinfo{year}{1994}\natexlab{}.
\newblock \showarticletitle{Fast algorithms for mining association rules}. In \bibinfo{booktitle}{\emph{Proc. 20th int. conf. very large data bases, VLDB}}, Vol.~\bibinfo{volume}{1215}. Santiago, Chile, \bibinfo{pages}{487--499}.
\newblock


\bibitem[Asudeh et~al\mbox{.}(2019)]%
        {asudeh2019assessing}
\bibfield{author}{\bibinfo{person}{Abolfazl Asudeh}, \bibinfo{person}{Zhongjun Jin}, {and} \bibinfo{person}{HV Jagadish}.} \bibinfo{year}{2019}\natexlab{}.
\newblock \showarticletitle{Assessing and remedying coverage for a given dataset}. In \bibinfo{booktitle}{\emph{2019 IEEE 35th International Conference on Data Engineering (ICDE)}}. IEEE, \bibinfo{pages}{554--565}.
\newblock


\bibitem[Austin(2011)]%
        {austin2011introduction}
\bibfield{author}{\bibinfo{person}{Peter~C Austin}.} \bibinfo{year}{2011}\natexlab{}.
\newblock \showarticletitle{An introduction to propensity score methods for reducing the effects of confounding in observational studies}.
\newblock \bibinfo{journal}{\emph{Multivariate behavioral research}} \bibinfo{volume}{46}, \bibinfo{number}{3} (\bibinfo{year}{2011}), \bibinfo{pages}{399--424}.
\newblock


\bibitem[Authors(2024)]%
        {repolink}
\bibfield{author}{\bibinfo{person}{Anonymous Authors}.} \bibinfo{year}{2024}\natexlab{}.
\newblock \bibinfo{title}{Causal Explanation for Disparity}.
\newblock
\newblock
\urldef\tempurl%
\url{https://anonymous.4open.science/r/CausalExplanationforDisparity-257A/}
\showURL{%
\tempurl}
\newblock
\shownote{Anonymous GitHub repository}.


\bibitem[Bholowalia and Kumar(2014)]%
        {bholowalia2014ebk}
\bibfield{author}{\bibinfo{person}{Purnima Bholowalia} {and} \bibinfo{person}{Arvind Kumar}.} \bibinfo{year}{2014}\natexlab{}.
\newblock \showarticletitle{EBK-means: A clustering technique based on elbow method and k-means in WSN}.
\newblock \bibinfo{journal}{\emph{International Journal of Computer Applications}} \bibinfo{volume}{105}, \bibinfo{number}{9} (\bibinfo{year}{2014}).
\newblock


\bibitem[Bickel et~al\mbox{.}(1975)]%
        {doi:10.1126/science.187.4175.398}
\bibfield{author}{\bibinfo{person}{P.~J. Bickel}, \bibinfo{person}{E.~A. Hammel}, {and} \bibinfo{person}{J.~W. O'Connell}.} \bibinfo{year}{1975}\natexlab{}.
\newblock \showarticletitle{Sex Bias in Graduate Admissions: Data from Berkeley}.
\newblock \bibinfo{journal}{\emph{Science}} \bibinfo{volume}{187}, \bibinfo{number}{4175} (\bibinfo{year}{1975}), \bibinfo{pages}{398--404}.
\newblock
\urldef\tempurl%
\url{https://doi.org/10.1126/science.187.4175.398}
\showDOI{\tempurl}
\showeprint{https://www.science.org/doi/pdf/10.1126/science.187.4175.398}


\bibitem[Bidoit et~al\mbox{.}(2014)]%
        {bidoit2014query}
\bibfield{author}{\bibinfo{person}{Nicole Bidoit}, \bibinfo{person}{Melanie Herschel}, {and} \bibinfo{person}{Katerina Tzompanaki}.} \bibinfo{year}{2014}\natexlab{}.
\newblock \showarticletitle{Query-based why-not provenance with nedexplain}. In \bibinfo{booktitle}{\emph{Extending database technology (EDBT)}}.
\newblock


\bibitem[Bourhis et~al\mbox{.}(2020)]%
        {BourhisDM20}
\bibfield{author}{\bibinfo{person}{Pierre Bourhis}, \bibinfo{person}{Daniel Deutch}, {and} \bibinfo{person}{Yuval Moskovitch}.} \bibinfo{year}{2020}\natexlab{}.
\newblock \showarticletitle{Equivalence-Invariant Algebraic Provenance for Hyperplane Update Queries}. In \bibinfo{booktitle}{\emph{Proceedings of the 2020 International Conference on Management of Data, {SIGMOD} Conference 2020, online conference [Portland, OR, USA], June 14-19, 2020}}. \bibinfo{publisher}{{ACM}}, \bibinfo{pages}{415--429}.
\newblock
\urldef\tempurl%
\url{https://doi.org/10.1145/3318464.3380578}
\showDOI{\tempurl}


\bibitem[Cabrera et~al\mbox{.}(2019)]%
        {CabreraEHKMC19}
\bibfield{author}{\bibinfo{person}{{\'{A}}ngel~Alexander Cabrera}, \bibinfo{person}{Will Epperson}, \bibinfo{person}{Fred Hohman}, \bibinfo{person}{Minsuk Kahng}, \bibinfo{person}{Jamie Morgenstern}, {and} \bibinfo{person}{Duen~Horng Chau}.} \bibinfo{year}{2019}\natexlab{}.
\newblock \showarticletitle{{FAIRVIS:} Visual Analytics for Discovering Intersectional Bias in Machine Learning}. In \bibinfo{booktitle}{\emph{14th {IEEE} Conference on Visual Analytics Science and Technology, {IEEE} {VAST} 2019, Vancouver, BC, Canada, October 20-25, 2019}}, \bibfield{editor}{\bibinfo{person}{Remco Chang}, \bibinfo{person}{Daniel~A. Keim}, {and} \bibinfo{person}{Ross Maciejewski}} (Eds.). \bibinfo{publisher}{{IEEE}}, \bibinfo{pages}{46--56}.
\newblock
\urldef\tempurl%
\url{https://doi.org/10.1109/VAST47406.2019.8986948}
\showDOI{\tempurl}


\bibitem[Chapman and Jagadish(2009)]%
        {chapman2009not}
\bibfield{author}{\bibinfo{person}{Adriane Chapman} {and} \bibinfo{person}{HV Jagadish}.} \bibinfo{year}{2009}\natexlab{}.
\newblock \showarticletitle{Why not?}. In \bibinfo{booktitle}{\emph{Proceedings of the 2009 ACM SIGMOD International Conference on Management of data}}. \bibinfo{pages}{523--534}.
\newblock


\bibitem[Chickering(2002)]%
        {chickering2002optimal}
\bibfield{author}{\bibinfo{person}{D.M Chickering}.} \bibinfo{year}{2002}\natexlab{}.
\newblock \showarticletitle{Optimal structure identification with greedy search}.
\newblock \bibinfo{journal}{\emph{JMLR}} \bibinfo{volume}{3}, \bibinfo{number}{Nov} (\bibinfo{year}{2002}), \bibinfo{pages}{507--554}.
\newblock


\bibitem[Chung et~al\mbox{.}(2019)]%
        {chung2019automated}
\bibfield{author}{\bibinfo{person}{Yeounoh Chung}, \bibinfo{person}{Tim Kraska}, \bibinfo{person}{Neoklis Polyzotis}, \bibinfo{person}{Ki~Hyun Tae}, {and} \bibinfo{person}{Steven~Euijong Whang}.} \bibinfo{year}{2019}\natexlab{}.
\newblock \showarticletitle{Automated data slicing for model validation: A big data-ai integration approach}.
\newblock \bibinfo{journal}{\emph{IEEE Transactions on Knowledge and Data Engineering}} \bibinfo{volume}{32}, \bibinfo{number}{12} (\bibinfo{year}{2019}), \bibinfo{pages}{2284--2296}.
\newblock


\bibitem[Deutch et~al\mbox{.}(2020)]%
        {DeutchFG20}
\bibfield{author}{\bibinfo{person}{Daniel Deutch}, \bibinfo{person}{Nave Frost}, {and} \bibinfo{person}{Amir Gilad}.} \bibinfo{year}{2020}\natexlab{}.
\newblock \showarticletitle{Explaining Natural Language query results}.
\newblock \bibinfo{journal}{\emph{{VLDB} J.}} \bibinfo{volume}{29}, \bibinfo{number}{1} (\bibinfo{year}{2020}), \bibinfo{pages}{485--508}.
\newblock


\bibitem[Deutch et~al\mbox{.}(2022)]%
        {DBLP:journals/pvldb/DeutchGMMS22}
\bibfield{author}{\bibinfo{person}{Daniel Deutch}, \bibinfo{person}{Amir Gilad}, \bibinfo{person}{Tova Milo}, \bibinfo{person}{Amit Mualem}, {and} \bibinfo{person}{Amit Somech}.} \bibinfo{year}{2022}\natexlab{}.
\newblock \showarticletitle{{FEDEX:} An Explainability Framework for Data Exploration Steps}.
\newblock \bibinfo{journal}{\emph{Proc. {VLDB} Endow.}} \bibinfo{volume}{15}, \bibinfo{number}{13} (\bibinfo{year}{2022}), \bibinfo{pages}{3854--3868}.
\newblock
\urldef\tempurl%
\url{https://www.vldb.org/pvldb/vol15/p3854-gilad.pdf}
\showURL{%
\tempurl}


\bibitem[Deutch et~al\mbox{.}(2015)]%
        {DeutchGM15}
\bibfield{author}{\bibinfo{person}{Daniel Deutch}, \bibinfo{person}{Amir Gilad}, {and} \bibinfo{person}{Yuval Moskovitch}.} \bibinfo{year}{2015}\natexlab{}.
\newblock \showarticletitle{Selective Provenance for Datalog Programs Using Top-K Queries}.
\newblock \bibinfo{journal}{\emph{Proc. {VLDB} Endow.}} \bibinfo{volume}{8}, \bibinfo{number}{12} (\bibinfo{year}{2015}), \bibinfo{pages}{1394--1405}.
\newblock
\urldef\tempurl%
\url{https://doi.org/10.14778/2824032.2824039}
\showDOI{\tempurl}


\bibitem[El~Gebaly et~al\mbox{.}(2014)]%
        {el2014interpretable}
\bibfield{author}{\bibinfo{person}{Kareem El~Gebaly}, \bibinfo{person}{Parag Agrawal}, \bibinfo{person}{Lukasz Golab}, \bibinfo{person}{Flip Korn}, {and} \bibinfo{person}{Divesh Srivastava}.} \bibinfo{year}{2014}\natexlab{}.
\newblock \showarticletitle{Interpretable and informative explanations of outcomes}.
\newblock \bibinfo{journal}{\emph{Proceedings of the VLDB Endowment}} \bibinfo{volume}{8}, \bibinfo{number}{1} (\bibinfo{year}{2014}), \bibinfo{pages}{61--72}.
\newblock


\bibitem[Galhotra et~al\mbox{.}(2022)]%
        {galhotra2022hyper}
\bibfield{author}{\bibinfo{person}{Sainyam Galhotra}, \bibinfo{person}{Amir Gilad}, \bibinfo{person}{Sudeepa Roy}, {and} \bibinfo{person}{Babak Salimi}.} \bibinfo{year}{2022}\natexlab{}.
\newblock \showarticletitle{Hyper: Hypothetical reasoning with what-if and how-to queries using a probabilistic causal approach}. In \bibinfo{booktitle}{\emph{Proceedings of the 2022 International Conference on Management of Data}}. \bibinfo{pages}{1598--1611}.
\newblock


\bibitem[Galhotra et~al\mbox{.}(2021)]%
        {GalhotraPS21}
\bibfield{author}{\bibinfo{person}{Sainyam Galhotra}, \bibinfo{person}{Romila Pradhan}, {and} \bibinfo{person}{Babak Salimi}.} \bibinfo{year}{2021}\natexlab{}.
\newblock \showarticletitle{Explaining Black-Box Algorithms Using Probabilistic Contrastive Counterfactuals}. In \bibinfo{booktitle}{\emph{SIGMOD}}. \bibinfo{publisher}{{ACM}}, \bibinfo{pages}{577--590}.
\newblock


\bibitem[Geerts et~al\mbox{.}(2004)]%
        {geerts2004tiling}
\bibfield{author}{\bibinfo{person}{Floris Geerts}, \bibinfo{person}{Bart Goethals}, {and} \bibinfo{person}{Taneli Mielik{\"a}inen}.} \bibinfo{year}{2004}\natexlab{}.
\newblock \showarticletitle{Tiling databases}. In \bibinfo{booktitle}{\emph{Discovery Science: 7th International Conference, DS 2004, Padova, Italy, October 2-5, 2004. Proceedings 7}}. Springer, \bibinfo{pages}{278--289}.
\newblock


\bibitem[Glymour et~al\mbox{.}(2019)]%
        {glymour2019review}
\bibfield{author}{\bibinfo{person}{Clark Glymour}, \bibinfo{person}{Kun Zhang}, {and} \bibinfo{person}{Peter Spirtes}.} \bibinfo{year}{2019}\natexlab{}.
\newblock \showarticletitle{Review of causal discovery methods based on graphical models}.
\newblock \bibinfo{journal}{\emph{Frontiers in genetics}}  \bibinfo{volume}{10} (\bibinfo{year}{2019}), \bibinfo{pages}{524}.
\newblock


\bibitem[Godbolt(2011)]%
        {godbolt2011black}
\bibfield{author}{\bibinfo{person}{Ricky~Charles Godbolt}.} \bibinfo{year}{2011}\natexlab{}.
\newblock \emph{\bibinfo{title}{Black and Blue: African Americans, Blue-Collar Bias, and the Construction Industry in Prince George's County, Maryland}}.
\newblock \bibinfo{thesistype}{Ph.\,D. Dissertation}. \bibinfo{school}{University of Phoenix}.
\newblock


\bibitem[Hipp et~al\mbox{.}(2000)]%
        {hipp2000algorithms}
\bibfield{author}{\bibinfo{person}{Jochen Hipp}, \bibinfo{person}{Ulrich G{\"u}ntzer}, {and} \bibinfo{person}{Gholamreza Nakhaeizadeh}.} \bibinfo{year}{2000}\natexlab{}.
\newblock \showarticletitle{Algorithms for association rule mining—a general survey and comparison}.
\newblock \bibinfo{journal}{\emph{ACM sigkdd explorations newsletter}} \bibinfo{volume}{2}, \bibinfo{number}{1} (\bibinfo{year}{2000}), \bibinfo{pages}{58--64}.
\newblock


\bibitem[Holland(1986)]%
        {holland1986statistics}
\bibfield{author}{\bibinfo{person}{Paul~W Holland}.} \bibinfo{year}{1986}\natexlab{}.
\newblock \showarticletitle{Statistics and causal inference}.
\newblock \bibinfo{journal}{\emph{Journal of the American statistical Association}} \bibinfo{volume}{81}, \bibinfo{number}{396} (\bibinfo{year}{1986}), \bibinfo{pages}{945--960}.
\newblock


\bibitem[Joglekar et~al\mbox{.}(2019)]%
        {DBLP:journals/tkde/JoglekarGP19}
\bibfield{author}{\bibinfo{person}{Manas Joglekar}, \bibinfo{person}{Hector Garcia{-}Molina}, {and} \bibinfo{person}{Aditya~G. Parameswaran}.} \bibinfo{year}{2019}\natexlab{}.
\newblock \showarticletitle{Interactive Data Exploration with Smart Drill-Down}.
\newblock \bibinfo{journal}{\emph{{IEEE} Trans. Knowl. Data Eng.}} \bibinfo{volume}{31}, \bibinfo{number}{1} (\bibinfo{year}{2019}), \bibinfo{pages}{46--60}.
\newblock
\urldef\tempurl%
\url{https://doi.org/10.1109/TKDE.2017.2685998}
\showDOI{\tempurl}


\bibitem[Julious and Mullee(1994)]%
        {Julious1480}
\bibfield{author}{\bibinfo{person}{Steven~A Julious} {and} \bibinfo{person}{Mark~A Mullee}.} \bibinfo{year}{1994}\natexlab{}.
\newblock \showarticletitle{Confounding and Simpson{\textquoteright}s paradox}.
\newblock \bibinfo{journal}{\emph{BMJ}} \bibinfo{volume}{309}, \bibinfo{number}{6967} (\bibinfo{year}{1994}), \bibinfo{pages}{1480--1481}.
\newblock
\showISSN{0959-8138}
\urldef\tempurl%
\url{https://doi.org/10.1136/bmj.309.6967.1480}
\showDOI{\tempurl}
\showeprint{https://www.bmj.com/content}


\bibitem[Karimi et~al\mbox{.}(2023)]%
        {KarimiBSV23}
\bibfield{author}{\bibinfo{person}{Amir{-}Hossein Karimi}, \bibinfo{person}{Gilles Barthe}, \bibinfo{person}{Bernhard Sch{\"{o}}lkopf}, {and} \bibinfo{person}{Isabel Valera}.} \bibinfo{year}{2023}\natexlab{}.
\newblock \showarticletitle{A Survey of Algorithmic Recourse: Contrastive Explanations and Consequential Recommendations}.
\newblock \bibinfo{journal}{\emph{Comput. Surveys}} \bibinfo{volume}{55}, \bibinfo{number}{5} (\bibinfo{year}{2023}), \bibinfo{pages}{95:1--95:29}.
\newblock


\bibitem[Karimi et~al\mbox{.}(2021)]%
        {KarimiSV21}
\bibfield{author}{\bibinfo{person}{Amir{-}Hossein Karimi}, \bibinfo{person}{Bernhard Sch{\"{o}}lkopf}, {and} \bibinfo{person}{Isabel Valera}.} \bibinfo{year}{2021}\natexlab{}.
\newblock \showarticletitle{Algorithmic Recourse: from Counterfactual Explanations to Interventions}. In \bibinfo{booktitle}{\emph{FAccT}}. \bibinfo{publisher}{{ACM}}, \bibinfo{pages}{353--362}.
\newblock


\bibitem[Kim et~al\mbox{.}(2014)]%
        {kim2014bayesian}
\bibfield{author}{\bibinfo{person}{Been Kim}, \bibinfo{person}{Cynthia Rudin}, {and} \bibinfo{person}{Julie~A Shah}.} \bibinfo{year}{2014}\natexlab{}.
\newblock \showarticletitle{The bayesian case model: A generative approach for case-based reasoning and prototype classification}.
\newblock \bibinfo{journal}{\emph{Advances in neural information processing systems}}  \bibinfo{volume}{27} (\bibinfo{year}{2014}).
\newblock


\bibitem[Kumbhare and Chobe(2014)]%
        {kumbhare2014overview}
\bibfield{author}{\bibinfo{person}{Trupti~A Kumbhare} {and} \bibinfo{person}{Santosh~V Chobe}.} \bibinfo{year}{2014}\natexlab{}.
\newblock \showarticletitle{An overview of association rule mining algorithms}.
\newblock \bibinfo{journal}{\emph{International Journal of Computer Science and Information Technologies}} \bibinfo{volume}{5}, \bibinfo{number}{1} (\bibinfo{year}{2014}), \bibinfo{pages}{927--930}.
\newblock


\bibitem[Lakshmanan et~al\mbox{.}(2002)]%
        {lakshmanan2002quotient}
\bibfield{author}{\bibinfo{person}{Laks~VS Lakshmanan}, \bibinfo{person}{Jian Pei}, {and} \bibinfo{person}{Jiawei Han}.} \bibinfo{year}{2002}\natexlab{}.
\newblock \showarticletitle{Quotient cube: How to summarize the semantics of a data cube}. In \bibinfo{booktitle}{\emph{VLDB'02: Proceedings of the 28th International Conference on Very Large Databases}}. Elsevier, \bibinfo{pages}{778--789}.
\newblock


\bibitem[Lee et~al\mbox{.}(2020)]%
        {lee13approximate}
\bibfield{author}{\bibinfo{person}{Seokki Lee}, \bibinfo{person}{Bertram Lud{\"a}scher}, {and} \bibinfo{person}{Boris Glavic}.} \bibinfo{year}{2020}\natexlab{}.
\newblock \showarticletitle{Approximate Summaries for Why and Why-not Provenance}.
\newblock \bibinfo{journal}{\emph{Proceedings of the VLDB Endowment}} \bibinfo{volume}{13}, \bibinfo{number}{6} (\bibinfo{year}{2020}).
\newblock


\bibitem[Li et~al\mbox{.}(2025)]%
        {li2025fair}
\bibfield{author}{\bibinfo{person}{Benton Li}, \bibinfo{person}{Nativ Levy}, \bibinfo{person}{Brit Youngmann}, \bibinfo{person}{Sainyam Galhotra}, {and} \bibinfo{person}{Sudeepa Roy}.} \bibinfo{year}{2025}\natexlab{}.
\newblock \showarticletitle{Fair and Actionable Causal Prescription Ruleset}.
\newblock \bibinfo{journal}{\emph{Proceedings of the ACM on Management of Data}} \bibinfo{volume}{3}, \bibinfo{number}{3} (\bibinfo{year}{2025}), \bibinfo{pages}{1--28}.
\newblock


\bibitem[Li et~al\mbox{.}(2021)]%
        {li2021putting}
\bibfield{author}{\bibinfo{person}{Chenjie Li}, \bibinfo{person}{Zhengjie Miao}, \bibinfo{person}{Qitian Zeng}, \bibinfo{person}{Boris Glavic}, {and} \bibinfo{person}{Sudeepa Roy}.} \bibinfo{year}{2021}\natexlab{}.
\newblock \showarticletitle{Putting Things into Context: Rich Explanations for Query Answers using Join Graphs}. In \bibinfo{booktitle}{\emph{Proceedings of the 2021 International Conference on Management of Data}}. \bibinfo{pages}{1051--1063}.
\newblock


\bibitem[Li et~al\mbox{.}(2023)]%
        {LiMJ23}
\bibfield{author}{\bibinfo{person}{Jinyang Li}, \bibinfo{person}{Yuval Moskovitch}, {and} \bibinfo{person}{H.~V. Jagadish}.} \bibinfo{year}{2023}\natexlab{}.
\newblock \showarticletitle{Detection of Groups with Biased Representation in Ranking}. In \bibinfo{booktitle}{\emph{39th {IEEE} International Conference on Data Engineering, {ICDE} 2023, Anaheim, CA, USA, April 3-7, 2023}}. \bibinfo{publisher}{{IEEE}}, \bibinfo{pages}{2167--2179}.
\newblock
\urldef\tempurl%
\url{https://doi.org/10.1109/ICDE55515.2023.00168}
\showDOI{\tempurl}


\bibitem[Lin et~al\mbox{.}(2021)]%
        {lin2021detecting}
\bibfield{author}{\bibinfo{person}{Yin Lin}, \bibinfo{person}{Brit Youngmann}, \bibinfo{person}{Yuval Moskovitch}, \bibinfo{person}{HV Jagadish}, {and} \bibinfo{person}{Tova Milo}.} \bibinfo{year}{2021}\natexlab{}.
\newblock \showarticletitle{On detecting cherry-picked generalizations}.
\newblock \bibinfo{journal}{\emph{Proceedings of the VLDB Endowment}} \bibinfo{volume}{15}, \bibinfo{number}{1} (\bibinfo{year}{2021}), \bibinfo{pages}{59--71}.
\newblock


\bibitem[Lou et~al\mbox{.}(2013)]%
        {lou2013accurate}
\bibfield{author}{\bibinfo{person}{Yin Lou}, \bibinfo{person}{Rich Caruana}, \bibinfo{person}{Johannes Gehrke}, {and} \bibinfo{person}{Giles Hooker}.} \bibinfo{year}{2013}\natexlab{}.
\newblock \showarticletitle{Accurate intelligible models with pairwise interactions}. In \bibinfo{booktitle}{\emph{Proceedings of the 19th ACM SIGKDD international conference on Knowledge discovery and data mining}}. \bibinfo{pages}{623--631}.
\newblock


\bibitem[Ma et~al\mbox{.}(2023)]%
        {abs-2207-12718}
\bibfield{author}{\bibinfo{person}{Pingchuan Ma}, \bibinfo{person}{Rui Ding}, \bibinfo{person}{Shuai Wang}, \bibinfo{person}{Shi Han}, {and} \bibinfo{person}{Dongmei Zhang}.} \bibinfo{year}{2023}\natexlab{}.
\newblock \showarticletitle{XInsight: EXplainable Data Analysis Through The Lens of Causality}.
\newblock \bibinfo{journal}{\emph{Proc. ACM Manag. Data}}, Article \bibinfo{articleno}{156} (\bibinfo{date}{jun} \bibinfo{year}{2023}), \bibinfo{numpages}{27}~pages.
\newblock


\bibitem[Meliou et~al\mbox{.}(2009)]%
        {meliou2009so}
\bibfield{author}{\bibinfo{person}{Alexandra Meliou}, \bibinfo{person}{Wolfgang Gatterbauer}, \bibinfo{person}{Katherine~F Moore}, {and} \bibinfo{person}{Dan Suciu}.} \bibinfo{year}{2009}\natexlab{}.
\newblock \showarticletitle{Why so? or why no? functional causality for explaining query answers}.
\newblock \bibinfo{journal}{\emph{arXiv preprint arXiv:0912.5340}} (\bibinfo{year}{2009}).
\newblock


\bibitem[Meliou et~al\mbox{.}(2010)]%
        {meliou2010complexity}
\bibfield{author}{\bibinfo{person}{Alexandra Meliou}, \bibinfo{person}{Wolfgang Gatterbauer}, \bibinfo{person}{Katherine~F Moore}, {and} \bibinfo{person}{Dan Suciu}.} \bibinfo{year}{2010}\natexlab{}.
\newblock \showarticletitle{The Complexity of Causality and Responsibility for Query Answers and non-Answers}.
\newblock \bibinfo{journal}{\emph{Proceedings of the VLDB Endowment}} \bibinfo{volume}{4}, \bibinfo{number}{1} (\bibinfo{year}{2010}).
\newblock


\bibitem[Miao et~al\mbox{.}(2019)]%
        {miao2019going}
\bibfield{author}{\bibinfo{person}{Zhengjie Miao}, \bibinfo{person}{Qitian Zeng}, \bibinfo{person}{Boris Glavic}, {and} \bibinfo{person}{Sudeepa Roy}.} \bibinfo{year}{2019}\natexlab{}.
\newblock \showarticletitle{Going beyond provenance: Explaining query answers with pattern-based counterbalances}. In \bibinfo{booktitle}{\emph{Proceedings of the 2019 International Conference on Management of Data}}. \bibinfo{pages}{485--502}.
\newblock


\bibitem[Miratrix et~al\mbox{.}(2013)]%
        {miratrix2013adjusting}
\bibfield{author}{\bibinfo{person}{Luke~W Miratrix}, \bibinfo{person}{Jasjeet~S Sekhon}, {and} \bibinfo{person}{Bin Yu}.} \bibinfo{year}{2013}\natexlab{}.
\newblock \showarticletitle{Adjusting treatment effect estimates by post-stratification in randomized experiments}.
\newblock \bibinfo{journal}{\emph{Journal of the Royal Statistical Society Series B: Statistical Methodology}} \bibinfo{volume}{75}, \bibinfo{number}{2} (\bibinfo{year}{2013}), \bibinfo{pages}{369--396}.
\newblock


\bibitem[Moskovitch et~al\mbox{.}(2023)]%
        {MoskovitchLJ23}
\bibfield{author}{\bibinfo{person}{Yuval Moskovitch}, \bibinfo{person}{Jinyang Li}, {and} \bibinfo{person}{H.~V. Jagadish}.} \bibinfo{year}{2023}\natexlab{}.
\newblock \showarticletitle{Dexer: Detecting and Explaining Biased Representation in Ranking}. In \bibinfo{booktitle}{\emph{Companion of the 2023 International Conference on Management of Data, {SIGMOD/PODS} 2023, Seattle, WA, USA, June 18-23, 2023}}. \bibinfo{publisher}{{ACM}}, \bibinfo{pages}{159--162}.
\newblock
\urldef\tempurl%
\url{https://doi.org/10.1145/3555041.3589725}
\showDOI{\tempurl}


\bibitem[Moumoulidou et~al\mbox{.}(2021)]%
        {DBLP:conf/icdt/Moumoulidou0M21}
\bibfield{author}{\bibinfo{person}{Zafeiria Moumoulidou}, \bibinfo{person}{Andrew McGregor}, {and} \bibinfo{person}{Alexandra Meliou}.} \bibinfo{year}{2021}\natexlab{}.
\newblock \showarticletitle{Diverse Data Selection under Fairness Constraints}. In \bibinfo{booktitle}{\emph{24th International Conference on Database Theory, {ICDT} 2021, March 23-26, 2021, Nicosia, Cyprus}} \emph{(\bibinfo{series}{LIPIcs}, Vol.~\bibinfo{volume}{186})}, \bibfield{editor}{\bibinfo{person}{Ke~Yi} {and} \bibinfo{person}{Zhewei Wei}} (Eds.). \bibinfo{publisher}{Schloss Dagstuhl - Leibniz-Zentrum f{\"{u}}r Informatik}, \bibinfo{pages}{13:1--13:25}.
\newblock
\urldef\tempurl%
\url{https://doi.org/10.4230/LIPICS.ICDT.2021.13}
\showDOI{\tempurl}


\bibitem[O’donnell et~al\mbox{.}(2006)]%
        {o2006incorporating}
\bibfield{author}{\bibinfo{person}{RT O’donnell}, \bibinfo{person}{Ann~E Nicholson}, \bibinfo{person}{B Han}, \bibinfo{person}{Kevin~B Korb}, \bibinfo{person}{MJ Alam}, {and} \bibinfo{person}{LR Hope}.} \bibinfo{year}{2006}\natexlab{}.
\newblock \showarticletitle{Incorporating expert elicited structural information in the CaMML causal discovery program}. In \bibinfo{booktitle}{\emph{Proceedings of the 19th Australian Joint Conference on Artificial Intelligence: Advances in Artificial Intelligence}}. \bibinfo{pages}{1--16}.
\newblock


\bibitem[Pastor et~al\mbox{.}(2021)]%
        {pastor2021looking}
\bibfield{author}{\bibinfo{person}{Eliana Pastor}, \bibinfo{person}{Luca De~Alfaro}, {and} \bibinfo{person}{Elena Baralis}.} \bibinfo{year}{2021}\natexlab{}.
\newblock \showarticletitle{Looking for trouble: Analyzing classifier behavior via pattern divergence}. In \bibinfo{booktitle}{\emph{Proceedings of the 2021 International Conference on Management of Data}}. \bibinfo{pages}{1400--1412}.
\newblock


\bibitem[Pearl(2009)]%
        {pearl2009causal}
\bibfield{author}{\bibinfo{person}{Judea Pearl}.} \bibinfo{year}{2009}\natexlab{}.
\newblock \showarticletitle{Causal inference in statistics: An overview}.
\newblock  (\bibinfo{year}{2009}).
\newblock


\bibitem[Plecko and Bareinboim(2023)]%
        {plecko2023causal}
\bibfield{author}{\bibinfo{person}{Drago Plecko} {and} \bibinfo{person}{Elias Bareinboim}.} \bibinfo{year}{2023}\natexlab{}.
\newblock \showarticletitle{Causal fairness for outcome control}.
\newblock \bibinfo{journal}{\emph{Advances in Neural Information Processing Systems}}  \bibinfo{volume}{36} (\bibinfo{year}{2023}), \bibinfo{pages}{47575--47597}.
\newblock


\bibitem[Plecko and Bareinboim(2024)]%
        {plecko2022causal}
\bibfield{author}{\bibinfo{person}{Drago Plecko} {and} \bibinfo{person}{Elias Bareinboim}.} \bibinfo{year}{2024}\natexlab{}.
\newblock \showarticletitle{Causal fairness analysis}.
\newblock \bibinfo{journal}{\emph{ECAI}} (\bibinfo{year}{2024}).
\newblock


\bibitem[Rosenbaum and Rubin(1983)]%
        {rosenbaum1983central}
\bibfield{author}{\bibinfo{person}{Paul~R Rosenbaum} {and} \bibinfo{person}{Donald~B Rubin}.} \bibinfo{year}{1983}\natexlab{}.
\newblock \showarticletitle{The central role of the propensity score in observational studies for causal effects}.
\newblock \bibinfo{journal}{\emph{Biometrika}} \bibinfo{volume}{70}, \bibinfo{number}{1} (\bibinfo{year}{1983}), \bibinfo{pages}{41--55}.
\newblock


\bibitem[Roy et~al\mbox{.}(2015)]%
        {roy2015explaining}
\bibfield{author}{\bibinfo{person}{Sudeepa Roy}, \bibinfo{person}{Laurel Orr}, {and} \bibinfo{person}{Dan Suciu}.} \bibinfo{year}{2015}\natexlab{}.
\newblock \showarticletitle{Explaining query answers with explanation-ready databases}.
\newblock \bibinfo{journal}{\emph{Proceedings of the VLDB Endowment}} \bibinfo{volume}{9}, \bibinfo{number}{4} (\bibinfo{year}{2015}), \bibinfo{pages}{348--359}.
\newblock


\bibitem[Roy and Suciu(2014)]%
        {roy2014formal}
\bibfield{author}{\bibinfo{person}{Sudeepa Roy} {and} \bibinfo{person}{Dan Suciu}.} \bibinfo{year}{2014}\natexlab{}.
\newblock \showarticletitle{A formal approach to finding explanations for database queries}. In \bibinfo{booktitle}{\emph{Proceedings of the 2014 ACM SIGMOD international conference on Management of data}}. \bibinfo{pages}{1579--1590}.
\newblock


\bibitem[Rubin(1971)]%
        {rubin1971use}
\bibfield{author}{\bibinfo{person}{Donald~Bruce Rubin}.} \bibinfo{year}{1971}\natexlab{}.
\newblock \emph{\bibinfo{title}{The use of matched sampling and regression adjustment in observational studies}}.
\newblock \bibinfo{thesistype}{Ph.\,D. Dissertation}. \bibinfo{school}{Harvard University}.
\newblock


\bibitem[Rubin(2005)]%
        {rubin2005causal}
\bibfield{author}{\bibinfo{person}{Donald~B Rubin}.} \bibinfo{year}{2005}\natexlab{}.
\newblock \showarticletitle{Causal inference using potential outcomes: Design, modeling, decisions}.
\newblock \bibinfo{journal}{\emph{J. Amer. Statist. Assoc.}} \bibinfo{volume}{100}, \bibinfo{number}{469} (\bibinfo{year}{2005}), \bibinfo{pages}{322--331}.
\newblock


\bibitem[Sagi and Rokach(2021)]%
        {sagi2021approximating}
\bibfield{author}{\bibinfo{person}{Omer Sagi} {and} \bibinfo{person}{Lior Rokach}.} \bibinfo{year}{2021}\natexlab{}.
\newblock \showarticletitle{Approximating XGBoost with an interpretable decision tree}.
\newblock \bibinfo{journal}{\emph{Information Sciences}}  \bibinfo{volume}{572} (\bibinfo{year}{2021}), \bibinfo{pages}{522--542}.
\newblock


\bibitem[Salimi et~al\mbox{.}(2018)]%
        {salimi2018bias}
\bibfield{author}{\bibinfo{person}{Babak Salimi}, \bibinfo{person}{Johannes Gehrke}, {and} \bibinfo{person}{Dan Suciu}.} \bibinfo{year}{2018}\natexlab{}.
\newblock \showarticletitle{Bias in olap queries: Detection, explanation, and removal}. In \bibinfo{booktitle}{\emph{Proceedings of the 2018 International Conference on Management of Data}}. \bibinfo{pages}{1021--1035}.
\newblock


\bibitem[Sarawagi(2000)]%
        {sarawagi2000user}
\bibfield{author}{\bibinfo{person}{Sunita Sarawagi}.} \bibinfo{year}{2000}\natexlab{}.
\newblock \showarticletitle{User-adaptive exploration of multidimensional data}. In \bibinfo{booktitle}{\emph{VLDB}}. ResearchGate GmbH, \bibinfo{pages}{307--316}.
\newblock


\bibitem[Sarawagi(2001)]%
        {sarawagi2001user}
\bibfield{author}{\bibinfo{person}{Sunita Sarawagi}.} \bibinfo{year}{2001}\natexlab{}.
\newblock \showarticletitle{User-cognizant multidimensional analysis}.
\newblock \bibinfo{journal}{\emph{The VLDB Journal}}  \bibinfo{volume}{10} (\bibinfo{year}{2001}), \bibinfo{pages}{224--239}.
\newblock


\bibitem[Sarawagi et~al\mbox{.}(1998)]%
        {sarawagi1998discovery}
\bibfield{author}{\bibinfo{person}{Sunita Sarawagi}, \bibinfo{person}{Rakesh Agrawal}, {and} \bibinfo{person}{Nimrod Megiddo}.} \bibinfo{year}{1998}\natexlab{}.
\newblock \showarticletitle{Discovery-driven exploration of OLAP data cubes}. In \bibinfo{booktitle}{\emph{Advances in Database Technology—EDBT'98: 6th International Conference on Extending Database Technology Valencia, Spain, March 23--27, 1998 Proceedings 6}}. Springer, \bibinfo{pages}{168--182}.
\newblock


\bibitem[Sathe and Sarawagi(2001)]%
        {sathe2001intelligent}
\bibfield{author}{\bibinfo{person}{Gayatri Sathe} {and} \bibinfo{person}{Sunita Sarawagi}.} \bibinfo{year}{2001}\natexlab{}.
\newblock \showarticletitle{Intelligent rollups in multidimensional OLAP data}. In \bibinfo{booktitle}{\emph{VLDB}}. \bibinfo{pages}{307--316}.
\newblock


\bibitem[Schielzeth(2010)]%
        {schielzeth2010simple}
\bibfield{author}{\bibinfo{person}{Holger Schielzeth}.} \bibinfo{year}{2010}\natexlab{}.
\newblock \showarticletitle{Simple means to improve the interpretability of regression coefficients}.
\newblock \bibinfo{journal}{\emph{Methods in Ecology and Evolution}} \bibinfo{volume}{1}, \bibinfo{number}{2} (\bibinfo{year}{2010}), \bibinfo{pages}{103--113}.
\newblock


\bibitem[Sharma and Kiciman(2020)]%
        {dowhypaper}
\bibfield{author}{\bibinfo{person}{Amit Sharma} {and} \bibinfo{person}{Emre Kiciman}.} \bibinfo{year}{2020}\natexlab{}.
\newblock \showarticletitle{DoWhy: An End-to-End Library for Causal Inference}.
\newblock \bibinfo{journal}{\emph{arXiv preprint arXiv:2011.04216}} (\bibinfo{year}{2020}).
\newblock


\bibitem[Spirtes et~al\mbox{.}(2000)]%
        {spirtes2000causation}
\bibfield{author}{\bibinfo{person}{P. Spirtes} {et~al\mbox{.}}} \bibinfo{year}{2000}\natexlab{}.
\newblock \bibinfo{booktitle}{\emph{Causation, prediction, and search}}.
\newblock \bibinfo{publisher}{MIT press}.
\newblock


\bibitem[Sun et~al\mbox{.}(2021)]%
        {sun2021treatment}
\bibfield{author}{\bibinfo{person}{Hao Sun}, \bibinfo{person}{Evan Munro}, \bibinfo{person}{Georgy Kalashnov}, \bibinfo{person}{Shuyang Du}, {and} \bibinfo{person}{Stefan Wager}.} \bibinfo{year}{2021}\natexlab{}.
\newblock \showarticletitle{Treatment allocation under uncertain costs}.
\newblock \bibinfo{journal}{\emph{arXiv preprint arXiv:2103.11066}} (\bibinfo{year}{2021}).
\newblock


\bibitem[Surve and Pradhan(2024)]%
        {DBLP:journals/corr/abs-2402-05007}
\bibfield{author}{\bibinfo{person}{Tanmay Surve} {and} \bibinfo{person}{Romila Pradhan}.} \bibinfo{year}{2024}\natexlab{}.
\newblock \showarticletitle{Example-based Explanations for Random Forests using Machine Unlearning}.
\newblock \bibinfo{journal}{\emph{CoRR}}  \bibinfo{volume}{abs/2402.05007} (\bibinfo{year}{2024}).
\newblock


\bibitem[Tao et~al\mbox{.}(2022)]%
        {tao2022dpxplain}
\bibfield{author}{\bibinfo{person}{Yuchao Tao}, \bibinfo{person}{Amir Gilad}, \bibinfo{person}{Ashwin Machanavajjhala}, {and} \bibinfo{person}{Sudeepa Roy}.} \bibinfo{year}{2022}\natexlab{}.
\newblock \showarticletitle{DPXPlain: Privately Explaining Aggregate Query Answers}.
\newblock \bibinfo{journal}{\emph{Proc. {VLDB} Endow.}} \bibinfo{volume}{16}, \bibinfo{number}{1} (\bibinfo{year}{2022}), \bibinfo{pages}{113--126}.
\newblock
\urldef\tempurl%
\url{https://www.vldb.org/pvldb/vol16/p113-tao.pdf}
\showURL{%
\tempurl}


\bibitem[ten Cate et~al\mbox{.}(2015)]%
        {ten2015high}
\bibfield{author}{\bibinfo{person}{Balder ten Cate}, \bibinfo{person}{Cristina Civili}, \bibinfo{person}{Evgeny Sherkhonov}, {and} \bibinfo{person}{Wang-Chiew Tan}.} \bibinfo{year}{2015}\natexlab{}.
\newblock \showarticletitle{High-level why-not explanations using ontologies}. In \bibinfo{booktitle}{\emph{Proceedings of the 34th ACM SIGMOD-SIGACT-SIGAI Symposium on Principles of Database Systems}}. \bibinfo{pages}{31--43}.
\newblock


\bibitem[{U.S. Census Bureau}(2024)]%
        {ACS_Data}
\bibfield{author}{\bibinfo{person}{{U.S. Census Bureau}}.} \bibinfo{year}{2024}\natexlab{}.
\newblock \bibinfo{title}{American Community Survey (ACS) - Data}.
\newblock
\newblock
\urldef\tempurl%
\url{https://www.census.gov/programs-surveys/acs/data.html}
\showURL{%
\tempurl}
\newblock
\shownote{Accessed: 2024-01-30}.


\bibitem[Vardi(1982)]%
        {Vardi82}
\bibfield{author}{\bibinfo{person}{Moshe~Y. Vardi}.} \bibinfo{year}{1982}\natexlab{}.
\newblock \showarticletitle{The Complexity of Relational Query Languages (Extended Abstract)}. In \bibinfo{booktitle}{\emph{Proceedings of the Fourteenth Annual ACM Symposium on Theory of Computing}} (San Francisco, California, USA) \emph{(\bibinfo{series}{STOC '82})}. \bibinfo{publisher}{ACM}, \bibinfo{address}{New York, NY, USA}, \bibinfo{pages}{137--146}.
\newblock
\showISBNx{0-89791-070-2}
\urldef\tempurl%
\url{https://doi.org/10.1145/800070.802186}
\showDOI{\tempurl}


\bibitem[Vartak et~al\mbox{.}(2015)]%
        {vartak2015seedb}
\bibfield{author}{\bibinfo{person}{Manasi Vartak}, \bibinfo{person}{Sajjadur Rahman}, \bibinfo{person}{Samuel Madden}, \bibinfo{person}{Aditya Parameswaran}, {and} \bibinfo{person}{Neoklis Polyzotis}.} \bibinfo{year}{2015}\natexlab{}.
\newblock \showarticletitle{Seedb: Efficient data-driven visualization recommendations to support visual analytics}. In \bibinfo{booktitle}{\emph{VLDB}}, Vol.~\bibinfo{volume}{8}. NIH Public Access, \bibinfo{pages}{2182}.
\newblock


\bibitem[Virtanen et~al\mbox{.}(2020)]%
        {virtanen2020scipy}
\bibfield{author}{\bibinfo{person}{Pauli Virtanen}, \bibinfo{person}{Ralf Gommers}, \bibinfo{person}{Travis~E Oliphant}, \bibinfo{person}{Matt Haberland}, \bibinfo{person}{Tyler Reddy}, \bibinfo{person}{David Cournapeau}, \bibinfo{person}{Evgeni Burovski}, \bibinfo{person}{Pearu Peterson}, \bibinfo{person}{Warren Weckesser}, \bibinfo{person}{Jonathan Bright}, {et~al\mbox{.}}} \bibinfo{year}{2020}\natexlab{}.
\newblock \showarticletitle{SciPy 1.0: fundamental algorithms for scientific computing in Python}.
\newblock \bibinfo{journal}{\emph{Nature methods}} \bibinfo{volume}{17}, \bibinfo{number}{3} (\bibinfo{year}{2020}), \bibinfo{pages}{261--272}.
\newblock


\bibitem[Wagner(1982)]%
        {wagner1982simpson}
\bibfield{author}{\bibinfo{person}{Clifford~H Wagner}.} \bibinfo{year}{1982}\natexlab{}.
\newblock \showarticletitle{Simpson's paradox in real life}.
\newblock \bibinfo{journal}{\emph{The American Statistician}} \bibinfo{volume}{36}, \bibinfo{number}{1} (\bibinfo{year}{1982}), \bibinfo{pages}{46--48}.
\newblock


\bibitem[Wang and Rudin(2022)]%
        {wang2022causal}
\bibfield{author}{\bibinfo{person}{Tong Wang} {and} \bibinfo{person}{Cynthia Rudin}.} \bibinfo{year}{2022}\natexlab{}.
\newblock \showarticletitle{Causal rule sets for identifying subgroups with enhanced treatment effects}.
\newblock \bibinfo{journal}{\emph{INFORMS Journal on Computing}} \bibinfo{volume}{34}, \bibinfo{number}{3} (\bibinfo{year}{2022}), \bibinfo{pages}{1626--1643}.
\newblock


\bibitem[Wang et~al\mbox{.}(2018)]%
        {DBLP:journals/pvldb/WangMM18}
\bibfield{author}{\bibinfo{person}{Yue Wang}, \bibinfo{person}{Alexandra Meliou}, {and} \bibinfo{person}{Gerome Miklau}.} \bibinfo{year}{2018}\natexlab{}.
\newblock \showarticletitle{RC-Index: Diversifying Answers to Range Queries}.
\newblock \bibinfo{journal}{\emph{Proc. {VLDB} Endow.}} \bibinfo{volume}{11}, \bibinfo{number}{7} (\bibinfo{year}{2018}), \bibinfo{pages}{773--786}.
\newblock
\urldef\tempurl%
\url{https://doi.org/10.14778/3192965.3192969}
\showDOI{\tempurl}


\bibitem[Wen et~al\mbox{.}(2018)]%
        {wen2018interactive}
\bibfield{author}{\bibinfo{person}{Yuhao Wen}, \bibinfo{person}{Xiaodan Zhu}, \bibinfo{person}{Sudeepa Roy}, {and} \bibinfo{person}{Jun Yang}.} \bibinfo{year}{2018}\natexlab{}.
\newblock \showarticletitle{Interactive summarization and exploration of top aggregate query answers}. In \bibinfo{booktitle}{\emph{Proceedings of the VLDB Endowment. International Conference on Very Large Data Bases}}, Vol.~\bibinfo{volume}{11}. NIH Public Access, \bibinfo{pages}{2196}.
\newblock


\bibitem[Wu and Madden(2013)]%
        {wu2013scorpion}
\bibfield{author}{\bibinfo{person}{Eugene Wu} {and} \bibinfo{person}{Samuel Madden}.} \bibinfo{year}{2013}\natexlab{}.
\newblock \showarticletitle{Scorpion: Explaining away outliers in aggregate queries}.
\newblock  (\bibinfo{year}{2013}).
\newblock


\bibitem[Xie et~al\mbox{.}(2012)]%
        {xie2012estimating}
\bibfield{author}{\bibinfo{person}{Yu Xie}, \bibinfo{person}{Jennie~E Brand}, {and} \bibinfo{person}{Ben Jann}.} \bibinfo{year}{2012}\natexlab{}.
\newblock \showarticletitle{Estimating heterogeneous treatment effects with observational data}.
\newblock \bibinfo{journal}{\emph{Sociological methodology}} \bibinfo{volume}{42}, \bibinfo{number}{1} (\bibinfo{year}{2012}), \bibinfo{pages}{314--347}.
\newblock


\bibitem[Youngmann et~al\mbox{.}(2022)]%
        {DBLP:journals/pvldb/YoungmannAP22}
\bibfield{author}{\bibinfo{person}{Brit Youngmann}, \bibinfo{person}{Sihem Amer{-}Yahia}, {and} \bibinfo{person}{Aur{\'{e}}lien Personnaz}.} \bibinfo{year}{2022}\natexlab{}.
\newblock \showarticletitle{Guided Exploration of Data Summaries}.
\newblock \bibinfo{journal}{\emph{Proc. {VLDB} Endow.}} \bibinfo{volume}{15}, \bibinfo{number}{9} (\bibinfo{year}{2022}).
\newblock


\bibitem[Youngmann et~al\mbox{.}(2024)]%
        {youngmann2024summarized}
\bibfield{author}{\bibinfo{person}{Brit Youngmann}, \bibinfo{person}{Michael Cafarella}, \bibinfo{person}{Amir Gilad}, {and} \bibinfo{person}{Sudeepa Roy}.} \bibinfo{year}{2024}\natexlab{}.
\newblock \showarticletitle{Summarized Causal Explanations For Aggregate Views}.
\newblock \bibinfo{journal}{\emph{Proceedings of the ACM on Management of Data}} \bibinfo{volume}{2}, \bibinfo{number}{1} (\bibinfo{year}{2024}), \bibinfo{pages}{1--27}.
\newblock


\bibitem[Youngmann et~al\mbox{.}(2023a)]%
        {youngmann2022explaining}
\bibfield{author}{\bibinfo{person}{Brit Youngmann}, \bibinfo{person}{Michael Cafarella}, \bibinfo{person}{Yuval Moskovitch}, {and} \bibinfo{person}{Babak Salimi}.} \bibinfo{year}{2023}\natexlab{a}.
\newblock \showarticletitle{On Explaining Confounding Bias}.
\newblock \bibinfo{journal}{\emph{2023 IEEE 39th International Conference on Data Engineering (ICDE)}} (\bibinfo{year}{2023}).
\newblock


\bibitem[Youngmann et~al\mbox{.}(2023b)]%
        {youngmann2023causal}
\bibfield{author}{\bibinfo{person}{Brit Youngmann}, \bibinfo{person}{Michael Cafarella}, \bibinfo{person}{Babak Salimi}, {and} \bibinfo{person}{Anna Zeng}.} \bibinfo{year}{2023}\natexlab{b}.
\newblock \showarticletitle{Causal Data Integration}.
\newblock \bibinfo{journal}{\emph{Proceedings of the VLDB Endowment}} \bibinfo{volume}{16}, \bibinfo{number}{10} (\bibinfo{year}{2023}), \bibinfo{pages}{2659--2665}.
\newblock


\bibitem[Yu et~al\mbox{.}(2009)]%
        {yu2009takes}
\bibfield{author}{\bibinfo{person}{Cong Yu}, \bibinfo{person}{Laks Lakshmanan}, {and} \bibinfo{person}{Sihem Amer-Yahia}.} \bibinfo{year}{2009}\natexlab{}.
\newblock \showarticletitle{It takes variety to make a world: diversification in recommender systems}. In \bibinfo{booktitle}{\emph{Proceedings of the 12th international conference on extending database technology: Advances in database technology}}. \bibinfo{pages}{368--378}.
\newblock


\bibitem[Zeng et~al\mbox{.}(2025)]%
        {zeng2025causal}
\bibfield{author}{\bibinfo{person}{Anna Zeng}, \bibinfo{person}{Michael Cafarella}, \bibinfo{person}{Batya Kenig}, \bibinfo{person}{Markos Markakis}, \bibinfo{person}{Brit Youngmann}, {and} \bibinfo{person}{Babak Salimi}.} \bibinfo{year}{2025}\natexlab{}.
\newblock \showarticletitle{Causal DAG Summarization}.
\newblock \bibinfo{journal}{\emph{Proceedings of the VLDB Endowment}} \bibinfo{volume}{18}, \bibinfo{number}{6} (\bibinfo{year}{2025}), \bibinfo{pages}{1933--1947}.
\newblock


\bibitem[Zhang et~al\mbox{.}(2021)]%
        {zhang2021viewseeker}
\bibfield{author}{\bibinfo{person}{Xiaozhong Zhang}, \bibinfo{person}{Xiaoyu Ge}, \bibinfo{person}{Panos~K Chrysanthis}, {and} \bibinfo{person}{Mohamed~A Sharaf}.} \bibinfo{year}{2021}\natexlab{}.
\newblock \showarticletitle{Viewseeker: An interactive view recommendation framework}.
\newblock \bibinfo{journal}{\emph{Big Data Research}}  \bibinfo{volume}{25} (\bibinfo{year}{2021}), \bibinfo{pages}{100238}.
\newblock


\bibitem[Zhou et~al\mbox{.}(2024)]%
        {10.1145/3637528.3671951}
\bibfield{author}{\bibinfo{person}{Jiehui Zhou}, \bibinfo{person}{Linxiao Yang}, \bibinfo{person}{Xingyu Liu}, \bibinfo{person}{Xinyue Gu}, \bibinfo{person}{Liang Sun}, {and} \bibinfo{person}{Wei Chen}.} \bibinfo{year}{2024}\natexlab{}.
\newblock \showarticletitle{CURLS: Causal Rule Learning for Subgroups with Significant Treatment Effect}. In \bibinfo{booktitle}{\emph{Proceedings of the 30th ACM SIGKDD Conference on Knowledge Discovery and Data Mining}} (Barcelona, Spain) \emph{(\bibinfo{series}{KDD '24})}. \bibinfo{publisher}{Association for Computing Machinery}, \bibinfo{address}{New York, NY, USA}, \bibinfo{pages}{4619–4630}.
\newblock
\showISBNx{9798400704901}
\urldef\tempurl%
\url{https://doi.org/10.1145/3637528.3671951}
\showDOI{\tempurl}


\end{thebibliography}

\begin{appendices}
%!TEX root=main.tex

\newpage

% \section{Proofs}
\section{NP-hardness Proof}\label{sec:proof}

In this part, we give the proof for showing the hardness of the decision problem defined in Section \ref{sec:framework} (proposition~\ref{prop:hard}).
% In this part, we give the missing proofs for the propositions presented in Section \ref{sec:framework}.

\begin{figure}[h]
    \centering
    \begin{subfigure}[b]{0.45\textwidth}
        \centering
        \begin{tikzpicture}
          % Define the vertices with labels
          \foreach \i/\label/\x/\y in {1/v_1/0/0, 2/v_2/1/0, 3/v_3/2/0, 4/v_4/0/-1, 5/v_5/1/-1} {
            \coordinate (V\i) at (\x,\y);
            \node[draw,circle,inner sep=1pt] (N\i) at (V\i) {$\label$};
          }

          % Draw the edges with annotations
          \foreach \i/\j/\label/\pos in {
            1/2/e_1/above, 
            2/3/e_2/above, 
            1/4/e_3/left, 
            2/5/e_4/right, 
            4/5/e_5/above} {
            \draw (N\i) -- (N\j) node[midway,\pos]{$\label$};
          }
        \end{tikzpicture}
        \caption{Input instance to the IS problem}
        \label{fig:graph}
    \end{subfigure}
    \hfill
    \begin{subfigure}[b]{0.45\textwidth}
        \centering
        \begin{tabular}{lccccc|cccc}
             ~&$v_1$& $v_2$&$v_3$&$v_4$&$v_5$&$g_1$&$g_2$&$T$&$O$\\ \cline{2-10}
             $e_1$&$ 1 $& $1$&$0$&$0$&$0$&$\star$&$\star$&$\star$&$\star$\\
             $e_2$&$ 0 $& $1$&$1$&$0$&$0$&$\star$&$\star$&$\star$&$\star$\\
             $e_3$&$ 1 $& $0$&$0$&$1$&$0$&$\star$&$\star$&$\star$&$\star$\\
             $e_4$&$ 0 $& $1$&$0$&$0$&$1$&$\star$&$\star$&$\star$&$\star$\\
             $e_5$&$ 0 $& $0$&$0$&$1$&$1$&$\star$&$\star$&$\star$&$\star$\\
        \end{tabular}
        \caption{The database $D$ of the reduction output}
        \label{fig:table}
    \end{subfigure}
    \caption{Reduction Example}
    \label{fig:example}
\end{figure}

\begin{proof}[Proof of Proposition \ref{prop:hard}]
We show a polynomial reduction from the Independent Set (IS) problem. Recall that  $\langle G= (V,E), k\rangle \in IS$ if there exists $I \subseteq V$ s.t. $|I| = k$ and $\forall v_i, v_j \in I$ there is no edge between $v_i$ and $v_j$ in $G$. Given an input instance $\langle G= (V, E), k\rangle$ to the IS problem, we create an input instance to our decision problem as follows. $D$ is a dataset with $|V|+4$ attributes and $|E|$ tuples, where there is an attribute $v_i$ for each $v_i\in V$, the attributes $g_1$, $g_2$ and an attribute $O$. There is a tuple $e_i$ for each $e_i\in E$ where $e_i[v_j] = 1 \iff v_j\in e_i$. The values in the $g_1$, $g_2$, $T$ and $O$ attributes can be assigned arbitrarily, as they do not affect the reduction correctness. Figure~\ref{fig:example} shows a simple example of the resulting database $D$. 
We set $\sigma$ to be $1$, $\tau$ to be $0$, $B$ to be $0$, and use the same $k$ value. 
%We set $\alpha$ to be $0$, $B$ to be $1$ and use the same $k$ value. 
We define the set of possible \facts\ to be $\{(v_i = 1, \star)\mid v_i\in V\}$. I.e. the set of possible \facts\ consist of $|V|$ \fact, for each one, the grouping pattern $\pattern_g$ is defined as $v_i = 1$ whereas the treatment pattern $\pattern_e$ can be assigned arbitrarily. Note that for $\phi_i = (v_i = 1, \star)$ and $\phi_j = (v_j = 1, \star)$ we have that $\expSim(\phi_i,\phi_j) =0$ 
% $\phi_i \cap \phi_i = 0$ 
if and only if there is no edge between $v_i$ and $v_j$.
Finally, the causal model \model\ can be defined arbitrarily. Clearly, this reduction is polynomial. We next show that $\langle G= (V,E), k\rangle \in IS \iff$ there exists a subset $\Phi\subseteq \{(v_i = 1, \star)\mid v_i\in V\}$, such that $|\Phi| \leq k$,  $\forall \phi_i\in \Phi,~support(\phi_i) \geq \sigma = 1$, $\forall \phi_i,\phi_j\in \Phi,~\expSim(\phi_i,\phi_j) \leq \tau = 0$ and $\sum_{\phi\in \Phi}\Delta(\phi) \geq B = 0$.
% and $\alpha \cdot \$\Delta$(\Phi) + (1-\alpha) \cdot NOverlap(\Phi) \geq B$.

First, assume that $\langle G= (V,E), k\rangle \in IS$. Thus, there is a set $I\subseteq V$ such that $|I| = k$ and $\forall v_i, v_j\in I$ there is no edge between $v_i$ and $v_j$. We define $\Phi = \{ (v_i=1, \star) \mid v_i\in I \}$. Clearly, $\forall \phi_i\in \Phi,~support(\phi_i) \geq 1$. Note that since $I$ is an IS, $\forall \phi_i,\phi_j\in \Phi,~\expSim(\phi_i,\phi_j) = 0$, and since $\disp(\phi)\geq 0$ for every $\phi$ we have that $\sum_{\phi\in\Phi}\disp(\phi)\geq 0 = B$.

% ~overlap(\phi_i,\phi_j) = 0$, thus, $Overlap(\Phi) = 1$ and since $\alpha=0$ we get that $|\Phi| \leq k$ and $\alpha \cdot \$\Delta$(\Phi) + (1-\alpha) \cdot NOverlap(\Phi) = 1 \geq B$

Now assume there exists a subset $\Phi\subseteq \{(v_i = 1, \star)\mid v_i\in V\}$, such that $|\Phi| \leq k$,  $\forall \phi_i\in \Phi,~support(\phi_i) \geq \sigma = 1$, $\forall \phi_i,\phi_j\in \Phi,~\expSim(\phi_i,\phi_j) = 0$ and $\sum_{\phi\in\Phi}\disp(\phi)\geq B = 0$. From the construction $\expSim(\phi_i,\phi_j) = 0$ where $\phi_i = (v_i = 1, \star)$ and $\phi_j = (v_j = 1, \star)$ only when here is no edge between $v_i$ and $v_j$. Therefore, the set $I = \{v_i\mid (v_i = 1,\star)\in \Phi\}$ is an IS of size $k$ in $G$.

% and $\alpha \cdot \$\Delta$(\Phi) + (1-\alpha) \cdot NOverlap(\Phi) \geq B$. Since $\alpha = 0$, we get that $NOverlap(\Phi) \geq B = 1$. Namely, $\forall \phi_i,\phi_j\in \Phi~overlap(\phi_i,\phi_j) = 0$ (otherwise, $NOverlap(\Phi) <1$). From the construction $overlap(\phi_i,\phi_j) = 0$ where $\phi_i = (v_i = 1, \star)$ and $\phi_j = (v_j = 1, \star)$ only when here is no edge between $v_i$ and $v_j$. Therefore, the set $I = \{v_i\mid (v_i = 1,\star)\in \Phi\}$ is an IS of size $k$ in $G$.

% where each $\phi_i = (T_i, O, \pattern_i, \pattern^p)$, such that $|\Phi| \leq k$ and $\alpha \cdot Utility(\Phi) + (1-\alpha) \cdot Overlap(\Phi) \geq Obj$. Since $\alpha = 0$ we get that $Overlap(\Phi) \geq Obj = 1$. Particularly, $overlap(\phi_i,\phi_j)$ mist be $0$ $\forall \phi_i,\phi_j\in \Phi$. We claim that \yuval{not sure if this is correct!!}

\end{proof}

\section{Results from the Baselines}
\label{app:exp}

DivExplorer and FairDebugger only produced explanations for the ACS dataset, which are shown in Table~\ref{tab:example_acs_de} and  Table~\ref{tab:example_acs_rf}, respectively.

\paragraph{Top-k}
\begin{itemize}
    \item The results for the SO dataset are in Table \ref{tab:example_so_topk}.

    \item The results for the ACS dataset are in Table \ref{tab:example_acs_topk}.

    \item The results for the MEPS dataset are in Table \ref{tab:example_meps_topk}.
\end{itemize}

\paragraph{Brute-Force}
\begin{itemize}
    \item The results for the SO dataset are in Table \ref{tab:example_so_bf}.

    \item The results for the ACS dataset are in Table \ref{tab:example_acs_bf}.

    \item The results for the MEPS dataset are in Table \ref{tab:example_meps_bf}.
\end{itemize}

\begin{table*}[h!]
\centering
\resizebox{1\textwidth}{!}{
\renewcommand{\arraystretch}{1.1}
\begin{tabular}[t]{|p{80mm}|c|>{\centering}m{25mm}c|@{}>{\centering}m{25mm}c|c|}
\hline
\multirow{3}{*}{\textbf{Disparity Explanation}} & 
\multirow{3}{*}{\textbf{Support}} & 
\multicolumn{4}{c|}{\textbf{Likelihood of having a health insurance}} &
\multirow{3}{*}{\textbf{$\Delta$}} \\
\cline{3-6}
& &
\multicolumn{2}{c|}{\textbf{Subpopulation}} &
\multicolumn{2}{c|}{\textbf{Global (XInsight)}}& \\
\cline{3-6}
&& 
\textbf{Average} & \textbf{CATE} & 
\textbf{Average} & \textbf{CATE} &\\
\hline
\hline
% Row 1
\begin{tabular}{@{}p{80mm}}
For \colorbox{Orange!30}{\strut{White individuals from the Southern region who speak}} \colorbox{Orange!30}{\strut{Spanish}}, the likelihood of having health insurance decreases when they \colorbox{Cyan!30}{\strut{have no personal earnings}} for \colorbox{Yellow!70}{\strut{manual labor occupations}} whereas it increases for \colorbox{Lavender!70}{\strut{all occupations}}.
\end{tabular} &
8.18\% &
\multicolumn{2}{l|}{
\begin{tabular}{p{25mm}r}
\drawbar{52.12}{7.74}{\%}{Yellow}{\downarrowbold}\\
\drawbar{61.03}{4.80}{\%}{Lavender}{\uparrowbold}\\
\end{tabular}
}
&
\multicolumn{2}{l|}{
\begin{tabular}{p{25mm}r}
\drawbar{78.61}{2.65}{\%}{Yellow}{\downarrowbold}\\
\drawbar{91.58}{1.98}{\%}{Lavender}{\uparrowbold}\\
\end{tabular}
} & 0.126 \\
\hline
% Row 2
\begin{tabular}{@{}p{80mm}}
Among \colorbox{Orange!30}{\strut{White natives from Texas who born in USA}}, the likelihood of having health insurance decreases when they \colorbox{Cyan!30}{\strut{have no personal earnings}} for \colorbox{Yellow!70}{\strut{manual labor occupations}} whereas it increases for \colorbox{Lavender!70}{\strut{all occupations}}.
\end{tabular} &
12.21\% &
\multicolumn{2}{l|}{
\begin{tabular}{p{25mm}r}
\drawbar{71.86}{5.20}{\%}{Yellow}{\downarrowbold}\\
\drawbar{80.52}{2.83}{\%}{Lavender}{\uparrowbold}\\
\end{tabular}
}
&
\multicolumn{2}{l|}{
\begin{tabular}{p{25mm}r}
\drawbar{78.61}{2.65}{\%}{Yellow}{\downarrowbold}\\
\drawbar{91.58}{1.98}{\%}{Lavender}{\uparrowbold}\\
\end{tabular}
} & 0.08 \\
\hline
% Row 3
\begin{tabular}{@{}p{80mm}}
For \colorbox{Orange!30}{\strut{individuals from the Southern region}}, the likelihood of having health insurance decreases when they \colorbox{Cyan!30}{\strut{have no personal earnings}} for \colorbox{Yellow!70}{\strut{manual labor occupations}} whereas it increases for \colorbox{Lavender!70}{\strut{all occupations}}.
\end{tabular}
&
34.1\%&
\multicolumn{2}{l|}{
\begin{tabular}{p{25mm}r}
\drawbar{65.80}{5.25}{\%}{Yellow}{\downarrowbold}\\
\drawbar{75.78}{2.45}{\%}{Lavender}{\uparrowbold}\\
\end{tabular}
}
&
\multicolumn{2}{l|}{
\begin{tabular}{p{25mm}r}
\drawbar{78.61}{2.65}{\%}{Yellow}{\downarrowbold}\\
\drawbar{91.58}{1.98}{\%}{Lavender}{\uparrowbold}\\
\end{tabular}
}
& 0.077\\
\hline
% Row 4
\begin{tabular}{@{}p{80mm}}
For \colorbox{Orange!30}{\strut{White males from the Southern region}}, the likelihood of having health insurance decreases when they \colorbox{Cyan!30}{\strut{have no personal earnings}} for \colorbox{Yellow!70}{\strut{manual labor occupations}} whereas it increases for \colorbox{Lavender!70}{\strut{all occupations}}.
\end{tabular} &
18.00\% &
\multicolumn{2}{l|}{
\begin{tabular}{p{25mm}r}
\drawbar{67.56}{4.22}{\%}{Yellow}{\downarrowbold}\\
\drawbar{74.20}{2.68}{\%}{Lavender}{\uparrowbold}\\
\end{tabular}
}
&
\multicolumn{2}{l|}{
\begin{tabular}{p{25mm}r}
\drawbar{78.61}{2.65}{\%}{Yellow}{\downarrowbold}\\
\drawbar{91.58}{1.98}{\%}{Lavender}{\uparrowbold}\\
\end{tabular}
} & 0.069 \\
\hline
% Row 5
\begin{tabular}{@{}p{80mm}}
For \colorbox{Orange!30}{\strut{White natives from the South region without disabilities}}, the likelihood of having health insurance decreases when they \colorbox{Cyan!30}{\strut{have no personal earnings}} for \colorbox{Yellow!70}{\strut{manual labor occupations}} whereas it increases for \colorbox{Lavender!70}{\strut{all occupations}}.
\end{tabular} &
18.17\% &
\multicolumn{2}{l|}{
\begin{tabular}{p{25mm}r}
\drawbar{73.05}{4.17}{\%}{Yellow}{\downarrowbold}\\
\drawbar{81.50}{2.11}{\%}{Lavender}{\uparrowbold}\\
\end{tabular}
}
&
\multicolumn{2}{l|}{
\begin{tabular}{p{25mm}r}
\drawbar{78.61}{2.65}{\%}{Yellow}{\downarrowbold}\\
\drawbar{91.58}{1.98}{\%}{Lavender}{\uparrowbold}\\
\end{tabular}
} & 0.063 \\
\hline
\end{tabular}}
\caption{DivExplorer results for the ACS dataset.}
%\vspace{-1mm}
\label{tab:example_acs_de}
\end{table*}

\begin{table*}[t]
\centering
\resizebox{1\textwidth}{!}{
\renewcommand{\arraystretch}{1.1}
\begin{tabular}[t]{|p{80mm}|c|>{\centering}m{25mm}c|@{}>{\centering}m{25mm}c|c|}
\hline
\multirow{3}{*}{\textbf{Disparity Explanation}} & 
\multirow{3}{*}{\textbf{Support}} & 
\multicolumn{4}{c|}{\textbf{Likelihood of having a health insurance}} &
\multirow{3}{*}{\textbf{$\Delta$}} \\
\cline{3-6}
& &
\multicolumn{2}{c|}{\textbf{Subpopulation}} &
\multicolumn{2}{c|}{\textbf{Global (XInsight)}}& \\
\cline{3-6}
&& 
\textbf{Average} & \textbf{CATE} & 
\textbf{Average} & \textbf{CATE} &\\
\hline
\hline
% Row 1
\begin{tabular}{@{}p{80mm}}
For \colorbox{Orange!30}{\strut{individuals from the Southern region}}, the likelihood of having health insurance decreases when they \colorbox{Cyan!30}{\strut{have no personal earnings}} for \colorbox{Yellow!70}{\strut{manual labor occupations}} whereas it increases for \colorbox{Lavender!70}{\strut{all occupations}}.
\end{tabular} &
34.10\% &
\multicolumn{2}{l|}{
\begin{tabular}{p{25mm}r}
\drawbar{65.85}{5.25}{\%}{Yellow}{\downarrowbold}\\
\drawbar{75.78}{2.43}{\%}{Lavender}{\uparrowbold}\\
\end{tabular}
}
&
\multicolumn{2}{l|}{
\begin{tabular}{p{25mm}r}
\drawbar{78.61}{2.65}{\%}{Yellow}{\downarrowbold}\\
\drawbar{91.58}{1.98}{\%}{Lavender}{\uparrowbold}\\
\end{tabular}
} & 0.076 \\
\hline
% Row 2
\begin{tabular}{@{}p{80mm}}
Among \colorbox{Orange!30}{\strut{individuals who born in USA}}, the likelihood of having health insurance decreases when they \colorbox{Cyan!30}{\strut{have no personal earnings}} for \colorbox{Yellow!70}{\strut{manual labor occupations}} whereas it increases for \colorbox{Lavender!70}{\strut{all occupations}}.
\end{tabular} &
75.67\% &
\multicolumn{2}{l|}{
\begin{tabular}{p{25mm}r}
\drawbar{84.11}{3.46}{\%}{Yellow}{\downarrowbold}\\
\drawbar{88.97}{1.36}{\%}{Lavender}{\uparrowbold}\\
\end{tabular}
}
&
\multicolumn{2}{l|}{
\begin{tabular}{p{25mm}r}
\drawbar{78.61}{2.65}{\%}{Yellow}{\downarrowbold}\\
\drawbar{91.58}{1.98}{\%}{Lavender}{\uparrowbold}\\
\end{tabular}
} & 0.048 \\
\hline
\end{tabular}}
\caption{FairDebugger results for the ACS dataset.}
\label{tab:example_acs_rf}
\end{table*}

\begin{table*}[t]
\centering
\resizebox{1\textwidth}{!}{
\renewcommand{\arraystretch}{1.1}
\begin{tabular}[t]{|p{80mm}|c|>{\centering}m{30mm}c|@{}>{\centering}m{30mm}c|c|}
\hline
\multirow{3}{*}{\textbf{Disparity Explanation}} & 
\multirow{3}{*}{\textbf{Support}} & 
\multicolumn{4}{c|}{\textbf{Total Compensation (TC)}} &
\multirow{3}{*}{\textbf{$\Delta$}} \\
\cline{3-6}
& &
\multicolumn{2}{c|}{\textbf{Subpopulation}} &
\multicolumn{2}{c|}{\textbf{Global (XInsight)}}& \\
\cline{3-6}
&& 
\textbf{Average} & \textbf{CATE} & 
\textbf{Average} & \textbf{CATE} &\\
\hline
\hline
% Row 1
\begin{tabular}{@{}p{80mm}}
For \colorbox{Orange!30}{\strut{White males aged between 18-24}}, TC growth is more influenced by \colorbox{Cyan!30}{\strut{having coding as a hobby and learned in undergrad major}} \colorbox{Cyan!30}{\strut{computer science}} for \colorbox{Yellow!70}{\strut{analysts}} compared to \colorbox{Lavender!70}{\strut{back-end developers}}.
\end{tabular}
&
10.64\%&
\multicolumn{2}{l|}{
\begin{tabular}{p{30mm}r}
\phantom{3}\$90,909 \hspace{2mm}
\resizebox{.08\textwidth}{!}{
\raisebox{-0.3\height}{
\begin{tikzpicture}
\draw[fill=Yellow!70] (0,-25) rectangle (50, -45);
\draw[fill=Yellow!10] (50, -25) rectangle (100, -45);
\end{tikzpicture}}}
& \$219,469 \uparrowbold
\\
\phantom{3}\$67,754 \hspace{2mm}
\resizebox{.08\textwidth}{!}{
\raisebox{-0.3\height}{
\begin{tikzpicture}
\draw[fill=Lavender!70] (0,-25) rectangle (37, -45);
\draw[fill=Lavender!10] (37, -25) rectangle (100, -45);
\end{tikzpicture}}}
& \$55,233 \uparrowbold\\
\end{tabular}
}
&
\multicolumn{2}{c|}{not statistically significant}
& 0.082\\
\hline
% Row 2
\begin{tabular}{@{}p{80mm}}
For \colorbox{Orange!30}{\strut{heterosexual individuals whose parents attended}} \colorbox{Orange!30}{\strut{secondary school}}, TC growth is more influenced by \colorbox{Cyan!30}{\strut{the desire become manager and not being a student}} for \colorbox{Yellow!70}{\strut{analysts}} compared to \colorbox{Lavender!70}{\strut{back-end developers}}.
\end{tabular}
&
14.41\%&
\multicolumn{2}{l|}{
\begin{tabular}{p{30mm}r}
\$112,896 \hspace{2mm}
\resizebox{.08\textwidth}{!}{
\raisebox{-0.3\height}{
\begin{tikzpicture}
\draw[fill=Yellow!70] (0,-25) rectangle (64, -45);
\draw[fill=Yellow!10] (64, -25) rectangle (100, -45);
\end{tikzpicture}}}
& \$176,698 \uparrowbold
\\
\phantom{3}\$99,126 \hspace{2mm}
\resizebox{.08\textwidth}{!}{
\raisebox{-0.3\height}{
\begin{tikzpicture}
\draw[fill=Lavender!70] (0,-25) rectangle (58, -45);
\draw[fill=Lavender!10] (58, -25) rectangle (100, -45);
\end{tikzpicture}}}
& \$53,814 \uparrowbold\\
\end{tabular}
}
&
\multicolumn{2}{c|}{not statistically significant}
& 0.061\\
\hline
% Row 3
\begin{tabular}{@{}p{80mm}}
For \colorbox{Orange!30}{\strut{White heterosexual individuals whose parents attended}} \colorbox{Orange!30}{\strut{secondary school}}, TC growth is more influenced by \colorbox{Cyan!30}{\strut{the desire become manager}} for \colorbox{Yellow!70}{\strut{analysts}} compared to \colorbox{Lavender!70}{\strut{back-end developers}}.
\end{tabular}
&
11.14\%&
\multicolumn{2}{l|}{
\begin{tabular}{p{30mm}r}
\$129,749 \hspace{2mm}
\resizebox{.08\textwidth}{!}{
\raisebox{-0.3\height}{
\begin{tikzpicture}
\draw[fill=Yellow!70] (0,-25) rectangle (71, -45);
\draw[fill=Yellow!10] (71, -25) rectangle (100, -45);
\end{tikzpicture}}}
& \$154,024 \uparrowbold
\\
\$110,026 \hspace{2mm}
\resizebox{.08\textwidth}{!}{
\raisebox{-0.3\height}{
\begin{tikzpicture}
\draw[fill=Lavender!70] (0,-25) rectangle (61, -45);
\draw[fill=Lavender!10] (61, -25) rectangle (100, -45);
\end{tikzpicture}}}
& \$31,354 \uparrowbold\\
\end{tabular}
}
&
\multicolumn{2}{c|}{not statistically significant}
& 0.061\\
\hline
% Row 4
\begin{tabular}{@{}p{80mm}}
For \colorbox{Orange!30}{\strut{White individuals aged between 25-34}}, TC  increases by \colorbox{Cyan!30}{\strut{having 3-5 years of professional coding experience and}} \colorbox{Cyan!30}{\strut{spending over 12 hours on computer daily}} for \colorbox{Yellow!70}{\strut{analysts}} whereas it decreases for \colorbox{Lavender!70}{\strut{back-end developers}}.
\end{tabular}
&
34.69\%&
\multicolumn{2}{l|}{
\begin{tabular}{p{30mm}r}
\$115,777 \hspace{2mm}
\resizebox{.08\textwidth}{!}{
\raisebox{-0.3\height}{
\begin{tikzpicture}
\draw[fill=Yellow!70] (0,-25) rectangle (68, -45);
\draw[fill=Yellow!10] (68, -25) rectangle (100, -45);
\end{tikzpicture}}}
& \phantom{3}\$93,826 \uparrowbold
\\
\$105,988 \hspace{2mm}
\resizebox{.08\textwidth}{!}{
\raisebox{-0.3\height}{
\begin{tikzpicture}
\draw[fill=Lavender!70] (0,-25) rectangle (60, -45);
\draw[fill=Lavender!10] (60, -25) rectangle (100, -45);
\end{tikzpicture}}}
& \phantom{3}\$25,951 \downarrowbold\\
\end{tabular}
}
&
\multicolumn{2}{l|}{
\begin{tabular}{p{30mm}r}
\$106,542 \hspace{2mm}
\resizebox{.08\textwidth}{!}{
\raisebox{-0.3\height}{
\begin{tikzpicture}
\draw[fill=Yellow!70] (0,-25) rectangle (58, -45);
\draw[fill=Yellow!10] (58, -25) rectangle (100, -45);
\end{tikzpicture}}}
& \$55,485 \uparrowbold
\\
\phantom{3}\$96,609 \hspace{2mm}
\resizebox{.08\textwidth}{!}{
\raisebox{-0.3\height}{
\begin{tikzpicture}
\draw[fill=Lavender!70] (0,-25) rectangle (53, -45);
\draw[fill=Lavender!10] (53, -25) rectangle (100, -45);
\end{tikzpicture}}}
& \$14,854 \downarrowbold\\
\end{tabular}
}
& 0.059\\
\hline
% Row 5
\begin{tabular}{@{}p{80mm}}
For \colorbox{Orange!30}{\strut{White heterosexual males whose parents attended}} \colorbox{Orange!30}{\strut{secondary school}}, TC growth is more influenced by \colorbox{Cyan!30}{\strut{the desire become manager}} for \colorbox{Yellow!70}{\strut{analysts}} compared to \colorbox{Lavender!70}{\strut{back-end developers}}.
\end{tabular}
&
10.77\%&
\multicolumn{2}{l|}{
\begin{tabular}{p{30mm}r}
\$125,581 \hspace{2mm}
\resizebox{.08\textwidth}{!}{
\raisebox{-0.3\height}{
\begin{tikzpicture}
\draw[fill=Yellow!70] (0,-25) rectangle (69, -45);
\draw[fill=Yellow!10] (69, -25) rectangle (100, -45);
\end{tikzpicture}}}
& \$147,514 \uparrowbold
\\
\$109,837 \hspace{2mm}
\resizebox{.08\textwidth}{!}{
\raisebox{-0.3\height}{
\begin{tikzpicture}
\draw[fill=Lavender!70] (0,-25) rectangle (60, -45);
\draw[fill=Lavender!10] (60, -25) rectangle (100, -45);
\end{tikzpicture}}}
& \$32,445 \uparrowbold\\
\end{tabular}
}
&
\multicolumn{2}{c|}{not statistically significant}
& 0.057\\
\hline
\end{tabular}}
\caption{Top-k results for the SO dataset.}
\label{tab:example_so_topk}
\end{table*}

\begin{table*}[t]
\centering
\resizebox{1\textwidth}{!}{
\renewcommand{\arraystretch}{1.1}
\begin{tabular}[t]{|p{80mm}|c|>{\centering}m{25mm}c|@{}>{\centering}m{25mm}c|c|}
\hline
\multirow{3}{*}{\textbf{Disparity Explanation}} & 
\multirow{3}{*}{\textbf{Support}} & 
\multicolumn{4}{c|}{\textbf{Likelihood of having a health insurance}} &
\multirow{3}{*}{\textbf{$\Delta$}} \\
\cline{3-6}
& &
\multicolumn{2}{c|}{\textbf{Subpopulation}} &
\multicolumn{2}{c|}{\textbf{Global (XInsight)}}& \\
\cline{3-6}
&& 
\textbf{Average} & \textbf{CATE} & 
\textbf{Average} & \textbf{CATE} &\\
\hline
\hline
% Row 1
\begin{tabular}{@{}p{80mm}}
For \colorbox{Orange!30}{\strut{natives from the Southern region who were born in USA}}, the likelihood of having health insurance increases when they \colorbox{Cyan!30}{\strut{didn't report absence from work and didn't attend school in}} \colorbox{Cyan!30}{\strut{the last 3 months}} for \colorbox{Yellow!70}{\strut{manual labor occupations}}, whereas it decreases for \colorbox{Lavender!70}{\strut{all occupations}}.
\end{tabular} &
25.09\% &
\multicolumn{2}{l|}{
\begin{tabular}{p{25mm}r}
\drawbar{72.77}{9.58}{\%}{Yellow}{\uparrowbold}\\
\drawbar{81.24}{4.22}{\%}{Lavender}{\downarrowbold}\\
\end{tabular}
}
&
\multicolumn{2}{l|}{
\begin{tabular}{p{25mm}r}
\drawbar{78.61}{5.47}{\%}{Yellow}{\uparrowbold}\\
\drawbar{91.58}{0.50}{\%}{Lavender}{\downarrowbold}\\
\end{tabular}
} & 0.138 \\
\hline
% Row 2
\begin{tabular}{@{}p{80mm}}
Among \colorbox{Orange!30}{\strut{individuals from the Southern region who were born}} \colorbox{Orange!30}{\strut{in USA}}, the likelihood of having health insurance increases when they \colorbox{Cyan!30}{\strut{didn't report absence from work and didn't attend school in}} \colorbox{Cyan!30}{\strut{the last 3 months}} for \colorbox{Yellow!70}{\strut{manual labor occupations}}, whereas it decreases for \colorbox{Lavender!70}{\strut{all occupations}}.
\end{tabular} &
25.09\% &
\multicolumn{2}{l|}{
\begin{tabular}{p{25mm}r}
\drawbar{72.77}{9.58}{\%}{Yellow}{\uparrowbold}\\
\drawbar{81.24}{4.22}{\%}{Lavender}{\downarrowbold}\\
\end{tabular}
}
&
\multicolumn{2}{l|}{
\begin{tabular}{p{25mm}r}
\drawbar{78.61}{5.47}{\%}{Yellow}{\uparrowbold}\\
\drawbar{91.58}{0.50}{\%}{Lavender}{\downarrowbold}\\
\end{tabular}
}
 & 0.138 \\
\hline
% Row 3
\begin{tabular}{@{}p{80mm}}
For \colorbox{Orange!30}{\strut{White individuals from the Southern region who speak}} \colorbox{Orange!30}{\strut{Spanish}}, the likelihood of having health insurance decreases when they \colorbox{Cyan!30}{\strut{have no personal earnings}} for \colorbox{Yellow!70}{\strut{manual labor occupations}}, whereas it increases for \colorbox{Lavender!70}{\strut{all occupations}}.
\end{tabular}
&
8.18\%&
\multicolumn{2}{l|}{
\begin{tabular}{p{25mm}r}
\drawbar{52.12}{7.74}{\%}{Yellow}{\downarrowbold}\\
\drawbar{61.03}{4.87}{\%}{Lavender}{\uparrowbold}\\
\end{tabular}
}
&
\multicolumn{2}{l|}{
\begin{tabular}{p{25mm}r}
\drawbar{78.61}{2.65}{\%}{Yellow}{\downarrowbold}\\
\drawbar{91.58}{1.98}{\%}{Lavender}{\uparrowbold}\\
\end{tabular}
}
& 0.126\\
\hline
% Row 4
\begin{tabular}{@{}p{80mm}}
For \colorbox{Orange!30}{\strut{individuals from the Southern region who speak Spanish}}, the likelihood of having health insurance decreases when they \colorbox{Cyan!30}{\strut{have no personal earnings}} for \colorbox{Yellow!70}{\strut{manual labor occupations}}, whereas it increases for \colorbox{Lavender!70}{\strut{all occupations}}.
\end{tabular} &
10.30\% &
\multicolumn{2}{l|}{
\begin{tabular}{p{25mm}r}
\drawbar{51.48}{6.63}{\%}{Yellow}{\downarrowbold}\\
\drawbar{59.88}{3.87}{\%}{Lavender}{\uparrowbold}\\
\end{tabular}
}
&
\multicolumn{2}{l|}{
\begin{tabular}{p{25mm}r}
\drawbar{78.61}{2.65}{\%}{Yellow}{\downarrowbold}\\
\drawbar{91.58}{1.98}{\%}{Lavender}{\uparrowbold}\\
\end{tabular}
}
 & 0.105 \\
\hline
% Row 5
\begin{tabular}{@{}p{80mm}}
For \colorbox{Orange!30}{\strut{White individuals from Texas, Southern region who were}} \colorbox{Orange!30}{\strut{born in USA}}, the likelihood of having health insurance decreases when they \colorbox{Cyan!30}{\strut{have no personal earnings}} for \colorbox{Yellow!70}{\strut{manual labor occupations}}, whereas it increases for \colorbox{Lavender!70}{\strut{all occupations}}.
\end{tabular} &
12.21\% &
\multicolumn{2}{l|}{
\begin{tabular}{p{25mm}r}
\drawbar{71.86}{5.20}{\%}{Yellow}{\downarrowbold}\\
\drawbar{80.52}{2.83}{\%}{Lavender}{\uparrowbold}\\
\end{tabular}
}
&
\multicolumn{2}{l|}{
\begin{tabular}{p{25mm}r}
\drawbar{78.61}{2.65}{\%}{Yellow}{\downarrowbold}\\
\drawbar{91.58}{1.98}{\%}{Lavender}{\uparrowbold}\\
\end{tabular}
} & 0.080 \\
\hline
\end{tabular}}
\caption{Top-k  results for the ACS dataset.}
\label{tab:example_acs_topk}
\end{table*}

\begin{table*}[t]
\centering
\resizebox{1\textwidth}{!}{
\renewcommand{\arraystretch}{1.1}
\begin{tabular}[t]{|p{80mm}|c|>{\centering}m{25mm}c|@{}>{\centering}m{25mm}c|c|}
\hline
\multirow{3}{*}{\textbf{Disparity Explanation}} & 
\multirow{3}{*}{\textbf{Support}} & 
\multicolumn{4}{c|}{\textbf{Likelihood of feeling nervous frequently}} &
\multirow{3}{*}{\textbf{$\Delta$}} \\
\cline{3-6}
& &
\multicolumn{2}{c|}{\textbf{Subpopulation}} &
\multicolumn{2}{c|}{\textbf{Global (XInsight)}}& \\
\cline{3-6}
&& 
\textbf{Average} & \textbf{CATE} & 
\textbf{Average} & \textbf{CATE} &\\
\hline
\hline
% Row 1
\begin{tabular}{@{}p{80mm}}
For individuals \colorbox{Orange!30}{\strut{who never married, are from South region, }} \colorbox{Orange!30}{\strut{don't have doctor's recommendation to exercise, and don't}} \colorbox{Orange!30}{\strut{have Diabetes}}, the likelihood of feeling nervous frequently decreases less for \colorbox{Yellow!70}{\strut{males}} if they \colorbox{Cyan!30}{\strut{don't smoke currently}} compared to \colorbox{Lavender!70}{\strut{non-males}}.
\end{tabular} &
5.78\% &
\multicolumn{2}{l|}{
\begin{tabular}{p{25mm}r}
\drawbar{46.08}{13.06}{\%}{Yellow}{\downarrowbold}\\
\drawbar{41.71}{20.73}{\%}{Lavender}{\downarrowbold}\\
\end{tabular}
}
&
\multicolumn{2}{l|}{
\begin{tabular}{p{25mm}r}
\drawbar{37.58}{3.39}{\%}{Yellow}{\downarrowbold}\\
\drawbar{45.10}{3.94}{\%}{Lavender}{\downarrowbold}\\
\end{tabular}
} & 0.076 \\
\hline
% Row 2
\begin{tabular}{@{}p{80mm}}
Among \colorbox{Orange!30}{\strut{individuals who never married, are from South region,}} \colorbox{Orange!30}{\strut{don't have doctor's recommendation to exercise, and don't}} \colorbox{Orange!30}{\strut{ have Diabetes or Asthma}}, the likelihood of feeling nervous frequently decreases less for \colorbox{Yellow!70}{\strut{males}} than \colorbox{Lavender!70}{\strut{non-males}} when they \colorbox{Cyan!30}{\strut{don't smoke currently}}.
\end{tabular} &
5.3\% &
\multicolumn{2}{l|}{
\begin{tabular}{p{25mm}r}
\drawbar{46.52}{13.12}{\%}{Yellow}{\downarrowbold}\\
\drawbar{41.33}{19.36}{\%}{Lavender}{\downarrowbold}\\
\end{tabular}
}
&
\multicolumn{2}{l|}{
\begin{tabular}{p{25mm}r}
\drawbar{37.58}{3.39}{\%}{Yellow}{\downarrowbold}\\
\drawbar{45.10}{3.94}{\%}{Lavender}{\downarrowbold}\\
\end{tabular}
}
& 0.062 \\
\hline
% Row 3
\begin{tabular}{@{}p{80mm}}
For individuals \colorbox{Orange!30}{\strut{who never married, are from South region,}} \colorbox{Orange!30}{\strut{don't have doctor's recommendation to exercise}}, the likelihood of feeling nervous frequently decreases less for \colorbox{Yellow!70}{\strut{males}} if they \colorbox{Cyan!30}{\strut{don't smoke currently}} compared to \colorbox{Lavender!70}{\strut{non-males}}.
\end{tabular}
&
5.89\%&
\multicolumn{2}{l|}{
\begin{tabular}{p{25mm}r}
\drawbar{46.00}{12.56}{\%}{Yellow}{\downarrowbold}\\
\drawbar{41.24}{18.40}{\%}{Lavender}{\downarrowbold}\\
\end{tabular}
}
&
\multicolumn{2}{l|}{
\begin{tabular}{p{25mm}r}
\drawbar{37.58}{3.39}{\%}{Yellow}{\downarrowbold}\\
\drawbar{45.10}{3.94}{\%}{Lavender}{\downarrowbold}\\
\end{tabular}
}
& 0.058\\
\hline
% Row 4
\begin{tabular}{@{}p{80mm}}
For individuals \colorbox{Orange!30}{\strut{who never married, are from South region,}} \colorbox{Orange!30}{\strut{don't have doctor's recommendation to exercise, and don't}} \colorbox{Orange!30}{\strut{ have Asthma}}, the likelihood of feeling nervous frequently decreases less for \colorbox{Yellow!70}{\strut{males}} who \colorbox{Cyan!30}{\strut{do not currently smoke}} compared to \colorbox{Lavender!70}{\strut{non-males}}.
\end{tabular} &
5.38\% &
\multicolumn{2}{l|}{
\begin{tabular}{p{25mm}r}
\drawbar{46.44}{12.66}{\%}{Yellow}{\downarrowbold}\\
\drawbar{40.89}{18.38}{\%}{Lavender}{\downarrowbold}\\
\end{tabular}
}
&
\multicolumn{2}{l|}{
\begin{tabular}{p{25mm}r}
\drawbar{37.58}{3.39}{\%}{Yellow}{\downarrowbold}\\
\drawbar{45.10}{3.94}{\%}{Lavender}{\downarrowbold}\\
\end{tabular}
}
& 0.057 \\
\hline
% Row 5
\begin{tabular}{@{}p{80mm}}
For \colorbox{Orange!30}{\strut{White individuals who never married, aged below 29 years,}} \colorbox{Orange!30}{\strut{and don't have Asthma or Diabetes}}, the likelihood of feeling nervous frequently decreases less for \colorbox{Yellow!70}{\strut{males}} who \colorbox{Cyan!30}{\strut{haven't health insurance}} compared to \colorbox{Lavender!70}{\strut{non-males}}.
\end{tabular} &
8.41\% &
\multicolumn{2}{l|}{
\begin{tabular}{p{25mm}r}
\drawbar{49.21}{14.16}{\%}{Yellow}{\downarrowbold}\\
\drawbar{48.95}{18.48}{\%}{Lavender}{\downarrowbold}\\
\end{tabular}
}
&
\multicolumn{2}{l|}{
\begin{tabular}{p{25mm}r}
\drawbar{37.58}{6.77}{\%}{Yellow}{\downarrowbold}\\
\drawbar{45.10}{4.93}{\%}{Lavender}{\downarrowbold}\\
\end{tabular}}  
& 0.043 \\
\hline
\end{tabular}}
\caption{Top-k  results for the MEPS dataset.}
\label{tab:example_meps_topk}
\end{table*}

\begin{table*}[t]
\centering
\resizebox{1\textwidth}{!}{
\renewcommand{\arraystretch}{1.1}
\begin{tabular}[t]{|p{80mm}|c|>{\centering}m{30mm}c|@{}>{\centering}m{30mm}c|c|}
\hline
\multirow{3}{*}{\textbf{Disparity Explanation}} & 
\multirow{3}{*}{\textbf{Support}} & 
\multicolumn{4}{c|}{\textbf{Total Compensation (TC)}} &
\multirow{3}{*}{\textbf{$\Delta$}} \\
\cline{3-6}
& &
\multicolumn{2}{c|}{\textbf{Subpopulation}} &
\multicolumn{2}{c|}{\textbf{Global (XInsight)}}& \\
\cline{3-6}
&& 
\textbf{Average} & \textbf{CATE} & 
\textbf{Average} & \textbf{CATE} &\\
\hline
\hline
% Row 1
\begin{tabular}{@{}p{80mm}}
For \colorbox{Orange!30}{\strut{White males aged between 18-24}}, TC growth is more influenced by \colorbox{Cyan!30}{\strut{having coding as a hobby and learned in undergrad major}} \colorbox{Cyan!30}{\strut{ computer science}} for \colorbox{Yellow!70}{\strut{analysts}} compared to \colorbox{Lavender!70}{\strut{back-end developers}}.
\end{tabular}
&
10.64\%&
\multicolumn{2}{l|}{
\begin{tabular}{p{30mm}r}
\phantom{3}\$90,909 \hspace{2mm}
\resizebox{.08\textwidth}{!}{
\raisebox{-0.3\height}{
\begin{tikzpicture}
\draw[fill=Yellow!70] (0,-25) rectangle (50, -45);
\draw[fill=Yellow!10] (50, -25) rectangle (100, -45);
\end{tikzpicture}}}
& \$219,469 \uparrowbold
\\
\phantom{3}\$67,754 \hspace{2mm}
\resizebox{.08\textwidth}{!}{
\raisebox{-0.3\height}{
\begin{tikzpicture}
\draw[fill=Lavender!70] (0,-25) rectangle (37, -45);
\draw[fill=Lavender!10] (37, -25) rectangle (100, -45);
\end{tikzpicture}}}
& \$55,233 \uparrowbold\\
\end{tabular}
}
&
\multicolumn{2}{c|}{not statistically significant}
& 0.082\\
\hline
% Row 2
\begin{tabular}{@{}p{80mm}}
For \colorbox{Orange!30}{\strut{heterosexual individuals whose parents attended to}} \colorbox{Orange!30}{\strut{ secondary school}}, TC growth is more influenced by \colorbox{Cyan!30}{\strut{the desire become manager and not being a student}} for \colorbox{Yellow!70}{\strut{analysts}} compared to \colorbox{Lavender!70}{\strut{back-end developers}}.
\end{tabular}
&
14.41\%&
\multicolumn{2}{l|}{
\begin{tabular}{p{30mm}r}
\$112,896 \hspace{2mm}
\resizebox{.08\textwidth}{!}{
\raisebox{-0.3\height}{
\begin{tikzpicture}
\draw[fill=Yellow!70] (0,-25) rectangle (64, -45);
\draw[fill=Yellow!10] (64, -25) rectangle (100, -45);
\end{tikzpicture}}}
& \$176,698 \uparrowbold
\\
\phantom{3}\$99,126 \hspace{2mm}
\resizebox{.08\textwidth}{!}{
\raisebox{-0.3\height}{
\begin{tikzpicture}
\draw[fill=Lavender!70] (0,-25) rectangle (58, -45);
\draw[fill=Lavender!10] (58, -25) rectangle (100, -45);
\end{tikzpicture}}}
& \$53,814 \uparrowbold\\
\end{tabular}
}
&
\multicolumn{2}{c|}{not statistically significant}
& 0.061\\
\hline
% Row 3
\begin{tabular}{@{}p{80mm}}
For \colorbox{Orange!30}{\strut{White individuals aged between 25-34}}, TC growth increases by \colorbox{Cyan!30}{\strut{having 3-5 years of professional coding experience}} \colorbox{Cyan!30}{\strut{and spending over 12 hours on computer daily}} for \colorbox{Yellow!70}{\strut{analysts}} whereas it decreases for \colorbox{Lavender!70}{\strut{back-end developers}}.
\end{tabular}
&
34.69\%&
\multicolumn{2}{l|}{
\begin{tabular}{p{30mm}r}
\$115,777 \hspace{2mm}
\resizebox{.08\textwidth}{!}{
\raisebox{-0.3\height}{
\begin{tikzpicture}
\draw[fill=Yellow!70] (0,-25) rectangle (68, -45);
\draw[fill=Yellow!10] (68, -25) rectangle (100, -45);
\end{tikzpicture}}}
& \phantom{3}\$93,826 \uparrowbold
\\
\$105,988 \hspace{2mm}
\resizebox{.08\textwidth}{!}{
\raisebox{-0.3\height}{
\begin{tikzpicture}
\draw[fill=Lavender!70] (0,-25) rectangle (60, -45);
\draw[fill=Lavender!10] (60, -25) rectangle (100, -45);
\end{tikzpicture}}}
& \phantom{3}\$25,951 \downarrowbold\\
\end{tabular}
}
&
\multicolumn{2}{l|}{
\begin{tabular}{p{30mm}r}
\$106,542 \hspace{2mm}
\resizebox{.08\textwidth}{!}{
\raisebox{-0.3\height}{
\begin{tikzpicture}
\draw[fill=Yellow!70] (0,-25) rectangle (58, -45);
\draw[fill=Yellow!10] (58, -25) rectangle (100, -45);
\end{tikzpicture}}}
& \$55,485 \uparrowbold
\\
\phantom{3}\$96,609 \hspace{2mm}
\resizebox{.08\textwidth}{!}{
\raisebox{-0.3\height}{
\begin{tikzpicture}
\draw[fill=Lavender!70] (0,-25) rectangle (53, -45);
\draw[fill=Lavender!10] (53, -25) rectangle (100, -45);
\end{tikzpicture}}}
& \$14,854 \downarrowbold\\
\end{tabular}
}
& 0.059\\
\hline
% Row 4
\begin{tabular}{@{}p{80mm}}
For \colorbox{Orange!30}{\strut{White individuals}}, TC increases by \colorbox{Cyan!30}{\strut{having 3-5 professional coding years and spending over 12}} \colorbox{Cyan!30}{\strut{hours on computer daily}} for \colorbox{Yellow!70}{\strut{analysts}} whereas it decreases for \colorbox{Lavender!70}{\strut{back-end developers}}.
\end{tabular}
&
67.49\%&
\multicolumn{2}{l|}{
\begin{tabular}{p{30mm}r}
\$122,765 \hspace{2mm}
\resizebox{.08\textwidth}{!}{
\raisebox{-0.3\height}{
\begin{tikzpicture}
\draw[fill=Yellow!70] (0,-25) rectangle (67, -45);
\draw[fill=Yellow!10] (67, -25) rectangle (100, -45);
\end{tikzpicture}}}
& \phantom{3}\$63,311 \uparrowbold
\\
\$108,953 \hspace{2mm}
\resizebox{.08\textwidth}{!}{
\raisebox{-0.3\height}{
\begin{tikzpicture}
\draw[fill=Lavender!70] (0,-25) rectangle (60, -45);
\draw[fill=Lavender!10] (60, -25) rectangle (100, -45);
\end{tikzpicture}}}
& \phantom{3}\$23,068 \downarrowbold\\
\end{tabular}
}
&
\multicolumn{2}{l|}{
\begin{tabular}{p{30mm}r}
\$106,542 \hspace{2mm}
\resizebox{.08\textwidth}{!}{
\raisebox{-0.3\height}{
\begin{tikzpicture}
\draw[fill=Yellow!70] (0,-25) rectangle (58, -45);
\draw[fill=Yellow!10] (58, -25) rectangle (100, -45);
\end{tikzpicture}}}
& \$55,485 \uparrowbold
\\
\phantom{3}\$96,609 \hspace{2mm}
\resizebox{.08\textwidth}{!}{
\raisebox{-0.3\height}{
\begin{tikzpicture}
\draw[fill=Lavender!70] (0,-25) rectangle (53, -45);
\draw[fill=Lavender!10] (53, -25) rectangle (100, -45);
\end{tikzpicture}}}
& \$14,854 \downarrowbold\\
\end{tabular}
}
& 0.043\\
\hline
% Row 5
\begin{tabular}{@{}p{80mm}}
	For \colorbox{Orange!30}{\strut{males between the age 25-34 whose parents hold a}} \colorbox{Orange!30}{\strut{bachelor's degree}}, TC growth is more influenced by \colorbox{Cyan!30}{\strut{working in a company size between 100 - 499 workers}} for \colorbox{Yellow!70}{\strut{analysts}} compared to \colorbox{Lavender!70}{\strut{back-end developers}}.
\end{tabular}
&
13.21\%&
\multicolumn{2}{l|}{
\begin{tabular}{p{30mm}r}
\$105,694 \hspace{2mm}
\resizebox{.08\textwidth}{!}{
\raisebox{-0.3\height}{
\begin{tikzpicture}
\draw[fill=Yellow!70] (0,-25) rectangle (58, -45);
\draw[fill=Yellow!10] (58, -25) rectangle (100, -45);
\end{tikzpicture}}}
& \phantom{3}\$70,069 \uparrowbold
\\
\phantom{3}\$96,085 \hspace{2mm}
\resizebox{.08\textwidth}{!}{
\raisebox{-0.3\height}{
\begin{tikzpicture}
\draw[fill=Lavender!70] (0,-25) rectangle (53, -45);
\draw[fill=Lavender!10] (53, -25) rectangle (100, -45);
\end{tikzpicture}}}
& \phantom{3}\$19,807 \uparrowbold\\
\end{tabular}
}
&
\multicolumn{2}{c|}{not statistically significant}
& 0.025\\
\hline
\end{tabular}}
\caption{Brute-Force results for the SO dataset.}
\label{tab:example_so_bf}
\end{table*}

\begin{table*}[t]
\centering
\resizebox{1\textwidth}{!}{
\renewcommand{\arraystretch}{1.1}
\begin{tabular}[t]{|p{80mm}|c|>{\centering}m{25mm}c|@{}>{\centering}m{25mm}c|c|}
\hline
\multirow{3}{*}{\textbf{Disparity Explanation}} & 
\multirow{3}{*}{\textbf{Support}} & 
\multicolumn{4}{c|}{\textbf{Likelihood of having a health insurance}} &
\multirow{3}{*}{\textbf{$\Delta$}} \\
\cline{3-6}
& &
\multicolumn{2}{c|}{\textbf{Subpopulation}} &
\multicolumn{2}{c|}{\textbf{Global (XInsight)}}& \\
\cline{3-6}
&& 
\textbf{Average} & \textbf{CATE} & 
\textbf{Average} & \textbf{CATE} &\\
\hline
\hline
% Row 1
\begin{tabular}{@{}p{80mm}}
For \colorbox{Orange!30}{\strut{natives from the Southern region who were born in USA}}, the likelihood of having health insurance increases when they \colorbox{Cyan!30}{\strut{didn't report absence from work and didn't attend school in}} \colorbox{Cyan!30}{\strut{the last 3 months}} for \colorbox{Yellow!70}{\strut{manual labor occupations}}, whereas it decreases for \colorbox{Lavender!70}{\strut{all occupations}}.
\end{tabular} &
25.09\% &
\multicolumn{2}{l|}{
\begin{tabular}{p{25mm}r}
\drawbar{72.77}{9.58}{\%}{Yellow}{\uparrowbold}\\
\drawbar{81.24}{4.22}{\%}{Lavender}{\downarrowbold}\\
\end{tabular}
}
&
\multicolumn{2}{l|}{
\begin{tabular}{p{25mm}r}
\drawbar{78.61}{5.47}{\%}{Yellow}{\uparrowbold}\\
\drawbar{91.58}{0.50}{\%}{Lavender}{\downarrowbold}\\
\end{tabular}
} & 0.138 \\
\hline
% Row 2
\begin{tabular}{@{}p{80mm}}
For \colorbox{Orange!30}{\strut{White individuals from the Southern region who speak}} \colorbox{Orange!30}{\strut{Spanish}}, the likelihood of having health insurance decreases when they \colorbox{Cyan!30}{\strut{have no personal earnings}} for \colorbox{Yellow!70}{\strut{manual labor occupations}}, whereas it increases for \colorbox{Lavender!70}{\strut{all occupations}}.
\end{tabular}
&
8.18\%&
\multicolumn{2}{l|}{
\begin{tabular}{p{25mm}r}
\drawbar{52.12}{7.74}{\%}{Yellow}{\downarrowbold}\\
\drawbar{61.03}{4.87}{\%}{Lavender}{\uparrowbold}\\
\end{tabular}
}
&
\multicolumn{2}{l|}{
\begin{tabular}{p{25mm}r}
\drawbar{78.61}{2.65}{\%}{Yellow}{\downarrowbold}\\
\drawbar{91.58}{1.98}{\%}{Lavender}{\uparrowbold}\\
\end{tabular}
}
& 0.126\%\\
\hline
% Row 3
\begin{tabular}{@{}p{80mm}}
For \colorbox{Orange!30}{\strut{White individuals from Texas or Southern region, who}} \colorbox{Orange!30}{\strut{were born in USA}}, the likelihood of having health insurance decreases when they \colorbox{Cyan!30}{\strut{have no personal earnings}} for \colorbox{Yellow!70}{\strut{manual labor occupations}}, whereas it increases for \colorbox{Lavender!70}{\strut{all occupations}}.
\end{tabular} &
12.21\% &
\multicolumn{2}{l|}{
\begin{tabular}{p{25mm}r}
\drawbar{71.86}{5.20}{\%}{Yellow}{\downarrowbold}\\
\drawbar{80.52}{2.8}{\%}{Lavender}{\uparrowbold}\\
\end{tabular}
}
&
\multicolumn{2}{l|}{
\begin{tabular}{p{25mm}r}
\drawbar{78.61}{2.65}{\%}{Yellow}{\downarrowbold}\\
\drawbar{91.58}{1.98}{\%}{Lavender}{\uparrowbold}\\
\end{tabular}
} & 0.080 \\
\hline
% Row 4
\begin{tabular}{@{}p{80mm}}
Among \colorbox{Orange!30}{\strut{White males from the Southern region}}, the likelihood of having health insurance decreases when they \colorbox{Cyan!30}{\strut{have no personal earnings}} for \colorbox{Yellow!70}{\strut{manual labor occupations}}, whereas it increases for \colorbox{Lavender!70}{\strut{all occupations}}.
\end{tabular} &
18.00\% &
\multicolumn{2}{l|}{
\begin{tabular}{p{25mm}r}
\drawbar{67.56}{4.19}{\%}{Yellow}{\downarrowbold}\\
\drawbar{74.20}{2.68}{\%}{Lavender}{\uparrowbold}\\
\end{tabular}
}
&
\multicolumn{2}{l|}{
\begin{tabular}{p{25mm}r}
\drawbar{78.61}{2.65}{\%}{Yellow}{\downarrowbold}\\
\drawbar{91.58}{1.98}{\%}{Lavender}{\uparrowbold}\\
\end{tabular}
}
& 0.068 \\
\hline
% Row 5
\begin{tabular}{@{}p{80mm}}
Among \colorbox{Orange!30}{\strut{natives who were born in USA}}, the likelihood of having health insurance increases when they \colorbox{Cyan!30}{\strut{didn't report absence from work and didn't attend school in}} \colorbox{Cyan!30}{\strut{the last 3 months}} for \colorbox{Yellow!70}{\strut{manual labor occupations}}, whereas it decreases for \colorbox{Lavender!70}{\strut{all occupations}}.
\end{tabular} &
75.67\% &
\multicolumn{2}{l|}{
\begin{tabular}{p{25mm}r}
\drawbar{84.11}{4.95}{\%}{Yellow}{\uparrowbold}\\
\drawbar{88.97}{1.91}{\%}{Lavender}{\downarrowbold}\\
\end{tabular}
}
&
\multicolumn{2}{l|}{
\begin{tabular}{p{25mm}r}
\drawbar{78.61}{5.47}{\%}{Yellow}{\uparrowbold}\\
\drawbar{91.58}{0.50}{\%}{Lavender}{\downarrowbold}\\
\end{tabular}
} & 0.068 \\
\hline
\end{tabular}}
\caption{Brute-Force results for the ACS dataset.}
\label{tab:example_acs_bf}
\end{table*}

\begin{table*}[t]
\centering
\resizebox{1\textwidth}{!}{
\renewcommand{\arraystretch}{1.1}
\begin{tabular}[t]{|p{80mm}|c|>{\centering}m{25mm}c|@{}>{\centering}m{25mm}c|c|}
\hline
\multirow{3}{*}{\textbf{Disparity Explanation}} & 
\multirow{3}{*}{\textbf{Support}} & 
\multicolumn{4}{c|}{\textbf{Likelihood of feeling nervous frequently}} &
\multirow{3}{*}{\textbf{$\Delta$}} \\
\cline{3-6}
& &
\multicolumn{2}{c|}{\textbf{Subpopulation}} &
\multicolumn{2}{c|}{\textbf{Global (XInsight)}}& \\
\cline{3-6}
&& 
\textbf{Average} & \textbf{CATE} & 
\textbf{Average} & \textbf{CATE} &\\
\hline
\hline
% Row 1
\begin{tabular}{@{}p{80mm}}
For individuals \colorbox{Orange!30}{\strut{who never married, are from the Southern}} \colorbox{Orange!30}{\strut{region, don't have a doctor's recommendation to exercise, and}} \colorbox{Orange!30}{\strut{ aren't diagnosed with Diabetes}}, the likelihood of feeling nervous frequently decreases less for \colorbox{Yellow!70}{\strut{males}} compared to \colorbox{Lavender!70}{\strut{non-males}} if they \colorbox{Cyan!30}{\strut{do not smoke currently}}.
\end{tabular} &
5.78\% &
\multicolumn{2}{l|}{
\begin{tabular}{p{25mm}r}
\drawbar{46.08}{13.06}{\%}{Yellow}{\downarrowbold}\\
\drawbar{41.71}{20.73}{\%}{Lavender}{\downarrowbold}\\
\end{tabular}
}
&
\multicolumn{2}{l|}{
\begin{tabular}{p{25mm}r}
\drawbar{37.58}{3.39}{\%}{Yellow}{\downarrowbold}\\
\drawbar{45.10}{3.94}{\%}{Lavender}{\downarrowbold}\\
\end{tabular}
} & 0.076 \\
\hline
% Row 2
\begin{tabular}{@{}p{80mm}}
Among \colorbox{Orange!30}{\strut{White individuals who never married, are under 29,}} \colorbox{Orange!30}{\strut{and don't have Asthma or Diabetes}}, the likelihood of feeling nervous frequently decreases less for \colorbox{Yellow!70}{\strut{males}} than \colorbox{Lavender!70}{\strut{non-males}} if they are \colorbox{Cyan!30}{\strut{uninsured for the whole year}}.
\end{tabular} &
8.41\% &
\multicolumn{2}{l|}{
\begin{tabular}{p{25mm}r}
\drawbar{49.20}{14.16}{\%}{Yellow}{\downarrowbold}\\
\drawbar{48.90}{18.48}{\%}{Lavender}{\downarrowbold}\\
\end{tabular}
}
&
\multicolumn{2}{l|}{
\begin{tabular}{p{25mm}r}
\drawbar{37.58}{7.02}{\%}{Yellow}{\downarrowbold}\\
\drawbar{45.10}{4.86}{\%}{Lavender}{\downarrowbold}\\
\end{tabular}
} & 0.043 \\
\hline
% Row 3
\begin{tabular}{@{}p{80mm}}
For individuals \colorbox{Orange!30}{\strut{who never married, don't have a}} \colorbox{Orange!30}{\strut{recommendation from the doctor to exercise, and were born in}} \colorbox{Orange!30}{\strut{ USA}}, the likelihood of feeling nervous frequently increases more for \colorbox{Yellow!70}{\strut{males}} compared to \colorbox{Lavender!70}{\strut{non-males}} if they \colorbox{Cyan!30}{\strut{have private insurance}}.
\end{tabular}
&
14.25\%&
\multicolumn{2}{l|}{
\begin{tabular}{p{25mm}r}
\drawbar{45.90}{10.84}{\%}{Yellow}{\uparrowbold}\\
\drawbar{45.70}{7.51}{\%}{Lavender}{\uparrowbold}\\
\end{tabular}
}
&
\multicolumn{2}{l|}{
\begin{tabular}{p{25mm}r}
\drawbar{37.58}{4.63}{\%}{Yellow}{\uparrowbold}\\
\drawbar{45.10}{5.76}{\%}{Lavender}{\uparrowbold}\\
\end{tabular}
}
& 0.033\\
\hline
% Row 4
\begin{tabular}{@{}p{80mm}}
For \colorbox{Orange!30}{\strut{white individuals who never married, don't have doctor's}} \colorbox{Orange!30}{\strut{recommendation to exercise, and aren't diagnosed with}} \colorbox{Orange!30}{\strut{ Asthma}}, the likelihood of feeling nervous frequently decreases less for \colorbox{Yellow!70}{\strut{males}} compared to \colorbox{Lavender!70}{\strut{non-males}} if they are \colorbox{Cyan!30}{\strut{uninsured the whole year}}.
\end{tabular} &
9.74\% &
\multicolumn{2}{l|}{
\begin{tabular}{p{25mm}r}
\drawbar{47.25}{12.52}{\%}{Yellow}{\downarrowbold}\\
\drawbar{46.64}{14.09}{\%}{Lavender}{\downarrowbold}\\
\end{tabular}
}
&
\multicolumn{2}{l|}{
\begin{tabular}{p{25mm}r}
\drawbar{37.58}{7.02}{\%}{Yellow}{\downarrowbold}\\
\drawbar{45.10}{4.86}{\%}{Lavender}{\downarrowbold}\\
\end{tabular}
} & 0.015 \\
\hline
% Row 5
\begin{tabular}{@{}p{80mm}}
For \colorbox{Orange!30}{\strut{individuals aged between 30--42, and aren't diagnosed with}} \colorbox{Orange!30}{\strut{Diabetes}}, the likelihood of feeling nervous frequently increases more for \colorbox{Yellow!70}{\strut{males}} compared to \colorbox{Lavender!70}{\strut{non-males}} if they \colorbox{Cyan!30}{\strut{have health insurance}}.
\end{tabular} &
19.64\% &
\multicolumn{2}{l|}{
\begin{tabular}{p{25mm}r}
\drawbar{43.93}{\phantom{3}7.47}{\%}{Yellow}{\uparrowbold}\\
\drawbar{43.57}{\phantom{3}6.47}{\%}{Lavender}{\uparrowbold}\\
\end{tabular}
}
&
\multicolumn{2}{l|}{
\begin{tabular}{p{25mm}r}
\drawbar{37.58}{2.34}{\%}{Yellow}{\uparrowbold}\\
\drawbar{45.10}{4.54}{\%}{Lavender}{\uparrowbold}\\
\end{tabular}
} & 0.010 \\
\hline
\end{tabular}
}
\caption{Brute-Force results for the MEPS dataset.}
\label{tab:example_meps_bf}
\end{table*}

\end{appendices}

\end{document}